\documentclass[11pt]{article}
\usepackage{graphicx}
\usepackage{epstopdf}
\usepackage{amssymb,bm}
\usepackage{bbm}
\usepackage{amsmath}
\usepackage{amsfonts}
\usepackage[margin=1.0in]{geometry}
\usepackage{cite}
\usepackage[shortlabels]{enumitem}
\usepackage{comment}
\usepackage{color}

\DeclareMathOperator{\E}{\mathbbmss{E}}

\DeclareMathOperator{\Hbb}{\mathbb{H}}
\newcommand*{\dotleq}{\mathrel{\dot{\leq}}}
\allowdisplaybreaks

\usepackage{amsthm}
\newtheorem{theorem}{Theorem}
\newtheorem*{theorem*}{Theorem}
\newtheorem{lemma}{Lemma}

\theoremstyle{definition}  %%% for non-italic theorem environments
\newtheorem{definition}{Definition}
\newtheorem{remark}{Remark}

\usepackage{chngcntr}
\counterwithout{theorem}{section}
\counterwithout{definition}{section}
\counterwithout{lemma}{section}
\counterwithout{remark}{section}
\counterwithout{assumption}{section}
\counterwithout{proposition}{section}
\counterwithout{corollary}{section}
\counterwithout{claim}{section}

\begin{document}
%%%
\title{Cellular Networks With Finite Precision CSIT: GDoF Optimality of Multi-Cell TIN and Extremal Gains of Multi-Cell Cooperation
\footnotetext{H. Joudeh is with the  Department of Electrical Engineering, Eindhoven University of Technology, 5600 MB Eindhoven, The
Netherlands (e-mail: h.joudeh@tue.nl).
%%%%
G. Caire is with the Faculty of Electrical Engineering and Computer Science, Technische Universit\"{a}t Berlin, 10587 Berlin, Germany (e-mail: caire@tu-berlin.de). 
%%%%%
This work was partially supported by the European Research Council (ERC) under the ERC Advanced Grant N. 789190 (CARENET). 
%%%%%
This work was presented in part at the 2020 IEEE International Symposium on Information Theory (ISIT) \cite{Joudeh2020a} and in part at the 2020 IEEE Global Communications Conference (GLOBECOM) \cite{Joudeh2020b}.}}
%%%%%%%%%%%%%%%%%%%%%%%%%%%%%%%%%%%%%%%%%%%%%%%%%%%%%%%%%%%%%%%%%%
\author{Hamdi~Joudeh and Giuseppe~Caire}
%%%%%%%%%%%%%%
\date{}
%%%%%%%%%%%%%%%%%%%%%%%%%%%%%%%%%%%%%%%%%%%%%%%%%%%%%%%%%%%%%%%%%%
\maketitle
%%%%
\begin{abstract}
%%%
We study the generalized degrees-of-freedom (GDoF) of cellular networks under finite precision channel state information at the transmitters (CSIT). We consider downlink settings modeled by the interfering broadcast channel (IBC) under no multi-cell cooperation, and the overloaded multiple-input-single-output broadcast channel (MISO-BC) under full multi-cell cooperation. We focus on three regimes of interest: the mc-TIN regime, where a scheme based on treating inter-cell interference as noise (mc-TIN) was shown to be GDoF optimal for the IBC; the mc-CTIN regime, where the GDoF region achievable by mc-TIN is convex without the need for time-sharing; and the mc-SLS regime which extends a previously identified regime, where a simple layered superposition (SLS) scheme is optimal for the 3-transmitter-3-user MISO-BC, to overloaded cellular-type networks with more users than transmitters. We first show that the optimality of mc-TIN for the IBC extends to the entire mc-CTIN regime when CSIT is limited to finite precision. The converse proof of this result relies on a new application of aligned images bounds. We then extend the IBC converse proof to the counterpart overloaded MISO-BC, obtained by enabling full transmitter cooperation. This, in turn, is utilized to show that a multi-cell variant of the SLS scheme is optimal in the mc-SLS regime under full multi-cell cooperation, albeit only for 2-cell networks. The overwhelming combinatorial complexity of the GDoF region stands in the way of extending this result to larger networks. Alternatively, we appeal to extremal network analysis, recently introduced by Chan et al., and study the GDoF gain of multi-cell cooperation over mc-TIN in the three regimes of interest. We show that this extremal GDoF gain is bounded by small constants in the mc-TIN and mc-CTIN regimes, yet scales logarithmically with the number of cells in the mc-SLS regime.
%%%
\end{abstract}
\newpage
%%%%%%%%%
\section{Introduction}
\label{sec:introduction}
%%%
For over a decade, generalized degrees-of-freedom (GDoF) studies have continued to contribute new insights into the fundamental limits of wireless communication networks \cite{Etkin2008,Jafar2010,Huang2012,Karmakar2012,Davoodi2017a,
Davoodi2017,Davoodi2018,Davoodi2019a}. 
%%%
The GDoF is defined by proportionally scaling all individual link capacities in a network to infinity, while normalizing the total network capacity by the same scaling factor. 
%%%
From an analytic standpoint, this scaling eliminates intricacies related to additive noise, unavoidable in exact capacity analysis, and focuses the attention on aspects related to interference management.
%%%%
This is accomplished while preserving the character of the original network, carried through to the GDoF limit
via the almost constant normalized link capacities (or channel strength parameters), rendering the GDoF a more comprehensive figure of merit compared to the more commonly employed  degrees-of-freedom (DoF) metric.
%%%
The validity of the GDoF metric as a good \emph{approximator} is confirmed by several constant gap capacity characterizations, obtained from their  GDoF counterparts (e.g. \cite{Etkin2008,Karmakar2013,Bresler2010}).
%%%

%%%
Amongst the various insights brought to light by GDoF studies, perhaps the most striking are those related to the fundamental role of channel state information at the transmitters (CSIT) in wireless interference networks.
%%%
Initial GDoF studies\footnote{This also includes early DoF results, which can be seen as a special case of GDoF results (see Remark \ref{rermak:GDoF_DoF}).} considered somewhat idealized settings, in which the availability of perfect CSIT is assumed; and gave 
rise to novel schemes that challenged conventional wisdom at the time, especially those based on interference alignment (IA) and the use of structured codes as opposed to random codes \cite{Cadambe2008,Bresler2010,Jafar2011}.
%%%%
In practical wireless networks, however, uncertainty in channel knowledge is inevitable due to fading and limited feedback resources \cite{Love2008}; and robustness against such uncertainty is of paramount importance to avoid potential catastrophic failures.
%%%%
This has recently motivated a surge of interest in \emph{robust} GDoF characterizations, obtained by limiting channel knowledge at the transmitters to finite precision.
%%%%

%%%%
While the study of wireless networks with finite precision CSIT is not new \cite{Lapidoth2005}, progress in this area was long hindered by the lack of tight information-theoretic outer bounds that match the best known inner bounds.
%%%%
This, however, has changed in recent years with a breakthrough by Davoodi and Jafar \cite{Davoodi2016} known as the aligned images (AI) approach: a combinatorial argument for bounding the number of codewords that align at one receiver and remain resolvable at another.
%%%%
This approach has been instrumental in deriving tight GDoF outer bounds under finite precision CSIT\footnote{While the AI approach was initially introduced to bound the DoF under partial CSIT \cite{Davoodi2016}, it is now well-understood that DoF bounds under partial CSIT translate to GDoF bounds under finite precision CSIT \cite{Davoodi2017,Davoodi2018}.} for a variety of interference and multi-antenna networks \cite{Davoodi2017a,Davoodi2017,Davoodi2018,Davoodi2019a}.
%%%%
A common theme emerging from these robust characterizations is that GDoF gains due to sophisticated schemes that rely on precise CSIT---as signal-space and signal-scale IA \cite{Cadambe2008,Bresler2010}, zero forcing (ZF) and dirty paper coding (DPC) \cite{Shamai2001,Caire2003,Weingarten2006}---are entirely lost under finite precision CSIT.
%%%%
The new insights into the fragility of those once-promising elegant schemes have brought back to the forefront simple and robust alternatives that require only coarse channel knowledge---as power control and treating interference as noise (TIN) \cite{Geng2015,Geng2016,Sun2016,Yi2016,Geng2015a,Gherekhloo2017,Yi2020,Gherekhloo2016,Joudeh2019a,Joudeh2019b},  rate-splitting and simple layered superposition (SLS) \cite{Joudeh2016,Piovano2016,Davoodi2019,Joudeh2020}---which, perhaps surprisingly, have turned out to be GDoF optimal in a number of settings. 
%%%%

%%%%
In the present paper, we extend the state-of-the-art in robust GDoF studies to cellular networks with finite precision CSIT.
%%%%
A basic cellular setting operating in the downlink mode is modeled by a Gaussian interfering broadcast channel (IBC), comprising $K$ mutually interfering Gaussian BCs with $L$ users each.
%%%%
When all transmitters are allowed to cooperate (i.e. full multi-cell cooperation), the setting becomes an overloaded multiple-input-single-output BC (MISO-BC), with $K$ transmit antennas and $KL$ users.
%%%%
We will refer to the above settings as $K \times KL$ networks, where the number of cells $K$ and the number of users per-cell $L$ may be arbitrary.
%%%%
A major challenge in GDoF studies of large networks, including cellular networks, is the richness of the underlying parameter space.
%%%%
GDoF characterizations depend on channel strength parameters which capture link strength levels in dB (or, equivalently, normalized capacities) between every transmitter and receiver in the network.
%%%
Channel strength parameters assume arbitrary non-negative real values in general,  giving rise to overwhelmingly many parameter regimes as the network  grows larger.
%%%
A one-size-fits-all approach is often insufficient for such problems---achievable schemes require separate tuning for different parameter regimes, and distinct matching outer bounds are often needed \cite{Etkin2008,Jafar2010,Huang2012,Karmakar2012,Davoodi2017a,
Davoodi2017,Davoodi2018,Davoodi2019a}. 
%%%
It is hence no surprise or coincidence that most GDoF characterizations are limited to either small or symmetric networks, described by very few parameters.
%%%
It is also often the case that insights derived from such settings do not directly generalize to large or asymmetric networks.
%%%

%%%
A promising approach to study the GDoF of large asymmetric networks, while circumventing the curse of dimensionality, has been to focus on special regimes of channel strength parameters, such as \emph{weak} interference regimes.
%%%
Prime examples are the fairly broad regimes in which simple schemes based on TIN turn out to be GDoF optimal, which have been identified for a variety of interference networks \cite{Geng2015,Geng2016,Sun2016,Yi2016,Geng2015a,
Gherekhloo2017,Yi2020,Gherekhloo2016,Joudeh2019a,Joudeh2019b}.
%%%%
This approach has also proven useful for more practical purposes---the inherent simplicity of TIN combined with its (approximate) information-theoretic optimality in some regimes have inspired the design of new high-performing resource allocation algorithms  for device-to-device (D2D) networks\cite{Naderializadeh2014,Yi2015,Yi2016}. 
%%%%
As far as cellular networks are concerned, it is often the case in practice that a user equipment (UE) connects to the base station (BS) with the strongest signal within its vicinity; and hence a direct link between a BS and an associated UE is typically no weaker than any of the inter-cell interference links to the same UE.
%%%%%
This intrinsic property of cellular networks, combined with the goal of making progress on GDoF problems that remain open in their generality, naturally motivate us to study weak \emph{inter-cell} interference regimes.
%%%%
Against this background, we focus our attention on two main questions in this paper.
%%%%

\emph{Q1. Robust GDoF Optimality of Multi-Cell TIN}:
The first question we set out to answer relates to the GDoF optimality of treating 
inter-cell interference as noise (multi-cell TIN, or mc-TIN) for the IBC with finite precision CSIT.
%%%%
In the mc-TIN scheme, single-cell transmissions are used---this includes power control, superposition coding and successive decoding in each cell---while all inter-cell interference is treated as additional Gaussian noise at the receivers \cite{Joudeh2019b}.
%%%%
This is arguably the \emph{simplest} multi-cell scheme for which GDoF optimality results can be show in fairly broad parameter regimes, and it is also especially suitable for weak inter-cell interference.
%%%%
Two relevant weak inter-cell interference regimes have been identified in \cite{Joudeh2019b}: a mc-TIN regime, in which mc-TIN achieves the entire GDoF region of the IBC; and a multi-cell convex TIN (mc-CTIN) regime, a strictly larger regime in which the GDoF region achieved through mc-TIN is a convex polyhedron without the need for time-sharing.\footnote{Similar results also hold in uplink networks modeled by the interfering multiple access channel (IMAC) \cite{Joudeh2019a}.}
%%%%
Note that the mc-TIN and mc-CTIN regimes respectively generalize the TIN regime of Geng \emph{et al.} \cite{Geng2015} and CTIN regime of Yi and Caire \cite{Yi2016}, identified for the $K \times K$ interference channel (IC), to $K \times KL$ cellular networks (see Section \ref{subsec:regimes}).
%%%%
Perfect CSIT is assumed in \cite{Joudeh2019b}, and the optimality of mc-TIN remains 
unexplored under finite precision CSIT.
%%%%

%%%%
\emph{Q2. Robust GDoF Gain of Multi-Cell Cooperation:}
The second question relates to understanding the extent to which multi-cell cooperation (mc-Co) is beneficial when CSIT is limited to finite precision, especially in weak inter-cell interference regimes where the simple non-cooperative mc-TIN scheme is known (or expected) to be quite powerful.
%%%%
While cooperation in cellular networks has been subject to intensive research for many years (see \cite{Gesbert2010,Simeone2012,Lozano2013} and references therein), the fundamental limits and benefits of robust  mc-Co under finite precision CSIT remain largely unexplored, apart from GDoF  results for $K \times K$ networks \cite{Davoodi2017,Davoodi2018,Davoodi2019,Chan2020}. 
%%%
These results, however, do not extend directly to overloaded $K \times KL$ networks, which better resemble cellular settings.

%%%%
In principle, the gain from mc-Co can be studied by comparing the GDoF of the $K \times KL$ IBC and the GDoF of the counterpart $K \times KL$ MISO-BC.
%%%
Nevertheless, such direct comparisons are currently infeasible in large asymmetric settings, mainly due to the difficulty of obtaining explicit GDoF characterizations---for instance, while it remains plausible that the simple layer superposition (SLS) scheme is GDoF optimal for the $K \times K$ MISO-BC in the so-called SLS regime, this has been shown only for $K \leq 3$ so far, as the problem becomes highly intractable for larger $K$  \cite{Davoodi2018,Davoodi2019}.
%%%
This hurdle has been recently circumvented through \emph{extremal network} analysis: a novel approach proposed by Chan \emph{et al.} \cite{Chan2020}, where the focus is shifted towards studying a class of networks in a regime of interest which maximize the gain of one scheme over another. 
%%%
This approach was successfully applied to $K \times K$ networks with finite precision CSIT in \cite{Chan2020} to characterize the maximum multiplicative GDoF gains from transmitter cooperation over TIN in three weak interference regimes of interest: the TIN, CTIN and SLS regimes. 
%%%%
The prospect of leveraging extremal network analysis to understand the gain of mc-Co over mc-TIN in $K \times KL$ cellular networks is an intriguing one, and it is yet to be explored.
%%%%
Next, we summarize the main findings and contributions of this work.
%%%%
\subsection{Overview of Results}
%%%%
In the first result of this paper (Theorem \ref{theorem:CTIN_optimality}), we show that when CSIT  is limited to finite precision in the IBC, the GDoF optimality of mc-TIN extends to the entire mc-CTIN regime.
%%%%
This stands in sharp contrast to the IBC with perfect CSIT, where IA-based schemes achieve strict GDoF gains over mc-TIN in some parts of the mc-CTIN regime.\footnote{This has been shown using a signal-scale IA scheme for fixed channels in \cite{Joudeh2019b}, and can also be shown for varying (generic) channels using the signal-space IA scheme in \cite{Suh2011}, as we will see in Section \ref{subsec:result_mc_CTIN}.}
%%%%
The converse proof of Theorem \ref{theorem:CTIN_optimality} 
relies on a new application of AI bounds to the $K \times KL$ IBC, 
which departs from previous applications that focus on the $K \times K$ interference channel (IC) and its cooperative  $K \times K$ MISO-BC counterpart \cite{Davoodi2017a,Davoodi2017,Davoodi2018,Chan2020,Davoodi2019}.
%%%%
In particular, the mixed BC-IC nature of the IBC, which captures both intra-cell and inter-cell interference in cellular networks, requires a careful selection  of auxiliary random variables to obtain a tight GDoF outer bound in this case 
(see Section \ref{sec:outer_bounds}).
%%%%

%%%%
Theorem \ref{theorem:CTIN_optimality} generalizes a previous result for the $K \times K$ IC, obtained by Chan \emph{et al.} in \cite[Th. 4.1]{Chan2020}, to the  $K \times KL$ IBC; 
and further reveals key insights into cellular networks which cannot be seen from studying the IC alone. 
%%%%
For instance, in the mc-CTIN regime, apart from the \emph{strongest} user with the highest signal-to-interference ratio (SIR) in each cell, all \emph{weaker} users turn out to be redundant in the sense that they do not contribute to increasing the overall GDoF of the IBC with finite precision CSIT.
%%%%
This again is in sharp contrast to the the IBC with perfect CSIT, where weaker users are still very useful for increasing the overall GDoF through IA.
%%%%
This point and other key insights are discussed in greater detail further on
 in Section \ref{subsec:result_mc_CTIN}.
%%%%

%%%%
Next, we consider the GDoF gain from mc-Co under finite precision CSIT, focusing on the mc-TIN and mc-CTIN regimes, and a strictly larger regime which we call the multi-cell SLS (mc-SLS) regime, which naturally extends the SLS regime of \cite{Davoodi2019} to $K \times KL$ settings.
%%%%
To this end, we start by deriving an AI-based GDoF outer bound for the $K \times KL$  MISO-BC with finite precision CSIT  (Theorem \ref{theorem:MISO_BC_outerbound}). 
%%%%
We further show that this outer bound is achievable in the $2$-cell case using a multi-cell variant of the SLS scheme, hence settling the GDoF region question for the $2 \times 2L$  MISO-BC with finite precision CSIT in the mc-SLS regime (Theorem \ref{theorem:2_cell_SLS}).
%%%%
Remarkably, the general structure of the SLS scheme greatly simplifies in $2 \times 2L$ cellular networks in the mc-SLS regime---instead of exponentially many encoded sub-messages, a linear number of sub-messages (in $L$) suffices.
%%%%
This reduction is key to the achievability proof---it enables an efficient elimination of the optimization variables used in describing the GDoF region achievable through the SLS scheme, from which we can match it to the outer bound of Theorem \ref{theorem:MISO_BC_outerbound} (see Section \ref{sec:2_cell_SLS}).
%%%%

%%%%
Beyond $2$-cell networks, characterizing the GDoF region of the $K \times KL$ MISO-BC is very challenging due to the explosion in the number of parameters and optimization variables, rendering a direct comparison between mc-Co and 
mc-TIN infeasible in general cellular networks.
%%%%
Inspired by \cite{Chan2020}, we circumvent  this complexity barrier by appealing to extremal network analysis and focusing on the extremal GDoF gain from mc-Co over mc-TIN.
%%%%
The cooperative GDoF outer bound in Theorem \ref{theorem:MISO_BC_outerbound} proves very useful for this purpose---it turns out to be tight for the classes of extremal networks, which maximize the GDoF gain of mc-Co over mc-TIN, in the regimes of interest.  
%%%%
We show that the extremal GDoF gain of mc-Co over mc-TIN is given by:  $3/2$ in the mc-TIN regime; $2 - 1 / K$ in the mc-CTIN regime; and scales as $\Theta\big(\log(K) \big)$ in the mc-SLS regime (the notation $\Theta( \cdot )$ is defined below).
%%%%
Interestingly, these gains do not depend on the number of users per-cell $L$, and they are exactly equal to 
their counterpart extremal gains in $K \times K$ networks \cite{Chan2020}.
%%%%
In other words, in the three weak interference regimes of interest, additional (weaker) users in each cell have no influence on extremal GDoF gains, governed by the underlying $K \times K$ networks of single-user cells.
%%%

%%%
Interestingly, in the past an intensive amount of research and industrial development work has been dedicated to schemes generally referred to as cooperative multi-cell processing (CoMP), which can be regarded as practical embodiments of the theoretical mc-Co paradigm studied here. After much work it was concluded that the gains offered by CoMP schemes over single-cell processing (i.e. mc-TIN) in typical cellular layouts, characterized by rather weak inter-cell interference, were somehow disappointing and not worth the extra complexity incurred by CoMP. Imperfect CSIT has been identified as an important factor in the limited CoMP gain. Viewed in this light, our extremal gain results can be seen as a theoretical justification to these empirical findings. 
%%%
\subsection{Notation}
\label{subsec:notation}
For positive integers $z_{1}$ and $z_{2}$ where $z_{1} \leq z_{2}$, the sets $\{1,2,\ldots,z_{1}\}$ and $\{z_{1},z_{1}+1,\ldots,z_{2}\}$ are denoted by
$\langle z_{1} \rangle$ and $\langle z_{1}:z_{2}\rangle$, respectively.
%%%
For any real number  $a \in \mathbb{R}$, we have $(a)^{+} = \max\{0,a\}$.
%%%
Bold symbols  denote tuples, e.g. $\mathbf{a} = (a_{1},\ldots,a_{Z})$ 
and $\bm{A} = (A_{1},\ldots,A_{Z})$, while 
calligraphic symbols denote sets, e.g. $\mathcal{A} = \{ a_{1},\ldots,a_{Z} \}$.
%%%
For any pair of sets $\mathcal{A},\mathcal{B} \subseteq \mathbb{R}^{K}$, their
Minkowski sum $\mathcal{A} \oplus \mathcal{B} $ is a set in $ \mathbb{R}^{K}$
defined as $\mathcal{A} \oplus \mathcal{B} \triangleq \left\{ \mathbf{a} + \mathbf{b} : \mathbf{a} \in \mathcal{A}, \mathbf{b} \in \mathcal{B}  \right\}$.
%%%
The indicator function with condition $\mathcal{A}$ is denoted by $\mathbbm{1}(\mathcal{A})$,
which is $1$ when $\mathcal{A}$ holds and  $0$ otherwise. 
%%%
For functions $f(K)$ and $g(K)$, we have $f(K)  = \Theta\big( g(K) \big)$ if $\lim_{K \to \infty}f(K) / g(K) = c$, 
where $c > 0$  is a finite constant.
%%%
\section{Problem Setting}
\label{sec:problem_setting}
%%%
Consider a $K$-cell cellular network in which each cell $k$, where $k \in \langle K \rangle$, comprises a base station denoted by BS-$k$ and $L$ user equipments, each denoted by UE-$(l_{k},k)$, where $l_{k} \in \langle L \rangle$.
%%%
The set of tuples corresponding to all UEs in the network is given by
$\mathcal{U} \triangleq \left\{(l_{k},k) : l_{k} \in \langle L \rangle, k \in 
\langle K \rangle  \right\}$.
%%%%
%%%

We focus on the downlink mode, 
where an independent message denoted by $W_{k}^{[l_k]}$ is communicated to each UE-$(l_{k},k)$. 
%%%
The input-output relationship at the $t$-th channel use is given by
%%%
\begin{equation}
%%%
\label{eq:IBC_system model}
Y_{k}^{[l_{k}]}(t)
 = \sum_{i = 1}^{K} \bar{P}^{\alpha_{ki}^{[l_{k}]}}  G_{ki}^{[l_{k}]}(t) X_{i}(t)
+ Z_{k}^{[l_{k}]}(t).
%%%
\end{equation}
%%%
In the above, $ X_{i}(t) , Y_{k}^{[l_{k}]}(t), Z_{k}^{[l_{k}]}(t) \in \mathbb{C}$ are, respectively, 
the symbol transmitted by BS-$i$, the symbol received by UE-$(l_{k},k)$, and the zero-mean unit-variance additive white Gaussian noise (AWGN) at UE-$(l_{k},k)$.
%%%
The signal transmitted by BS-$i$ is subject to the unit average power constraint 
$\frac{1}{T}\sum_{t=1}^{T}\E \big[\big|X_{i}(t)\big|^{2}\big] \leq 1$,
where $T$ is the communication duration in channel uses.
%%%
The $T$-channel-use-long signal (or codeword) of BS-$i$ is given by 
$ \bm{X}_{i} \triangleq \big(X_{i}(t) :  t \in \langle T \rangle \big)$.
%%%

%%%
For GDoF purposes, we define $\bar{P} \triangleq  \sqrt{P}$, where $P$ is a nominal power parameter
that approaches infinity in the GDoF limit.
%%%
The exponent ${\alpha_{ki}^{[l_{k}]}} > 0$ is known as the channel strength parameter between BS-$i$ and UE-$(l_{k},k)$, while $G_{ki}^{[l_{k}]}(t) \in \mathbb{C}$ is the corresponding channel fading coefficient.
%%%
The array of all fading coefficients is given by $\bm{G} \triangleq \big( G_{ki}^{[l_{k}]}(t) : (l_k,k) \in \mathcal{U}, i \in \langle K \rangle, t \in \langle T \rangle \big)$.
%%%
We assume that UEs have access to both channel strength parameters and channel fading coefficients, while BSs know channel strength parameters (and the probability distribution of fading coefficients) only.
%%%
These assumptions, which are discussed in 
greater detail further on, lead to robust GDoF characterizations, that are far 
less sensitive to perturbations in the channel state as seen by transmitters, compared
to GDoF results that assume perfect channel knowledge at the transmitters.  
%%%

%%%
We define $\bm{\alpha} \in \mathbb{R}_{+}^{K \times K \times L}$ as the $3$-dimensional array comprising all channel strength parameters of a given network, where  the $(i,j,l_i)$-th element of $\bm{\alpha}$ is given by 
$\bm{\alpha}(i,j,l_i) = \alpha_{i j}^{[l_i]}$.
%%%
Note that $\bm{\alpha}$ describes the topology of a $K \times KL$ network, specifying the strengths of connections between different BS-UE pairs; and hence we will often refer to $\bm{\alpha}$ as a network. As seen further on, we focus
on special regimes (i.e. subsets of networks) specified by imposing conditions on $\bm{\alpha}$.
%%%
\begin{remark}
For ease of exposition, we assume that we have an equal number of users given by $L$ in each cell.
%%%
As far as GDoF results are concerned, this symmetrization in the number of users per-cell incurs no loss of generality---for any cell $i$ with $L_{i} < L$ users, we may add $L - L_i$ trivial users, with strengths $\alpha_{ij}^{[l_i]} = 0$ for all $j \in \langle K \rangle$, which have no influence on the GDoF.
%%%
\hfill $\lozenge$
%%%
\end{remark}
%%%
\begin{remark}
\label{rermak:GDoF_DoF}
In the considered model, the capacity of a link connecting BS-$j$ and UE-$(l_i,i)$  is given by $C_{ij}^{[l_i]}(P) \approx {\alpha_{ij}^{[l_i]} } \log ( P ) $.
%%%
While individual link capacities, and hence the network capacity, are scaled up to infinity in the GDoF limit, defined by taking $P \to \infty$ while normalizing by $\log(P)$, their ratios remain approximately fixed, i.e.  $C_{ij}^{[l_i]}(P) / C_{kl}^{[l_k]}(P) \approx 
\alpha_{ij}^{[l_i]}/ \alpha_{kl}^{[l_k]}$.
%%%
This is the main advantage of  the GDoF  model \cite{Etkin2008,Jafar2010} compared to the less refined (yet more intuitive) DoF model, recovered by setting $\alpha_{ij}^{[l_i]} = 1$ for all links \cite{Cadambe2008,Suh2011}. 
%%%
Indeed, the GDoF model allows us to de-emphasize the effects of additive noise, hence focusing on the interaction between different signals, while maintaining the disparity amongst the strengths of different links, which in turn allows us to distinguish between different regimes (e.g. weak and strong interference regimes).  A caveat is the overwhelming complexity incurred by the many parameters in $\bm{\alpha}$.
%%%
\hfill $\lozenge$
%%%
\end{remark}
%%%
\subsection{Finite Precision CSIT}
%%%
We assume that channel strength parameters $\bm{\alpha}$ are perfectly known to both the transmitters (BSs) and the receivers (UEs). 
%%%
Channel fading coefficients $\bm{G}$, however, are perfectly known to the receivers but only available up to finite precision at the transmitters (finite precision CSIT).
%%%
In particular, the channel fading coefficients in $\bm{G}$ satisfy the so-called bounded density assumption \cite{Davoodi2016} (see also \cite[Def. 1]{Davoodi2017a}); and the transmitters are only aware of the joint probability density functions of the channel fading coefficients and not their actual realizations.\footnote{A set of real-valued random variables $\mathcal{G}$ satisfies the bounded density assumption if 
for all finite-cardinality disjoint subsets  $\mathcal{G}_1$,$\mathcal{G}_2$ of $\mathcal{G}$, the joint probability density function of random variables in $\mathcal{G}_1$ conditioned on random variables in $\mathcal{G}_2$ exists and is bounded above by  $f_{\max}^{|\mathcal{G}_1|}$, for some constant $1 \leq f_{\max} < \infty$ (independent of $P$). 
%%%%
In the case of independent channel coefficients, it is sufficient to have the marginal densities bounded by $f_{\max}$.
For complex-valued random variables, $\mathcal{G}$ is taken as the union of the sets of all real parts and all imaginary parts (see \cite[Sec. II.C]{Davoodi2017a}).} 
%%%
Therefore, the transmitted codewords $\bm{X}_{1},\ldots,\bm{X}_{K}$ are independent of the realizations of channel fading coefficients in $\bm{G}$, yet may depend on their joint distributions as well as channel strength parameters.
%%%%
For more on the significance of the finite precision CSIT assumption, and a simple instructive example highlighting the essential  features of this model, readers are referred to\cite[Sec. II.C]{Chan2020}.
%%%%
\begin{remark}
%%%%
It is worthwhile highlighting that the above-described finite precision CSIT model can be easily extended to include imperfect feedback of channel fading coefficients, where transmitters have access to an estimate $\hat{\bm{G}}$ of $\bm{G}$ (see\cite[Sec. II.D]{Davoodi2016}).
%%%%
In this case, $\bm{G}$ must satisfy the bounded density assumption when conditioned on $\hat{\bm{G}}$; and 
the transmitted codewords $\bm{X}_{1},\ldots,\bm{X}_{K}$ are independent of $\bm{G}$ given $\hat{\bm{G}}$.
%%%%
For ease of exposition, we proceed without including an estimate $\hat{\bm{G}}$ in the CSIT model, however, 
results in this paper extend directly to such case of imperfect feedback.
%%%
\hfill $\lozenge$
%%%%
\end{remark}
%%%%
\subsection{Levels of Multi-Cell Cooperation}
\subsubsection{No Cooperation (IBC)}
\label{subsec:no_cooperation}
%%%
Under no cooperation, each BS has access to the set of messages intended to  UEs in the same cell only,  and the network is modeled by an IBC comprising $K$ mutually interfering Gaussian BCs (cells)---see Fig. \ref{fig:IBC_MBC}(left).
%%%
For instance, BS-$k$ has messages $\bm{W}_k = \big( W_{k}^{[l_k]} : l_k \in \langle L \rangle \big) $ and encodes them into the codeword  $ \bm{X}_{k}$, 
independently of all other BSs. 
%%%
On the other end,  UE-$(l_{k},k)$ in cell $k$ sees the contributions from all 
codewords $\bm{X}_{i}$ with $i \in \langle K \rangle \setminus \{k \}$ as inter-cell interference.
%%%

%%%
For any given $P$, achievable rate tuples $\mathbf{R}(P)= \big(R_{k}^{[l_{k}]}(P): (l_{k} , k) \in \mathcal{U}\big)$ and the capacity region $\mathcal{C}^{\mathrm{IBC}}(P)$  are defined in  a standard manner, see, e.g.  \cite{Joudeh2019a,Joudeh2019b}.
%%%
A GDoF tuple is denote by $\mathbf{d}= \big(d_{k}^{[l_{k}]}: (l_{k}, k) \in \mathcal{U}\big)$, and the GDoF region is defined in a standard fashion as 
%%%
\begin{equation}
\nonumber
\mathcal{D}^{\mathrm{IBC}} \triangleq  \left\{ \mathbf{d} \in \mathbb{R}_{+}^{|\mathcal{U}|}
:
d_{k}^{[l_k]} = \lim_{P \to \infty}\frac{R_{k}^{[l_k]}(P)}{\log(P)},  \  
\mathbf{R}(P) \in \mathcal{C}^{\mathrm{IBC}}(P)
\right\}.
\end{equation}
%%%
\subsubsection{Full Cooperation (MISO-BC)}
\label{subsec:full_cooperation}
%%%
Under full cooperation, all BSs have access to all messages $\big( \bm{W}_1  , \ldots, \bm{W}_K \big)$, jointly encoded into the vector codeword $\bm{X} \triangleq 
(\bm{X}_{1}, \ldots, \bm{X}_{K})$, comprising $K$ scalar codewords, and of which the $k$-th component $\bm{X}_{k}$ is transmitted through BS-$k$.
%%%
In this cooperative setting, each BS is viewed as an antenna in a large multi-antenna transmitter, and the network is modeled by a MISO-BC\footnote{This setting, with more receivers than transmit antennas, is also called an \emph{overloaded} MISO-BC \cite{Piovano2016}.} with a $K$-antenna transmitter and $KL$ single-antenna receivers---see Fig. \ref{fig:IBC_MBC}(right). 
%%%
The capacity region of the MISO-BC  is denoted by $\mathcal{C}^{\mathrm{MBC}}(P)$, and the GDoF region is defined  as 
%%%
\begin{equation}
\nonumber
\mathcal{D}^{\mathrm{MBC}} \triangleq \left\{ \mathbf{d} \in \mathbb{R}_{+}^{|\mathcal{U}|}
:
d_{k}^{[l_k]} = \lim_{P \to \infty}\frac{R_{k}^{[l_k]}(P)}{\log(P)},  \  
\mathbf{R}(P) \in \mathcal{C}^{\mathrm{MBC}}(P)
\right\}.
\end{equation}
%%%
\begin{figure}[t]
%\vspace{-1mm}
%%
\centering
\includegraphics[width = 0.8\textwidth]{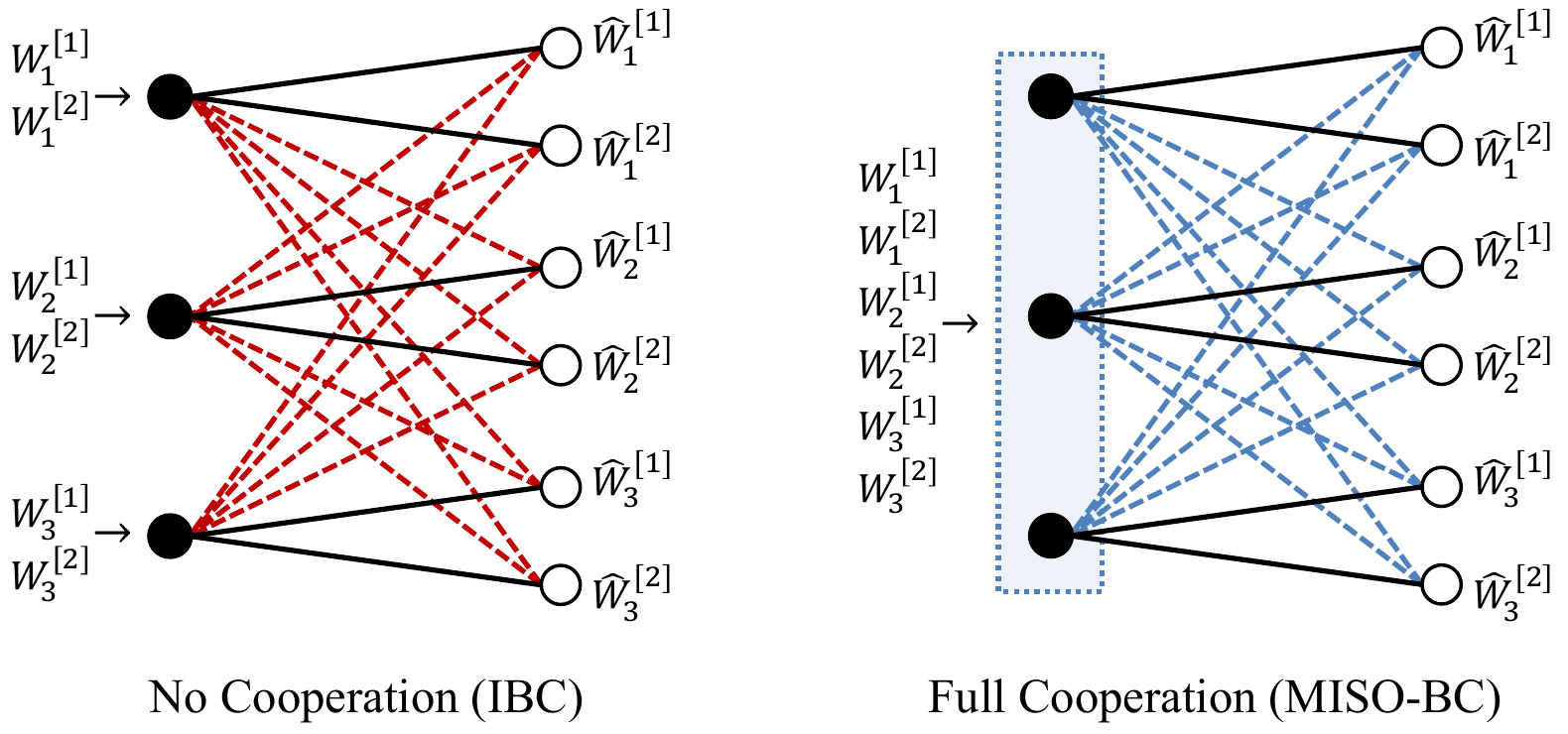}
\caption{
$3$-cell network with $2$ users per-cell under different levels of cooperation. Left: No multi-cell cooperation (IBC). Right: 
Full multi-cell cooperation (MISO-BC).}
\label{fig:IBC_MBC}
\end{figure}
%%%
%%%
\subsection{From Multi-Cell TIN to Multi-Cell SLS}
%%%% 
The finite precision CSIT assumption eliminates the GDoF benefits of schemes that rely on precise CSIT \cite{Davoodi2016}, such as IA for the IBC \cite{Suh2011}, and ZF and DPC for the MISO-BC \cite{Caire2003},
which naturally gives way to classes of robust schemes
based on power control, TIN, rate-splitting and SLS
\cite{Geng2015,Geng2016,Sun2016,Yi2016,Geng2015a,Gherekhloo2017,Yi2020,Gherekhloo2016,Joudeh2019a,Joudeh2019b,
Joudeh2016,Piovano2016,Davoodi2019,Joudeh2020},
%%%%
which only require coarse CSIT  (i.e. channel strength parameters).
%%%%
\subsubsection{Multi-Cell TIN}
\label{subsubsec:scheme_TIN}
%%%%
Perhaps the simplest robust multi-cell scheme, for which somewhat general optimality results can be shown, is the mc-TIN scheme  \cite{Joudeh2019b}.
%%%%
In mc-TIN, no BS cooperation is assumed and each BS employs a power-controlled single-cell-type transmission scheme, while each UE treats all inter-cell interference as additional Gaussian noise.
%%%
This boils down to superposition coding and successive decoding in each cell.
%%%
The signal of BS-$i$ is composed as
%%%
\begin{equation}
X_{i}(t) = \sum_{l_i \in \langle L \rangle } \sqrt{q_{i}^{[l_i]} } X_{i}^{[l_i]}(t)
\end{equation}
%%%
where each message $W_{i}^{[l_i]}$ is encoded into the signal $\bm{X}_{i}^{[l_i]}$ using an independent Gaussian codebook with unit average power, 
and $q_{i}^{[l_i]} \geq 0$ is the power allocated to UE-$(l_i,i)$ such
that the per-BS power constraint $\sum_{l_i \in \langle L \rangle  }  q_{i}^{[l_i]} \leq 1$
is not violated.
%%%
For successive decoding orders given by the permutation functions $p_1,\ldots, p_K$, 
%%%
each UE-$\big(p(l),i \big)$ successively decodes and cancels
the signals $X_{i}^{[p_i(1)]},\ldots, X_{i}^{[p_i(l-1)]}$, in this specific order, before decoding its own signal $X_{i}^{[p_i(l)]}$, while treating all other signals (i.e. both intra-cell and inter-cell interference) as noise.
%%%

%%%%
The set of GDoF tuples achieved  through all feasible power allocation policies and successive decoding orders constitute the mc-TIN achievable GDoF region (i.e. TINA region), denoted by $\mathcal{D}^{\mathrm{TINA}}$ in this paper. An explicit characterizations of the TINA region is obtained in \cite{Joudeh2019b}.
%%%%
\subsubsection{Multi-Cell SLS}
\label{subsubsec:scheme_SLS}
%%%%
Under full BS cooperation, i.e. the MISO-BC with finite precision CSIT,  there is generally no non-trivial regime for which mc-TIN is GDoF optimal. 
%%%%
The closest robust alternative in this case is the SLS scheme \cite{Davoodi2019}, which we call mc-SLS in multi-cell settings.\footnote{Schemes similar to SLS have been investigated in tandem with precoder design under the name of rate-splitting, see, e.g., \cite{Mao2018,Li2020}. These works focus on optimization aspects, with less emphasis on fundamental limits.}  
%%%%
In the most general form of mc-SLS, the $K$ BSs jointly transmit a superposition of independent Gaussian codewords $(\bm{X}_{\mathcal{S}}: \mathcal{S} \subseteq \mathcal{U} )$,
where each $\bm{X}_{\mathcal{S}}$  is decoded by all users in the subset $\mathcal{S}$, and treated as additional Gaussian noise by users not in $\mathcal{S}$.
%%%
The signal transmitted by BS-$i$ is hence composed as 
%%%
\begin{equation}
X_{i}(t) = \sum_{\mathcal{S} \subseteq \mathcal{U}} \sqrt{q_{\mathcal{S},i}} X_{\mathcal{S}}(t)
\end{equation}
%%%
where $q_{\mathcal{S},i} \geq 0$ is the power allocated to codeword $\bm{X}_{\mathcal{S}}$
by BS-$i$, such that
$ \sum_{\mathcal{S} \subseteq \mathcal{U}} q_{\mathcal{S},i} \leq 1$.
%%%
UEs employ successive decoding with arbitrary orders, optimized alongside power control variables.
%%%

%%%
The above scheme is essentially a multiple multicast transmission, where a distinct message is communicated to each subset of users, and the GDoF achieved by different messages comprise the tuple $(d_{\mathcal{S}} : \mathcal{S} \subseteq \mathcal{U})$.
%%%
To go from the original unicast messages $\big(W_{k}^{[l_k]} : (l_k,k) \in \mathcal{U}\big)$ to multiple multicast transmission, rate-splitting is used.
%%%
Each message $W_{k}^{[l_k]}$ is split into sub-messages given by
$\big(W_{k,\mathcal{S}}^{[l_k]}  : \mathcal{S} \subseteq \mathcal{U}, (l_k,k) \in \mathcal{S}\big)$, and then each set of sub-messages indexed by the same $\mathcal{S} \subseteq \mathcal{U}$, i.e.  
$\bm{W}_{\mathcal{S}} = \big(W_{k,\mathcal{S}}^{[l_k]}  :  (l_i,i) \in \mathcal{U} \big)$,
is encoded into the codeword $\bm{X}_{\mathcal{S}}$.
%%%
It follows that
%%%
\begin{equation}
\label{eq:multiple_multicast_GDoF}
d_{\mathcal{S}} = \sum_{(l_i,i) \in \mathcal{S}} d_{i,\mathcal{S}}^{[l_i]}
\end{equation} 
%%%
where $d_{i,\mathcal{S}}^{[l_i]}$ is the GDoF carried by sub-message $W_{i,\mathcal{S}}^{[l_i]} $, i.e. the portion of $d_{\mathcal{S}} $ assigned to UE-$(l_i,i)$.
%%%
It follows that the SLS region, denoted by  $\mathcal{D}^{\mathrm{SLS}} $, is given by all GDoF tuples $\mathbf{d} \in \mathbb{R}_{+}^{|\mathcal{U}|}$
that satisfy 
%%%
\begin{align}
d_i^{[l_i]} \leq  \sum_{\mathcal{S} \subseteq \mathcal{U}: (l_i,i) \in \mathcal{S}} d_{i,\mathcal{S}}^{[l_i]}
\end{align}
%%%
for some multiple multicast GDoF tuple $(d_{\mathcal{S}} :  \mathcal{S} \subseteq \mathcal{U} )$, with components given by \eqref{eq:multiple_multicast_GDoF}, achieved through a feasible power control policy and successive decoding orders.
%%%%
It is readily seen that the SLS scheme depends on a high number of auxiliary design variables---a major barrier in the face of obtaining an explicit characterization of $\mathcal{D}^{\mathrm{SLS}}$ and proving optimality results \cite{Davoodi2019}.
%%%%
However, in some cases, the representation of $\mathcal{D}^{\mathrm{SLS}}$  can be simplified as we will see further on in Section \ref{sec:2_cell_SLS}. 
%%%%
\section{Regimes of Interest and Useful Definitions}
%%%%
In this section, we present a number of definitions which are essential for the formulation and proofs of the main results, presented in subsequent sections. 
%%%
We start by highlighting that, without loss of generality, we may assume that in each cell $i$,  direct link strengths are ordered as
\begin{equation}
\label{eq:strength_order}
\alpha_{ii}^{[1]}  \leq  \alpha_{ii}^{[2]} \leq \cdots \leq \alpha_{ii}^{[L]}.
\end{equation}
%%%
That is, same-cell users are in an ascending order with respect to their SNRs.
%%%%
Note that in the absence of inter-cell interference, \eqref{eq:strength_order}  determines the degradedness order in each cell (or BC).
%%%%

%%%
Throughout this work, we focus on three regimes of channel parameters (or networks). Each regime is a subset of  $\mathbb{R}_{+}^{K \times K \times L}$, and can be thought of as a collection of networks that share certain (desirable) properties.
%%%%
Before introducing the regimes of interest, it is worthwhile highlighting that all three regimes are included in a larger \emph{weak} inter-cell interference regime, described by the set of all networks $\bm{\alpha} \in \mathbb{R}_{+}^{K \times K \times L}$ that satisfy the SIR order
%%%
\begin{equation}
\label{eq:SIR_order}
0 \leq  \alpha_{ii}^{[1]} - \alpha_{ij}^{[1]} \leq \alpha_{ii}^{[2]} - \alpha_{ij}^{[2]} \leq \cdots \leq \alpha_{ii}^{[L]} - \alpha_{ij}^{[L]} 
\end{equation}
%%%%
for all $i,j \in \langle K \rangle$.
%%%%
We find it quite instructive to think of \eqref{eq:SIR_order}, as well as the three regimes of interest introduced further on,  in terms of two types of conditions: 1) inter-cell conditions, which control interference levels between distinct cells; and 2) intra-cell conditions, which govern (and preserve) the order amongst same-cell users under inter-cell interference. 
%%%%

%%%%
The first type of conditions is captured by the left-most inequality in \eqref{eq:SIR_order}, which (alongside the other inequalities) implies that for all $(l_i,i) \in \mathcal{U}$ and $j \in \langle K \rangle$, we have
%%%%
\begin{equation}
\label{eq:inter_cell_order}
\alpha_{ii}^{[l_i]}  \geq \alpha_{ij}^{[l_i]}.
\end{equation}
%%%%
That is, a direct link between a BS and any of its associated UEs must be no weaker than interfering (or cross) links to the same UE.
%%%%
This places the cellular network in the \emph{weak} inter-cell interference regime, an analog of the IC's weak interference regime \cite{Etkin2008}.
%%%%
As we will see in the following subsection, all three regimes of interest fall within this weak inter-cell interference regime.

%%%
The second type of conditions is captured by the right-most $L -1 $ inequalities in \eqref{eq:SIR_order}, which in turn guarantee that for any cell $i$, and against interference from any other cell $j$, the SIR order of same-cell users should follow their SNR order in \eqref{eq:strength_order}. 
%%%
In other words, a stronger user in the SNR sense must also be stronger in the SIR sense, that is
%%%%
\begin{equation}
\label{eq:intra_cell_order}
\alpha_{ii}^{[l_i]} \geq \alpha_{ii}^{[l_i-1]} \implies \alpha_{ii}^{[l_i]} - \alpha_{ij}^{[l_i]} \geq \alpha_{ii}^{[l_i-1]} - \alpha_{ij}^{[l_i-1]}.
\end{equation}
%%%%
As we will see, the above SIR order also holds in all three regimes of interest.
%%%%%
\begin{remark}
%%%
The SIR order in \eqref{eq:SIR_order} greatly simplifies the mc-TIN scheme described in Section \ref{subsubsec:scheme_TIN}.
%%%%
In particular, under this order it is optimum from the TINA region perspective to use the optimum interference-free successive decoding order (natural order) in each cell, i.e.   $p_i(l) = l$ for all $l \in \langle L \rangle$ and $i \in \langle K \rangle$.
%%%
This is not necessarily the case when \eqref{eq:SIR_order} does not hold, see \cite{Joudeh2019a,Joudeh2019b}.
%%%
Moreover, the order in \eqref{eq:SIR_order} may be seen as quite natural---cell-centre users with high long-term SNRs are also expected to have higher long-term SIRs compared to cell-edge users.
%%%
%%%
\hfill $\lozenge$
%%%
\end{remark}
%%%%%
\subsection{Multi-Cell TIN, CTIN and SLS regimes}
\label{subsec:regimes}
%%%%%
We now present the three regimes of interest, starting with the mc-TIN regime 
\cite{Joudeh2019a,Joudeh2019b}. 
%%%%%
\begin{definition}
\label{def:TIN_regime}
\textbf{(mc-TIN Regime).}
%%%
This regime is denoted by $\mathcal{A}^{\mathrm{TIN}}$, and is given by all networks $\bm{\alpha} \in \mathbb{R}_{+}^{K \times K \times L}$ that satisfy
%%%
\begin{align}
%%%
\label{eq:TIN_cond_2}
\alpha_{ii}^{[l_i]}   & \geq \alpha_{ij}^{[l_i]} +  \alpha_{ki}^{[l_{k}]} \\
%%%
\label{eq:TIN_cond_1}
\alpha_{ii}^{[l_i ]}   & \geq  \alpha_{ij}^{[l_i ]} +  \alpha_{ii}^{[l_{i} - 1 ]} - 
\big( \alpha_{ij}^{[l_i - 1]} - \alpha_{ij}^{[l_i]} \big)^{+}
%%%
\end{align}
%%%
for all cells $i,j,k \in \langle K \rangle$ such that  $i \notin \{j,k\}$, and for all users  $l_{i} \in \langle 2:L \rangle$ and $ l_{k} \in \langle L \rangle$.
%%%
\hfill $\lozenge$
%%%
\end{definition}
%%%
It can be verified that the inequalities in \eqref{eq:TIN_cond_2} and \eqref{eq:TIN_cond_1}  are stricter versions of their counterparts in \eqref{eq:inter_cell_order} and \eqref{eq:intra_cell_order}, respectively. 
%%%%%%
The mc-TIN regime derives its significance from the fact that mc-TIN achieves the entire GDoF region of the IBC in this regime \cite{Joudeh2019b}.  
%%%%%%
Hence, this regime can be thought of as a \emph{very weak} inter-cell interference regime.
%%%%
The mc-TIN regime generalizes the TIN regime introduced for $K\times K$ interference networks by Geng et al. \cite{Geng2015} to $K \times KL$ cellular networks.
%%%%%%
Next, we present a strictly larger regime called the mc-CTIN regime \cite{Joudeh2019a,Joudeh2019b}. 
%%%%%
\begin{definition}
\label{def:CTIN_regime}
\textbf{(mc-CTIN Regime).}
%%%
This regime is denoted by $\mathcal{A}^{\mathrm{CTIN}}$, and is given by all networks $\bm{\alpha} \in \mathbb{R}_{+}^{K \times K \times L}$ that satisfy
%%%
\begin{align}
%%%
\label{eq:CTIN_cond_2}
\alpha_{ii}^{[l_i]}   & \geq  \max \big(\alpha_{ij}^{[l_i]}  +  \alpha_{ji}^{[l_{j}]},  \alpha_{ik}^{[l_i]} + \alpha_{ji}^{[l_{j}]} - \alpha_{jk}^{[l_{j}]} \big) \\
%%%
\label{eq:CTIN_cond_1}
\alpha_{ii}^{[l_i ]}   & \geq  \alpha_{ij}^{[l_i ]} +   \alpha_{ii}^{[l_{i} - 1]} - \alpha_{ij}^{[l_i - 1]}
%%%
\end{align}
%%%
%%%
for all cells $i,j,k \in \langle K \rangle$ such that  $i \notin \{j,k\}$, and for all users  $l_{i} \in \langle 2:L \rangle$ and $ l_{j} \in \langle L \rangle$.
%%%
\hfill $\lozenge$
%%%
\end{definition}
%%%
Note that \eqref{eq:CTIN_cond_2} is a stricter version of \eqref{eq:inter_cell_order}, while \eqref{eq:CTIN_cond_1} is identical to \eqref{eq:intra_cell_order}.
%%%
In the mc-CTIN regime, the TINA GDoF region $\mathcal{D}^{\mathrm{TINA}}$ is a convex polyhedron without the need for time-sharing, which is not necessarily the case outside the mc-CTIN regime \cite{Joudeh2019b,Joudeh2019a}.
%%%
The mc-CTIN regime generalizes the CTIN regime introduced for $K \times K$ interference networks by Yi and Caire \cite{Yi2016}.
%%%

%%%
While the TINA region is convex in the mc-CTIN regime, the GDoF region of the IBC  remains unknown in this regime. In Section \ref{subsec:result_mc_CTIN}, we settle this question under the assumption of finite precision CSIT, and we show that
mc-TIN is GDoF optimal for the IBC in the mc-CTIN regime.
%%%%%%
Next, we introduce the largest of the three regimes, which we call
the mc-SLS regime. 
%%%%
\begin{definition}
\label{def:SLS_regime}
\textbf{(mc-SLS Regime).}
%%%
This regime is denoted by $\mathcal{A}^{\mathrm{SLS}}$, and is given by all networks $\bm{\alpha} \in \mathbb{R}_{+}^{K \times K \times L}$ that satisfy
%%%
\begin{align}
%%%
%%%
\label{eq:SLS_cond_2}
\alpha_{ii}^{[l_i]}   & \geq  \max \big(\alpha_{ij}^{[l_i]}  ,  \alpha_{ki}^{[l_{k}]},  \alpha_{ik}^{[l_i]} + \alpha_{ji}^{[l_{j}]} - \alpha_{jk}^{[l_{j}]} \big) \\
%%%
\label{eq:SLS_cond_1}
\alpha_{ii}^{[l_i ]}   & \geq  \alpha_{ij}^{[l_i ]} +   \alpha_{ii}^{[l_{i} - 1]} - \alpha_{ij}^{[l_i - 1]}
%%%
%%%
\end{align}
%%%
for all cells $i,j,k \in \langle K \rangle$ such that  $i \notin \{j,k\}$, and for all users  $l_{i} \in \langle 2:L \rangle$ and $ l_{k} \in \langle L\rangle$.
%%%
\hfill $\lozenge$
%%%
\end{definition}
%%%
Note that \eqref{eq:SLS_cond_2} is a stricter version of \eqref{eq:inter_cell_order}, while \eqref{eq:SLS_cond_1} is identical to \eqref{eq:intra_cell_order}.
%%%
Moreover, we have
\begin{equation}
\nonumber
\mathcal{A}^{\mathrm{TIN}} \subseteq \mathcal{A}^{\mathrm{CTIN}} \subseteq \mathcal{A}^{\mathrm{SLS}}.
\end{equation}
%%%
The mc-SLS regime extends the SLS regime introduced by Davoodi and Jafar \cite{Davoodi2019} to $K \times KL$ cellular networks.
%%%
The significance of the SLS regime is due to the fact that the SLS scheme is GDoF optimal for $K \times K$ MISO-BCs with $K \leq 3$  in this regime.
%%%%
While this result may also hold for $K > 3$,  no proof is heretofore available, mainly due to the overwhelming complexity of the GDoF region when $K$ is arbitrary.
%%%%
Nevertheless, the SLS regime is still very useful in the sense that it lends itself to extremal network analysis as recently shown by Chan et al. \cite{Chan2020}, where extremal gains of  transmitter cooperation over TIN are explicitly obtained. 
%%%%

%%%%
In Sections \ref{subsec:results_SLS} and \ref{subsec:results_extremal_gains}, we show that $K \times KL$ cellular networks enjoy desirable properties in the mc-SLS regime. 
%%%%
In particular, we derive a malleable outer bound for  the $K \times KL$  MISO-BC in the mc-SLS regime, which lends itself directly to extremal network analysis, through which we bound the gain of mc-Co over mc-TIN in all three regimes of interest.
%%%
\subsection{Useful Definitions}
%%%%
\begin{definition}
\label{def:cycle}
\textbf{(Cycles).}
%%%
A cycle $\pi$ of length $|\pi| = M$ is an ordered sequence of $M$ users from distinct cells, given by
%%%
\begin{equation}
\label{eq:cycle_pi}
\pi  = \big( (l_{i_1},i_1) \rightarrow (l_{i_2},i_2) \rightarrow \cdots \rightarrow (l_{i_M},i_M) \big).
\end{equation}
%%%
We define $\{\pi\} \triangleq  \big\{ (l_{i_1},i_1), (l_{i_2},i_2) , \cdots , (l_{i_M},i_M) \big\}$ as the set of users involved in cycle $\pi$.
%%%
The $m$-th user in a cycle $\pi$ is also denoted by $\pi(m) = (l_{i_m},i_m)$, 
from which \eqref{eq:cycle_pi} is equivalently expressed as 
$\pi  = \big( \pi(1)  \rightarrow \cdots \rightarrow\pi(M) \big) $.
%%%
The set of all cycles (of all lengths) is denoted by $\Pi$.
%%%
Each cycle $\pi$ is associated with an implicit cycle encompassing BS indices, given by
%%%
\begin{equation}
\sigma   = (i_1 \rightarrow i_2 \rightarrow \cdots \rightarrow i_M)
\end{equation}
%%%
which is also written as $\sigma   = \big( \sigma(1)  \rightarrow \cdots \rightarrow \sigma(M) \big) $.
For any cycle of length $M$, indices are interpreted modulo $M$, e.g.
$\pi(M+m) = \pi(m)$ and $\sigma(M+m) = \sigma(m)$, for all integers $m$.
%%%
\hfill $\lozenge$
%%%
\end{definition}
%%%%
Next, we define cycle bounds for the IBC and MISO-BC associated with cycles defined above.
%%%%
For each cycle $\pi \in \Pi$, a cycle bound is a bound on the sum-GDoF of users in the set 
\begin{equation}
\nonumber
 \big\{ (s_k,k) : s_k \in \langle l_k \rangle , (l_k,k) \in \{\pi\}  \big\}.
\end{equation} 
%%%
That is, the set comprising each participating user $(l_k,k) \in \{\pi\}$, as well as same-cell users that precede
user $(l_k,k) $ in the SNR (or SIR) order.
%%%%
Cycle bounds are defined as follows.
%%%%
\begin{definition}
\label{def:cycle_bounds_IBC}
\textbf{(IBC Cycle Bounds).}
%%%
The IBC cycle bound associated with $\pi \in \Pi$ is given by 
%%%
\begin{equation}
\label{eq:IBC_cycle_bound}
\sum_{(l_k,k) \in \{\pi\} }  \sum_{s_k \in \langle l_k \rangle}   d_{k}^{[s_k]} \leq \Delta_{\pi}
\end{equation}
%%%
where the quantity $\Delta_{\pi}$ is defined as
%%%%
\begin{equation}
\Delta_{\pi} \triangleq 
\begin{cases} 
\alpha_{i_1 i_ 1}^{[l_{i_1}]}, & \textrm{if} \ M = 1 \\
\sum_{m = 1}^{M} \Big(  \alpha_{i_m i_m}^{[l_{i_m}]} - 
\alpha_{i_{m+1} i_m}^{[l_{i_{m+1}}]} \Big) , & \textrm{if} \ M > 1.
\end{cases}
\end{equation}
%%%
As shown in \cite{Joudeh2019b}, the TINA region $\mathcal{D}^{\mathrm{TINA}}$ can be expressed as a union of polyhedra (polyhedral TIN regions), each described in terms of the above IBC cycle bounds.
%%%
These cycle bounds will also constitute outer bounds for the IBC in the mc-CTIN regime, as we will show further on.
%%%
\hfill $\lozenge$
%%%
\end{definition}
%%%%
\begin{definition}
\label{def:cycle_bounds_MBC}
\textbf{(MISO-BC Cycle Bounds).}
%%%
Each  $\pi \in \Pi$ gives rise to $M = |\pi|$ MISO-BC cycle bounds.
%%%
The $m$-th bound associated with $\pi$ is given by
%%%
\begin{equation}
\sum_{(l_k,k) \in \{\pi\} }  \sum_{s_k \in \langle l_k \rangle}   d_{k}^{[s_k]} \leq \Delta_{\pi,m}^{+}
\end{equation}
%%%
where the  quantity 
$\Delta_{\pi,m}^{+}$ is defined as 
%%%%
\begin{equation}
\Delta_{\pi,m}^{+} \triangleq 
\begin{cases} 
\Delta_{\pi}, & \textrm{if} \ M = 1, \\
\Delta_{\pi} + \alpha_{i_{m+1} i_m}^{[l_{i_{m+1}}]}, & \textrm{if} \ M > 1.
\end{cases}
\end{equation}
%%%
The above cycle bounds will constitute outer bounds for the MISO-BC in the mc-SLS regime, as we will show further on.
%%%
\hfill $\lozenge$
%%%
\end{definition}
%%%% 
\section{Main Results and Insights}
%%%%
In this section, we present the main results of this work, alongside some observations and insights.
%%%%
It is worthwhile highlighting that our main results are all under the assumption of finite precision CSIT. 
%%%%
Therefore, this is always assumed (implicitly), unless mentioned otherwise.
%%%%
\subsection{Optimality of Multi-Cell TIN}
\label{subsec:result_mc_CTIN} 
%%%%
From the IBC cycle bounds in Definition \ref{def:cycle_bounds_IBC}, we construct what is known as the polyhedral-TIN (PTIN) region, denoted as $\mathcal{D}^{\mathrm{PTIN}}$. This is given by all tuples $\mathbf{d} \in \mathbb{R}_{+}^{KL}$ that satisfy
%%%
\begin{align}
%%%
\label{eq:TIN_region_CTIN_regime}
 \sum_{(l_k,k) \in \{ \pi \}}  \bar{d}_k^{[ l_k ]}  &   \leq \Delta_{\pi},  \forall \pi \in \Pi
\end{align}
%%%
where $\bar{d}_k^{[ l_k ]} $ is a shorthand notation for $\sum_{s_k \in \langle l_k \rangle} d_k^{[ s_k ]} $.
From \cite{Joudeh2019b}, we know that $\mathcal{D}^{\mathrm{TINA}} = \mathcal{D}^{\mathrm{PTIN}}$ in the mc-CTIN regime,
and $ \mathcal{D}^{\mathrm{IBC}} = \mathcal{D}^{\mathrm{TINA}} = \mathcal{D}^{\mathrm{PTIN}}$ in the mc-TIN regime.\footnote{These statements hold regardless of the CSIT assumptions. Moreover, outside the mc-CTIN regime, 
$\mathcal{D}^{\mathrm{TINA}}$ is a union of PTIN regions, each defined by selecting a subset of active users and deactivating remaining users.}
%%%
Next, we show that the latter also holds in the mc-CTIN regime under finite precision CSIT.
%%%
\begin{theorem}
\label{theorem:CTIN_optimality}
In the mc-CTIN regime, mc-TIN is GDoF optimal for the $K \times KL$ IBC under finite precision CSIT. 
That is,
$\bm{\alpha} \in \mathcal{A}^{\mathrm{CTIN}} \implies 
\mathcal{D}^{\mathrm{IBC}} = 
\mathcal{D}^{\mathrm{TINA}} = \mathcal{D}^{\mathrm{PTIN}}$.
%%%
\end{theorem}
%%%
The direct part (achievability) of Theorem \ref{theorem:CTIN_optimality} follows from \cite{Joudeh2019b}, as highlighted above.
%%% 
The converse relies on a new application of AI bounds \cite{Davoodi2016,Davoodi2017a} to cellular networks, and is presented in Section \ref{sec:outer_bounds}.
%%%
Next, we draw some insights from the results in Theorem \ref{theorem:CTIN_optimality}.
%%%

%%%
\begin{figure}[t]
%\vspace{-1mm}
%%
\centering
\includegraphics[width = 0.7\textwidth]{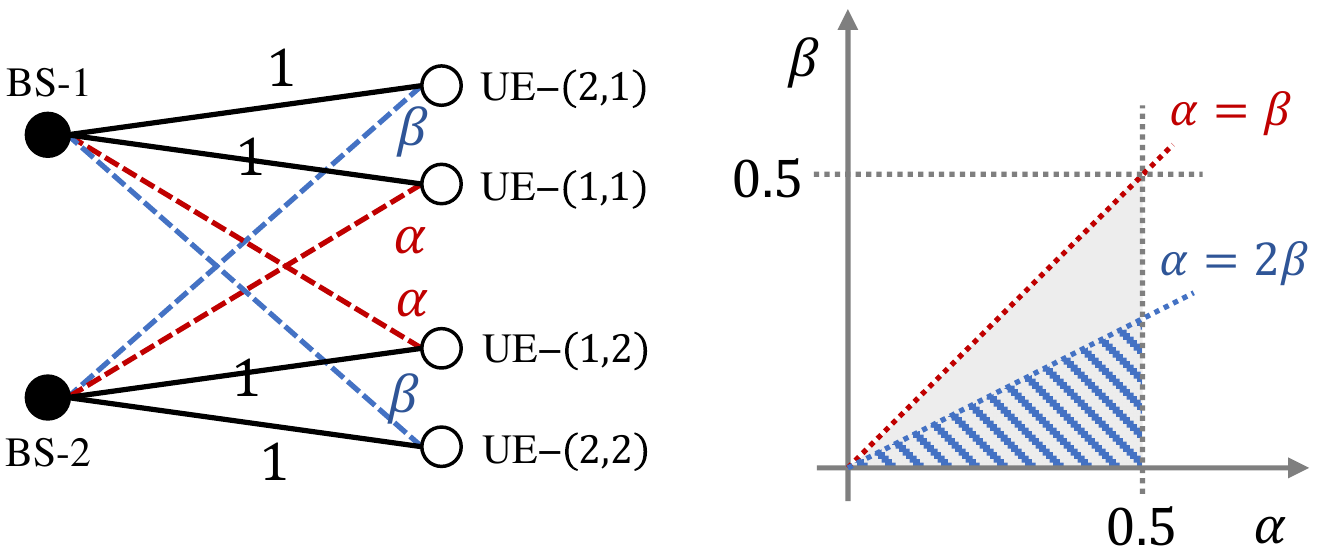}
\caption{Left: Symmetric $2$-cell network with $2$ users per-cell. Right: The mc-CTIN regime in grey and the mc-TIN regime in striped blue, assuming   $\alpha \geq \beta$.}
\label{fig:IBC_IA_example}
\end{figure}
%%%
%%%
We first observe that under prefect CSIT,
mc-TIN is not necessarily GDoF optimal for the IBC in the mc-CTIN regime, which is in sharp contrast to the finite precision CSIT result in Theorem \ref{theorem:CTIN_optimality}.
%%%
This is seen through the example in Fig. \ref{fig:IBC_IA_example}, particularly in the case where $\alpha = 0.5$ and $\frac{1}{3} < \beta \leq 0.5$,  which places this network in the mc-CTIN regime but not in the mc-TIN regime.
%%%
In this case, a sum-GDoF of  $\frac{4}{3}$ is achieved through the signal-space downlink IA scheme proposed by Suh and Tse in \cite{Suh2011}.
%%%
On the other hand,  the sum-GDoF achieved using mc-TIN is at most $2(1 - \beta) < \frac{4}{3}$.
%%%

%%%
The observation that signal-space IA seizes to be useful for the IBC under finite precision CSIT from the GDoF
perspective is perhaps not very surprising in light of previous results \cite{Davoodi2016,Davoodi2017a}. 
%%%
It is, however, not immediately clear (without Theorem \ref{theorem:CTIN_optimality}) whether a 
form of robust inter-cell interference management that is more sophisticated than mc-TIN, e.g. based on rate-splitting and partial interference decoding \cite{Etkin2008,Davoodi2017a}, is required to achieve the optimal GDoF in the mc-CTIN regime.
%%%
Theorem \ref{theorem:CTIN_optimality} settles this issue, and shows that the simple  mc-TIN scheme is indeed  GDoF optimal for the IBC in the mc-CTIN regime.
%%%%
Moreover, we observe in this regime that the sum-GDoF of the $K \times KL$ IBC is equal to the sum-GDoF of
the underlying $K \times K$ IC, comprising only the strongest user in the SNR (or SIR) sense from each cell, i.e. UE-$(L,i)$ for all $i \in \langle K \rangle$.
%%%
This is seen from the fact that in the mc-TIN scheme, UE-$(L,i)$ decodes all $L$ messages in cell $i$, hence limiting the total GDoF achieved in cell $i$ to the GDoF achievable by  UE-$(L,i)$.
%%%
As mc-TIN  is GDoF optimal in the mc-CTIN regime under finite precision CSIT, it follows that  in each cell $i$, all users other than UE-$(L,i)$ are \emph{redundant} from a standpoint of maximizing the sum-GDoF.
%%%

%%%
It is also worthwhile  noting that the redundancy of \emph{weaker} users in the mc-CTIN regime does not necessarily hold under perfect CSIT.
%%%
This is clearly seen through the above example of Fig. \ref{fig:IBC_IA_example}, where lower SIR users are \emph{necessary} to achieve IA gains---the sum-GDoF is limited to that achieved via TIN in the absence of these weaker users \cite{Etkin2008}.
%%%
Finite precision CSIT, on the other hand, eliminates all GDoF gains of IA (in both signal-space and signal-scale), and brings classical schemes, based on Gaussian random codes, superposition and power control, back to the forefront.
%%%
Under these schemes, inter-cell interference is Gaussian and unstructured, and the SIR order in \eqref{eq:SIR_order} 
takes on the  interference-free role of the SNR order in \eqref{eq:strength_order}, i.e.
cells resemble their interference-free degraded BC counterparts, in which weaker users are redundant.
%%%%
\subsection{Cooperative Outer Bound and Optimality of 2-Cell SLS}
\label{subsec:results_SLS}
%%%%
Similar to the region $\mathcal{D}^{\mathrm{PTIN}}$  constructed from IBC cycle bounds, the MISO-BC cycle bounds in
Definition \ref{def:cycle_bounds_MBC} can be used to construct a region given by all tuples $\mathbf{d} \in \mathbb{R}_{+}^{KL}$ that satisfy
%%%
\begin{align}
\label{eq:MBC_GDoF_outerbound}
 \sum_{(l_k,k) \in \{ \pi \}} \bar{d}_k^{[ l_k ]}  &   \leq \Delta_{\pi,m}^{+}, \;
 \forall  m \in \langle |\pi| \rangle, \pi \in \Pi. 
\end{align}
%%%
We denote this region by $\mathcal{D}^{\mathrm{SLS}}_{\mathrm{out}}$, a notational choice justified by the following theorem.
%%%%
\begin{theorem}
\label{theorem:MISO_BC_outerbound}
In the mc-SLS regime, the GDoF region of the $K \times KL$ MISO-BC under finite precision CSIT is included in $\mathcal{D}^{\mathrm{SLS}}_{\mathrm{out}}$. That is,
$\bm{\alpha} \in \mathcal{A}^{\mathrm{SLS}} \implies 
\mathcal{D}^{\mathrm{MBC}} \subseteq  
\mathcal{D}^{\mathrm{SLS}}_{\mathrm{out}}$.
%%%
\end{theorem}
%%%
The proof of Theorem \ref{theorem:MISO_BC_outerbound} also relies on AI bounds, and it is presented 
in Section \ref{subsec:MISO_BC_outer_bound}.
%%%
A natural question that follows is whether $\mathcal{D}^{\mathrm{SLS}}_{\mathrm{out}}$ is tight for the $K \times KL$ MISO-BC in the mc-SLS regime.
%%%
In the following result, we show that this is indeed the case in $2$-cell networks.
%%%
\begin{theorem}
\label{theorem:2_cell_SLS}
In the mc-SLS regime, mc-SLS is GDoF optimal for the $2 \times 2L$ MISO-BC under finite precision CSIT.  
In this case, we have
$\bm{\alpha} \in \mathcal{A}^{\mathrm{SLS}} \implies 
\mathcal{D}^{\mathrm{MBC}} = 
\mathcal{D}^{\mathrm{SLS}} = \mathcal{D}^{\mathrm{SLS}}_{\mathrm{out}}$.
%%%
\end{theorem}
%%%
The proof of Theorem \ref{theorem:2_cell_SLS} is presented in Section \ref{sec:2_cell_SLS}.
%%%
Key to the proof is recognizing that the achievable GDoF region $\mathcal{D}^{\mathrm{SLS}} $, described in its generality in Section \ref{subsubsec:scheme_SLS}, greatly simplifies for $2$-cell networks in the mc-SLS regime.
%%%
Instead of the $2^{2L} - 1$ independent codewords communicated in the general mc-SLS scheme, it is sufficient to transmit $2L + 1$ codewords in this case: $2L$ non-cooperative codewords from the mc-TIN scheme; and a superimposed cooperative common codeword decoded by all users.
%%%
This simplification allows for an efficient elimination of auxiliary variables used in describing $\mathcal{D}^{\mathrm{SLS}} $, which turns out to coincide with $ \mathcal{D}^{\mathrm{SLS}}_{\mathrm{out}}$.
%%%

%%%%
Beyond $2$-cell networks,  the outer bound in Theorem \ref{theorem:MISO_BC_outerbound} is tight in some special cases. In Appendix \ref{appendix:H_ICI}, we show that this is the case for a class of $K \times KL$ networks in the mc-SLS regime  with homogeneous inter-cell interference.
%%%%%
Nevertheless, the outer bound in Theorem \ref{theorem:MISO_BC_outerbound} is not tight in general.
%%%%
This is seen from the $3 \times 3$ MISO-BC setting studied in \cite{Davoodi2019}---while it was shown that $\mathcal{D}^{\mathrm{MBC}} = \mathcal{D}^{\mathrm{SLS}}$ in the SLS regime (i.e. SLS is optimal here), 
the achievable region $\mathcal{D}^{\mathrm{SLS}}$  and the outer bound  
$\mathcal{D}^{\mathrm{SLS}}_{\mathrm{out}}$ do not coincide in this settings. In particular, additional bounds, which are not implied by the cycle bounds in \eqref{eq:MBC_GDoF_outerbound}, are generally required to fully describe $\mathcal{D}^{\mathrm{SLS}}$ in the $3 \times 3$ setting.
%%%
Nevertheless, these additional inequalities are not required for our  next purpose of extremal network analysis.
%%%
Despite the fact that $\mathcal{D}^{\mathrm{SLS}}_{\mathrm{out}}$ is not tight in general, this outer bound remains very useful for studying the potential benefits of full mc-Co  over mc-TIN under finite precision CSIT. 
%%%
In particular, in each of the three regimes on interest, $\mathcal{D}^{\mathrm{SLS}}_{\mathrm{out}}$ is tight for a subset of \emph{extremal} networks, in which sum-GDoF gains of full multi-cell cooperation over mc-TIN are maximized.
%%%
This allows us to obtain sharp characterizations of these extremal  GDoF gains, as we see next.
%%%%
\subsection{Extremal GDoF Gains of Multi-Cell Cooperation over Multi-Cell TIN}
\label{subsec:results_extremal_gains}
%%%%
Equipped with the results in Theorem \ref{theorem:CTIN_optimality} and Theorem \ref{theorem:MISO_BC_outerbound}, we are ready to conduct extremal network analysis to understand the robust GDoF benefits of  mc-Co over mc-TIN in the weak inter-cell interference regimes of interest.
%%%
To this end, we define the sum-GDoF  achievable by mc-TIN as 
%%% 
\begin{equation}
\label{eq:sum_GDoF_TIN}
d^{\mathrm{TIN}}_{\Sigma}(\bm{\alpha} )  \triangleq \max_{\mathbf{d} \in \mathcal{D}^{\mathrm{TINA}} (\bm{\alpha} )}
\sum_{(l_k,k) \in \mathcal{U}} d_{k}^{[l_k]}.
\end{equation}
%%%
In a similar manner, the sum-GDoF achieved through mc-Co is defined as
%%% 
\begin{equation}
\label{eq:sum_GDoF_MBC}
d^{\mathrm{MBC}}_{\Sigma}(\bm{\alpha} )  \triangleq \max_{\mathbf{d} \in \mathcal{D}^{\mathrm{MBC}} (\bm{\alpha} )}
\sum_{(l_k,k) \in \mathcal{U}} d_{k}^{[l_k]}.
\end{equation}
%%%
Note that the dependency of the GDoF on $\bm{\alpha}$ is made explicit in this part.
%%%
It is also worth noting that $d^{\mathrm{MBC}}_{\Sigma}(\bm{\alpha} ) $ is the optimal sum-GDoF  of the underlying 
$K \times KL$ MISO-BC, with no restriction on the employed scheme, hence representing the ultimate 
performance of mc-Co schemes.
%%%
On the other hand, $d^{\mathrm{TIN}}_{\Sigma}(\bm{\alpha} ) $ is the maximum sum-GDoF achieved while restricting to the simple non-cooperative scheme of mc-TIN. This is optimal for the underlying $K \times KL$ IBC in the mc-CTIN regime, yet its optimality is not guaranteed outside the mc-CTIN regime.
%%%

%%%
We are interested in the extremal sum-GDoF gain from  mc-Co relative to mc-TIN in the three regimes of interest.
%%%
For a regime $\mathcal{A} \subset \mathbb{R}^{K \times K \times L}_{+}$, the extremal gain is defined as 
%%%
\begin{equation}
\label{eq:def_eta}
\eta_{K,L}(\mathcal{A}) \triangleq \max_{\bm{\alpha} \in \mathcal{A}} 
\frac{d_{\Sigma}^{\mathrm{MBC}} (\bm{\alpha} )}
{d_{\Sigma}^{\mathrm{TIN}}(\bm{\alpha} )}
\end{equation}
which is parametrized by the network dimensions $K,L$.
%%%
As elaborated by Chan \emph{et al.} \cite{Chan2020}, the extremal gain $\eta_{K,L}(\mathcal{A})$ captures the \emph{potential} benefits of 
mc-Co over mc-TIN in $\mathcal{A}$, and does not necessarily reflect \emph{typical} or \emph{average} performance gains.
%%%
Nevertheless, studying $\eta_{K,L}(\mathcal{A})$ can still be very useful for bringing closure to regimes in which potential gains are small; as well as identifying regimes  that warrant further investigation due to large potential gains.
%%%
The two types of conclusions are seen through the following theorem, which generalizes the results in  
\cite[Th. 5.1, 6.1 and 7.1]{Chan2020} to $K \times KL$ networks.
%%%
We implicitly assume that $K \geq 2$, as the single-cell case is degenerate.
%%%
\begin{theorem}
\label{theorem: extremal gain}
The extremal sum-GDoF gain of mc-Co over mc-TIN in $K \times KL$ networks under finite precision CSIT  in the three regimes of interest is as follows:
%%%
\begin{equation}
\label{eq:extremal_gains_IBC}
\eta_{K,L}(\mathcal{A})  = 
\begin{cases}
\frac{3}{2}, & \mathcal{A} = \mathcal{A}^{\mathrm{TIN}} \\ 
2 -  \frac{1}{K}, & \mathcal{A} = \mathcal{A}^{\mathrm{CTIN}} \\
\Theta\big(\log(K) \big), & \mathcal{A} = \mathcal{A}^{\mathrm{SLS}}.
\end{cases}
\end{equation}
%%%
\end{theorem}
%%%%
It is readily seen that in each of the three regimes of interest, the extremal gain $\eta_{K,L}(\mathcal{A})$ is independent of the number of users per-cell $L$, and therefore we have $\eta_{K,L}(\mathcal{A}) = \eta_{K,1}(\mathcal{A})$.
%%%%
This is a consequence of the SIR order in \eqref{eq:SIR_order}, which under finite precision CSIT and from a sum-GDoF standpoint, renders the first $L - 1$ users in each cell \emph{redundant}  with respect to user  $L$ (i.e. the strongest user).
%%%%
This redundancy has already been highlighted in the context of mc-TIN and the IBC in Section \ref{subsec:result_mc_CTIN}.
%%%%
The same type of redundancy is exhibited by the MISO-BC in the mc-SLS regime, as seen from the outer bound  $\mathcal{D}^{\mathrm{SLS}}_{\mathrm{out}}$ (see Section \ref{sec:extremal_gains}).
%%%%
The redundancy of weaker users in the regimes of interest is exploited in the proof of Theorem \ref{theorem: extremal gain},
presented in Section \ref{sec:extremal_gains}.
%%%%

%%%%
Theorem \ref{theorem: extremal gain} quantifies the intuition that in regimes where inter-cell interference is sufficiently weak to the extent  that it is GDoF optimal to treat it as noise for the IBC,  gains due to BS cooperation are limited.
%%%% 
In particular, mc-Co schemes provide at most constant factor GDoF gains
over mc-TIN in the mc-TIN and mc-CTIN regimes, bounded above by $1.5$ and $2$ respectively.
%%%%
These multiplicative gains are relatively small and, more critically, they do not scale with $K$ and $L$.
%%%%
On the other hand, there is far greater potential in the mc-SLS regime, specifically the part not included in the mc-CTIN regime.
%%%%
Here the extremal GDoF gain of mc-Co over mc-TIN may scale as $\log(K)$ in large networks, 
rendering this regime more interesting for further investigation.
%%%%

\begin{remark}
%%%%
Achieving the extremal GDoF gains in \eqref{eq:extremal_gains_IBC} for $K\times K$ networks has been demonstrated in \cite{Chan2020} through constructive proofs.
%%%%
For instance, for the GDoF gain of  $\Theta\big(\log(K) \big)$ in the SLS regime,  Chan et al. \cite{Chan2020} find a $K\times K$ network  with a special hierarchical topology that greatly benefits from transmitter cooperation.
%%%%
In the identified network, the sum-GDoF achieved through TIN is bounded above by $2$ irrespective of the number of users $K$, while transmitter cooperation through the SLS scheme achieves a sum-GDoF of $1 + \frac{1}{2} \log (K)$ (see \cite[Sec. VII.B]{Chan2020}). 
%%%%
Note that showing the existence of such special network is sufficient for proving that the corresponding extremal GDoF gain is attainable.
%%%%
As $K\times K$  networks are a special case of $K\times KL$ networks, the constructive proofs of Chan et al. \cite{Chan2020} extend directly to the $K\times KL$ cellular networks considered here. 
%%%%
In Section \ref{sec:extremal_gains}, we further show that such gains cannot be exceeded in cellular networks.
%%%
\hfill $\lozenge$
%%%%
\end{remark}
\section{Outer Bounds}
\label{sec:outer_bounds}
%%%
In this section we present proofs for the outer bounds in Theorem  \ref{theorem:CTIN_optimality} and Theorem \ref{theorem:MISO_BC_outerbound}.
%%%
For this purpose, we work with a deterministic approximation of the channel model in \eqref{eq:IBC_system model} given by
%%%
\begin{equation}
%%%
\label{eq:IBC_system model_det}
\bar{Y}_{k}^{[l_{k}]}(t)
 = \sum_{i = 1}^{K} \big\lfloor \bar{P}^{\alpha_{ki}^{[l_{k}]} - \alpha_{\max , i} }  G_{ki}^{[l_{k}]}(t) \bar{X}_{i}(t) 
 \big\rfloor.
%%%
\end{equation}
%%%
In \eqref{eq:IBC_system model_det}, we have $ \alpha_{\max , i} \triangleq  \max_{(l_{j} , j)  \in \mathcal{U} } \alpha_{ji}^{[l_{j}]}$, and  both the real and imaginary components of $ \bar{X}_{i}(t) $ are drawn from the
integer alphabet $\big\langle 0 : \lceil \bar{P}^{\alpha_{\max , i}} \rceil  \big\rangle$. 
%%%
It can be easily checked that in all three regimes of interest, we have $\alpha_{\max , i} = \alpha_{ii}^{[L_{i}]}$.
%%%
As shown in \cite{Davoodi2016}, the GDoF region of the above deterministic channel model contains its counterpart GDoF of the original Gaussian model.
%%%

%%%
We now recall a key lemma from  \cite{Chan2020} (see  \cite{Davoodi2017a} for the proof). To this end, we consider the outputs
%%%
\begin{align}
\bar{Y}_{k}(t) & =  \sum_{i = 1}^{K} \big\lfloor \bar{P}^{\lambda_{i} - \alpha_{\max,i} } G_{ki}(t) \bar{X}_{i}(t)  \big\rfloor  \\
%%%
\bar{Y}_{j}(t) & =  \sum_{i = 1}^{K} \big\lfloor \bar{P}^{\nu_{i} - \alpha_{\max,i} } G_{ji}(t) \bar{X}_{i}(t)  \big\rfloor
\end{align}
%%%
where $\lambda_{i},\nu_{i} \in [0, \alpha_{\max,i}]$, for all $i \in \langle K \rangle$, are the corresponding channel strengths.
%%%
We use our standard notation $\bar{\bm{X}}_{i}$, $\bar{\bm{Y}}_{k}$ and $\bm{G}$ for codewords, received signals and channel coefficients, respectively.
%%%
\begin{lemma} {\normalfont{\textbf{(Aligned Images Bounds \cite[Lemma 1]{Davoodi2017a})}}}
\label{lemma:AI_diff_enropies}
%%%
Let $U$ be an auxiliary random variable and assume that $(U,\bar{\bm{X}}_{1},\ldots,\bar{\bm{X}}_K)$ are independent of  $\bm{G}$.
%%%
We have 
%%%
\begin{equation}
\label{eq:entropy_diff_lemma}
H\big( \bar{\bm{Y}}_{k} \mid \bm{G}, U \big) - 
H\big( \bar{\bm{Y}}_{j} \mid \bm{G}, U \big) \leq  
\max_{i \in \langle K \rangle} 
(\lambda_{i} - \nu_{i})^{+} T \log(P) + T o(\log(P)).
\end{equation}
%%%
\end{lemma}
%%%
%%%
Lemma \ref{lemma:AI_diff_enropies} bounds\footnote{The bound in Lemma \ref{lemma:AI_diff_enropies} follows directly from the aligned images bound in \cite[Lemma 1]{Davoodi2017a}. The same bound was also used more recently in \cite{Chan2020}. Our statement of Lemma \ref{lemma:AI_diff_enropies}  follows the statement of \cite[Lemma 2.1]{Chan2020}.} the maximum difference of entropies (in the GDoF sense) between the two received signals 
$ \bar{\bm{Y}}_{k} $ and $ \bar{\bm{Y}}_{j}$, that can be created by any set of codewords 
$\bar{\bm{X}}_{1}, \ldots, \bar{\bm{X}}_{K}$, which are independent of the exact realizations of channel coefficients 
in $\bm{G}$.
%%%
The bound in \eqref{eq:entropy_diff_lemma} tells us that a maximum difference of entropies is created through $\bar{\bm{X}}_{i}$, where $i$ is the index yielding a maximum difference in strengths $(\lambda_{i} - \nu_{i})^{+}$.
%%%
An example of Lemma \ref{lemma:AI_diff_enropies} is shown in Fig. \ref{fig:AI_bounds_example}.
%%%
\begin{figure}[h]
\vspace{-1mm}
\centering
\includegraphics[width = 0.6\textwidth]{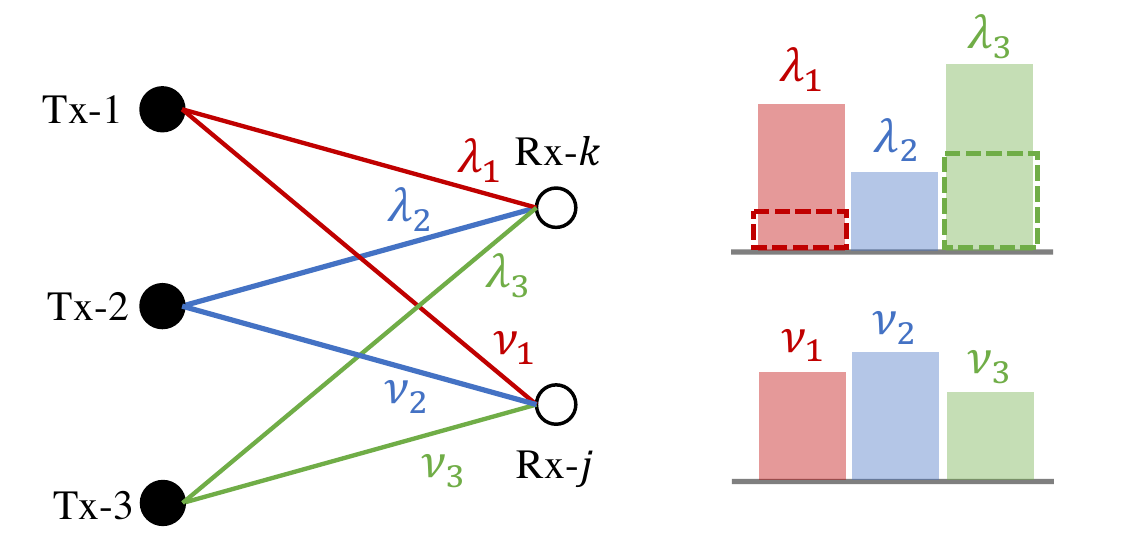}
\caption{ Left: $3$-transmitter, $2$-receiver network. Right: The corresponding received power levels. The entropy difference between receivers $k$ and $j$ is maximized through transmitter $3$.}
\label{fig:AI_bounds_example}
\end{figure}
%%% 

Since  $o(\log(P))$ terms  are inconsequential for GDoF results, they will be dropped henceforth for brevity.
%%%
Whenever we do so,  we use $\dotleq$ to write inequalities  to indicate that $o(\log(P))$ has been dropped (see below).
%%%
Moreover, we adopt the compact notation of \cite{Chan2020} to represent differences of entropies of the type  in  Lemma \ref{lemma:AI_diff_enropies}.
%%%
Using this brief notion,  inequality \eqref{eq:entropy_diff_lemma} is expressed as
%%%
\begin{equation}
\label{eq:entropy_diff_lemma_brief}
\Hbb \big(\bm{\lambda} \mid U \big)  -  \Hbb \big( \bm{\nu} \mid U \big)   \; \dotleq \; \max_{i \in \langle K \rangle} 
(\lambda_{i} - \nu_{i})^{+} T \log(P)
\end{equation}
%%%
where $ \bm{\lambda} \triangleq  [\lambda_{1} \;  \cdots \; \lambda_{K}] $ and $ \bm{\nu} \triangleq  [\nu_{1} \;  \cdots \; \nu_{K}] $.
%%%
The expression in \eqref{eq:entropy_diff_lemma_brief} succinctly captures the essential parts in \eqref{eq:entropy_diff_lemma}, 
i.e. channel strength levels, and will be employed further on.
%%%
\subsection{Single-Cell Outer Bound}
%%%
Next, we employ Lemma  \ref{lemma:AI_diff_enropies}  to bound the sum-rate of users in cell $i$, for any $i \in \langle K \rangle$, under finite precision CSIT.
%%%
This single-cell bound holds in the general weak inter-cell interference regime specified by the order in \eqref{eq:SIR_order}, and hence holds in all three regimes of interest.
%%%%
Moreover, the  bound is applicable to both the IBC and MISO-BC, as we will see further on.
%%%%
\begin{lemma}
\label{lemma:single_cell_bound}
%%%%
In the weak inter-cell interference regime specified by \eqref{eq:SIR_order}, 
the sum-rate of users associated with cell $i$, where $i \in \langle K \rangle$,
 is bounded above as
%%%%
\begin{equation}
\label{eq:single_cell_sum_rate_bound}
T \sum_{l \in \langle L \rangle} R_{i}^{[l]} \; \dotleq \; \big( \alpha_{ii}^{[L]} - \alpha_{ii}^{[1]} \big) T \log(P) + 
H  \big(   \bar{\bm{Y}}_{i}^{[1]} |  \bm{G}, U_i  \big)  - H \big(  \bar{\bm{Y}}_{i}^{[L]}  | \bm{G}, \bm{W}_{i} , U_i \big)
\end{equation}
%%%%
where $U_i$ is a side information random variable, independent of $\bm{G}$ and $\bm{W}_i$.
%%%%
\end{lemma}
%%%%
\begin{proof}
%%%
In what follows, we use $W_{i}^{[1:l]}$ to denoted $\big(W_{i}^{[1]}, \ldots, W_{i}^{[l]} \big)$, where $(l,i) \in \mathcal{U}$.
%%%
From Fano's inequality, and the independence of $W_{i}^{[l]} $ and $\big( W_{i}^{[1:l - 1]}, U_i  \big)$, we obtain
\begin{align}
\nonumber
 T &  \sum_{l = 1}^{L }  \  (R_{i}^{[l]} - \epsilon_{T} ) \leq \
 \sum_{l = 1}^{L}    I \Big(  W_{i}^{[l]} ;  \bar{\bm{Y}}_{i}^{[l]}  | \bm{G},  W_{i}^{[1:l - 1]}, U_i   \Big) \\
 %%%
 \nonumber
 = &  \sum_{l = 1}^{L}    H  \Big(   \bar{\bm{Y}}_{i}^{[l]}  |   \bm{G},W_{i}^{[1:l - 1]},U_i  \Big)  - 
 H \Big(  \bar{\bm{Y}}_{i}^{[l]}  |  \bm{G}, W_{i}^{[1:l]} , U_i \Big)  \\
 %%%
%\label{eq:converse_IBC_single_cell_1}
\nonumber
  = &   \sum_{l = 2}^{L} \!   \Big[ \!  H  \Big(   \bar{\bm{Y}}_{i}^{[l ]}  |  \bm{G}, W_{i}^{[1:l - 1]} ,U_i \Big) \!  - \! 
 H \Big(  \bar{\bm{Y}}_{i}^{[l - 1]}  |  \bm{G}, W_{i}^{[1:l - 1 ]}, U_i  \Big)  \!  \Big]
%%%
\! + \! H  \Big(   \bar{\bm{Y}}_{i}^{[1]} |  \bm{G},U_i  \Big) \!  - \! H \Big(  \bar{\bm{Y}}_{i}^{[L]}  |  \bm{G}, \bm{W}_{i},U_i \Big)  \\
   %%%
\label{eq:converse_IBC_single_cell_2}
  = &   \sum_{l = 2}^{L}  \!    \Big[ \!  \Hbb \Big( \bm{\alpha}_i^{[l]}  |   W_{i}^{[1:l - 1]},U_i \Big)  \!  -  \!   
  \Hbb \Big( \bm{\alpha}_i^{[l-1]}  |  W_{i}^{[1:l - 1]},U_i  \Big) \!  \Big]
  \! + \! H  \Big(   \bar{\bm{Y}}_{i}^{[1]} |  \bm{G} , U_i  \Big)  \! - \! H \Big(  \bar{\bm{Y}}_{i}^{[L]}  |  \bm{G}, \bm{W}_{i} , U_i\Big).
\end{align}
%%%
In the above,  $\epsilon_T $ denotes terms that approach zero as $T$ goes to infinity. 
As these terms  disappear once we take $T \to \infty$, they will be dropped henceforth. 
%%%
Next, we focus on bounding the sum of differences of entropies described in brief notation in \eqref{eq:converse_IBC_single_cell_2}.
%%%
By employing Lemma \ref{lemma:AI_diff_enropies}, we obtain 
%%%
\begin{align}
%\label{eq:converse_IBC_single_cell_3}
\nonumber
 \sum_{l = 2}^{L}   \Hbb \Big( \bm{\alpha}_i^{[l]}  |  W_{i}^{[1:l - 1]}  , U_i \Big)  \!  -  \!    
  \Hbb \Big( \bm{\alpha}_i^{[l-1]}  | W_{i}^{[1:l - 1]}  , U_i  \Big)  
& \; \dotleq \; \sum_{l = 2}^{L}  \max_{j \in \langle K \rangle } \big(  \alpha^{[l ]}_{ij}  - \alpha^{[l  -1 ]}_{ij}   \big)^{+}  T \log (P) \\ 
\label{eq:converse_IBC_single_cell_4}
& = \sum_{l = 2}^{L}  \big(    \alpha^{[l ]}_{ii}  - \alpha^{[l -1 ]}_{ii} \big) T \log (P)   \\
\label{eq:converse_IBC_single_cell_5}
& =  \big(    \alpha^{[L ]}_{ii}  - \alpha^{[1]}_{ii} \big) T \log (P).
\end{align}
%%%
\eqref{eq:converse_IBC_single_cell_4} follows from the SIR order $\alpha^{[l ]}_{ii} - \alpha^{[l  ]}_{ij}  \geq  \alpha^{[l - 1 ]}_{ii} - \alpha^{[l -1  ]}_{ij} $ in \eqref{eq:SIR_order}, which implies
\begin{equation}
\nonumber 
\alpha^{[l ]}_{ii} - \alpha^{[l - 1 ]}_{ii}   \geq   \alpha^{[l  ]}_{ij}  - \alpha^{[l -1  ]}_{ij}, \ \forall l  \in \langle 2 : L \rangle, 
\end{equation}
%%%
as well as the SNR order $\alpha_{ii}^{[l]} \geq \alpha_{ii}^{[l-1]}$ in \eqref{eq:strength_order}, which allows us to drop $(\cdot)^{+}$.
%%%
By combining  \eqref{eq:converse_IBC_single_cell_2} and \eqref{eq:converse_IBC_single_cell_5},  we obtain the desired inequality in 
\eqref{eq:single_cell_sum_rate_bound}, which holds for any cell $i \in \langle K \rangle$.
%%%
\end{proof}
%%%%
Note that  Lemma \ref{lemma:single_cell_bound} is derived without making explicit assumptions on transmitter cooperation, and it will be adapted to prove the outer bounds in Theorem \ref{theorem:CTIN_optimality} and Theorem \ref{theorem:MISO_BC_outerbound} by adjusting the transmitter cooperation assumptions and selecting suitable auxiliary variables $U_i$.
%%%
\begin{remark}
\label{remark:IBC_conditioning}
Before we prove the multi-cell converse bounds, it is worthwhile noting that when using Lemma  \ref{lemma:single_cell_bound} in the context of the IBC, the negative entropy term in  \eqref{eq:single_cell_sum_rate_bound} can be written as
\begin{equation}
\nonumber
H \big(  \bar{\bm{Y}}_{i}^{[L]}  \mid  \bm{G}, \bm{W}_{i} , U_i \big) =  H \big(  \widetilde{\bm{Y}}_{i}^{[L]}  \mid  \bm{G}, U_i \big) 
\end{equation}
where $ \widetilde{\bm{Y}}_{i}^{[L]}$ is obtained from $ \bar{\bm{Y}}_{i}^{[L]}$ by subtracting the contribution of the input signal $ \bar{\bm{X}}_{i}$.
%%%
This holds since the signal transmitted from BS-$i$ depends only on the set of messages in the same cell under no transmitter
 cooperation, i.e. $ \bar{\bm{X}}_{i}$ is fully determined by $\bm{W}_{i}$ (see Section \ref{subsec:no_cooperation}). Therefore, the contribution of $ \bar{\bm{X}}_{i}$ to $ \bar{\bm{Y}}_{i}^{[L]}$ can be subtracted, after which the conditioning on $ \bar{\bm{X}}_{i}$  (or $\bm{W}_{i}$) can be dropped, since $ \bar{\bm{X}}_{i}$  is independent of all remaining signals from BSs other than BS-$i$. 
%%%
This, however, does not hold for the MISO-BC, where messages to all users in the network are jointly encoded by all BSs (see Section \ref{subsec:full_cooperation}).
%%%
In this case, choosing the \emph{right} side information variables becomes critical to obtain the desired outer bound, as we will see in Section \ref{subsec:MISO_BC_outer_bound}.
%%%
\hfill $\lozenge$
%%%
\end{remark}
%%%
We are now equipped to prove the multi-cell converse bounds in Theorem \ref{theorem:CTIN_optimality}
and Theorem \ref{theorem:MISO_BC_outerbound}.
%%%%
We follow in the footsteps of the proofs of \cite[Th. 4.1]{Chan2020} and \cite[Lem. B.4]{Chan2020}, where similar converse bounds are derived for $K \times K$ networks. 
%%%%
In the following proofs, we generalize the bounds in \cite{Chan2020}  to  $K \times KL$ cellular networks, with the aid of single-cell bounds derived in Lemma \ref{lemma:single_cell_bound}.
%%%
\subsection{IBC Outer Bound: Proof of Theorem \ref{theorem:CTIN_optimality}}
\label{subsec:IBC_outer_bound}
%%%
In this part, we prove that each of the inequalities in 
\eqref{eq:TIN_region_CTIN_regime} is a valid outer bound for the IBC in the mc-CTIN regime.
%%%
Single-cell bounds in \eqref{eq:TIN_region_CTIN_regime},  associated with cycles of length $|\pi| = 1$, are directly obtained from the capacity region of the degraded Gaussian BC (see, e.g., \cite{Cover2012}).
%%%
We hence focus on multi-cell bounds in  \eqref{eq:TIN_region_CTIN_regime},  associated with cycles of length 
$|\pi| \geq 2$.
%%%

%%%
Consider a cycle $\pi \in \Pi$  of length $M = |\pi | \geq 2$, and eliminate all non-participating users (and their message), i.e. users not in the set $\big\{ (s_k,k) : s_k \in \langle l_k \rangle , (l_k,k) \in \{\pi\}  \big\}$.
%%%
This cannot hurt participating users.
%%%
For any participating cell $i$ and its $l_{i}$ participating users, where $(l_i,i) \in \{\pi\}$,  the corresponding sum-rate  $\bar{R}_{i}^{[l_{i}]} \triangleq \sum_{s_i = 1}^{l_{i}  } R_{i}^{[s_{i}]}$ is bounded above as 
%%%
\begin{equation}
\label{eq:single_cell_bound_IBC}
T\bar{R}_{i}^{[l_{i}]}   \; \dotleq \;   \big(    \alpha^{[l_{i } ]}_{i i}  -  \alpha^{[1]}_{i i } \big) T \log (P) 
+ H  \Big(   \bar{\bm{Y}}_{i}^{[1]} |  \bm{G}  \Big)  - H \Big(  \bar{\bm{Y}}_{i}^{[l_{i}]}  |  \bm{G}, \bm{W}_{i}  \Big).
\end{equation}
%%%
which follows directly from Lemma \ref{lemma:single_cell_bound}. Note that after eliminating non-participating users, $\bm{W}_{i}$ now corresponds to $\big(W_{i}^{[1]},\ldots,W_{i}^{[l_i]}  \big)$.
%%%
By summing over single-cell bounds obtained from \eqref{eq:single_cell_bound_IBC} of all $M$ participating cells, and after rearranging entropy terms, we obtain a cycle sum-rate bound as 
%%%
\begin{align}
\label{eq:converse_IBC_multi_cell_1}
T \! \! \!  \!   \sum_{(l_i,i) \in \{ \pi \}} \! \! \!  \!  \bar{R}_{i}^{[l_{i}]}  \; \dotleq \;
%%%
\sum_{m = 1}^{M}    \big(    \alpha^{[l_{i_m } ]}_{i_m i_m}  -  \alpha^{[1]}_{i_m i_m } \big) T \log (P) 
%%%
+ \sum_{m = 1}^{M} H  \Big(   \bar{\bm{Y}}_{i_m}^{[1]} |  \bm{G}  \Big)  - H \Big(  \bar{\bm{Y}}_{i_{m+1}}^{[l_{i_{m+1}}]}  |  \bm{G}, \bm{W}_{i_{m+1}}  \Big).
 %%% 
\end{align}
%%%
Focusing on the the 
sum of differences of entropies in  \eqref{eq:converse_IBC_multi_cell_1}, we obtain
%%%
\begin{align}
%%%
\label{eq:converse_IBC_multi_cell_2_hadamard_0}
 \sum_{m = 1}^{M}  H  \Big(   \bar{\bm{Y}}_{i_m}^{[1]} |  \bm{G}  \Big)  - H \Big(  \bar{\bm{Y}}_{i_{m+1}}^{[l_{i_{m+1}}]}  |  \bm{G}, \bm{W}_{i_{m+1}}  \Big)  & =  \sum_{m = 1}^{M} \mathbb{H} \big( \bm{\alpha}_{i_m}^{[1]} \big) -   \mathbb{H} \big( \bm{\alpha}_{i_{m+1}}^{[l_{i_{m+1}}]}   \mid  \bar{\bm{X}}_{i_{m+1}}  \big)  \\
%%%
\label{eq:converse_IBC_multi_cell_2_hadamard}
& = \sum_{m = 1}^{M} \mathbb{H} \big( \bm{\alpha}_{i_m}^{[1]} \big) -   \mathbb{H} \big( \bm{\alpha}_{i_{m+1}}^{[l_{i_{m+1}}]} \circ 
\bar{\mathbf{e}}_{i_{m+1}}  \big).
%%% 
\end{align}
%%%
The equality in \eqref{eq:converse_IBC_multi_cell_2_hadamard_0} is obtained by invoking \eqref{eq:entropy_diff_lemma_brief}, while noting that $\bar{\bm{X}}_{i_{m+1}}$ depends only on $\bm{W}_{i_{m+1}}$ in the IBC.
%%%
In \eqref{eq:converse_IBC_multi_cell_2_hadamard}, 
 $\bar{\mathbf{e}}_{i_{m+1}}$ is an appropriately-sized vector of all ones except for a single zero at the $(i_{m+1})$-th element, while $\circ$ denotes the element-wise product. Therefore
 %%%
\begin{equation}
\nonumber
\bm{\alpha}_{i_{m+1}}^{[l_{i_{m+1}}]} \circ
\bar{\mathbf{e}}_{i_{m+1}}  = \big[ \alpha_{i_{m+1} 1 }^{[l_{i_{m+1}}]}  \; \cdots \! \! \!  \! \! \! \! \!   \overbrace{0}^{(i_{m+1})\text{-th entry}} \! \! \! \! \! \! \! \! \cdots \; \alpha_{i_{m+1} K}^{[l_{i_{m+1}}]}  \big].
\end{equation}
%%%
For each $m \in \langle M \rangle$ in  \eqref{eq:converse_IBC_multi_cell_2_hadamard}, we subtracted the contribution of $\bar{\bm{X}}_{i_{m+1}}$ to  $\bar{\bm{Y}}_{i_{m+1}}^{[l_{i_{m+1}}]}$, after which the conditioning on $\bar{\bm{X}}_{i_{m+1}}$ can be dropped (see Remark \ref{remark:IBC_conditioning}).
%%%
From a GDoF perspective, this is equivalent to replacing $\alpha_{i_{m+1} i_{m+1}}^{[l_{i_{m+1}}]}$ in $\bm{\alpha}_{i_{m+1}}^{[l_{i_{m+1}}]}$ with a zero, as in \eqref{eq:converse_IBC_multi_cell_2_hadamard}.
%%%
By applying Lemma \ref{lemma:AI_diff_enropies} to \eqref{eq:converse_IBC_multi_cell_2_hadamard}, we obtain
%%%
\begin{align}
\nonumber
\sum_{m = 1}^{M} & \mathbb{H} \big( \bm{\alpha}_{i_m}^{[1]} \big) -   \mathbb{H} \big( \bm{\alpha}_{i_{m+1}}^{[l_{i_{m+1}}]} \circ 
\bar{\mathbf{e}}_{i_{m+1}}  \big)  \\
%%%
\label{eq:converse_IBC_multi_cell_pre3}
& \; \dotleq \;    \sum_{m = 1}^{M} \max \Big(  \alpha^{[1]}_{i_m i_{m+1}} ,  
%%%
\max_{k \in \langle K \rangle, k \neq i_{m+1}} \Big(  \alpha^{[1]}_{i_m k} 
- \alpha^{[l_{i_{m+1}}]}_{i_{m+1} k}   \Big)^{+}  \Big) T \log (P)   \\
%%%
\label{eq:converse_IBC_multi_cell_3}
&  \leq   \sum_{m = 1}^{M}   \big(   \alpha^{[1]}_{i_m i_m}   - \alpha^{[l_{i_{m+1}}]}_{i_{m+1} i_m}  \big) T \log (P).
%%% 
%%% 
\end{align}
%%%
%%%
The bound in  \eqref{eq:converse_IBC_multi_cell_3} holds due to the mc-CTIN condition in \eqref{eq:CTIN_cond_2}, which implies that
%%%
\begin{align}
\nonumber
\alpha^{[1]}_{i_m i_m}   - \alpha^{[l_{i_{m+1}}]}_{i_{m+1} i_m }  &  \geq  \alpha^{[1]}_{i_m k } 
- \alpha^{[l_{i_{m+1}}]}_{i_{m+1} k}  \\
\nonumber
\alpha^{[1]}_{i_m i_m}   - \alpha^{[l_{i_{m+1}}]}_{i_{m+1} i_m }   & \geq \alpha^{[1]}_{i_m i_{m+1}}.
\end{align}
%%%
By combining  the bounds in  \eqref{eq:converse_IBC_multi_cell_3}  and \eqref{eq:converse_IBC_multi_cell_1},
we obtain the desired cycle bound as
%%%%
\begin{align}
\nonumber
T \! \! \!  \!   \sum_{(l_i,i) \in \{ \pi \}} \! \! \!  \!  \bar{R}_{i}^{[l_{i}]} 
& \; \dotleq \;  \sum_{m = 1}^{M}   \left(
    \alpha^{[l_{i_m}]}_{i_m i_m}   - 
 \alpha^{[l_{i_{m+1}}]}_{i_{m+1} i_m}  \right) T \log (P) \\ 
 %%%
&  =  \Delta_{\pi} T \log (P).
\end{align}
%%%
The above applies to any cycle $\pi$ of length $| \pi | \geq 2$, which concludes the converse proof for the IBC.
%%%
\subsection{MISO-BC Outer Bound: Proof of Theorem \ref{theorem:MISO_BC_outerbound}}
\label{subsec:MISO_BC_outer_bound}
%%%
We now move on to proving that each of the inequalities in \eqref{eq:MBC_GDoF_outerbound} is a valid outer bound for the MISO-BC in the mc-SLS regime.
%%%
As in the previous part, we consider a cycle $\pi \in \Pi$  of length $M = |\pi | \geq 2$.
%%%
Unlike the IBC, however, such cycle is associated with $M$ bounds for the MISO-BC. We first focus on the $M$-th of such bounds, i.e. $ \sum_{(l_i,i) \in \{ \pi \}} \sum_{s_i \in \langle l_i \rangle} d_k^{[ s_i ]}    \leq \Delta_{\pi,M}^{+}$, and we address remaining bounds further on.
%%%
Following similar steps leading to \eqref{eq:converse_IBC_multi_cell_1}, we obtain
%%%
\begin{align}
\nonumber
T \! \! \sum_{(l_i,i) \in \{ \pi \}} \! \!  \bar{R}_{i}^{[l_{i}]}  \!  \; \dotleq \;
%%%
\sum_{m = 1}^{M}  &    \big(    \alpha^{[l_{i_m } ]}_{i_m i_m}  -  \alpha^{[1]}_{i_m i_m } \big) T \log (P) 
 + \\ 
%%%%%
\label{eq:converse_MISO_BC_1}
& \sum_{m = 1}^{M}   H  \Big(   \bar{\bm{Y}}_{i_m}^{[1]} |  \bm{G} , \bm{U}_{i_m} \Big)   -  H \Big(  \bar{\bm{Y}}_{i_{m+1}}^{[l_{i_{m+1}}]}  |  \bm{G}, \bm{W}_{i_{m+1}}, \bm{U}_{i_{m+1}} \Big).
 %%% 
\end{align}
%%%
For $m \in \langle M \rangle$, we  select the side information variable $\bm{U}_{i_m}$ as 
%%%%
\begin{equation}
\label{eq:side_information_U_i_m}
\bm{U}_{i_m} = \big(\bm{W}_{i_{m+1}} ,\ldots, \bm{W}_{i_M}  \big).
\end{equation} 
%%%
Note that $\bm{U}_{i_M}$ is empty, and hence users in cell $i_M$ are given no side information about messages intended to users in other cells.\footnote{It is worth noting that modulo $M$ is not used for cell indices in \eqref{eq:side_information_U_i_m}.}
%%% 
The sum of entropy differences in \eqref{eq:converse_MISO_BC_1} leads to
%%%
\begin{align}
%%%
\nonumber
 \sum_{m = 1}^{M} &  H  \Big(   \bar{\bm{Y}}_{i_m}^{[1]} |  \bm{G} , \bm{U}_{i_m} \Big)   -  H \Big(  \bar{\bm{Y}}_{i_{m+1}}^{[l_{i_{m+1}}]}  |  \bm{G}, \bm{W}_{i_{m+1}}, \bm{U}_{i_{m+1}} \Big) \\
%%%% 
\label{eq:converse_MISO_BC_2}
& \leq  \sum_{m = 1}^{M-1}  \bigg[   H  \Big(   \bar{\bm{Y}}_{i_m}^{[1]} |  \bm{G}, \bm{U}_{i_{m}} \Big)   - H \Big(  \bar{\bm{Y}}_{i_{m+1}}^{[l_{i_{m+1}}]}  |  \bm{G}, \bm{U}_{i_{m}}  \Big) \bigg] + H  \Big(  \bar{\bm{Y}}_{i_M}^{[1]} |  \bm{G}\Big) \\
%%%
\label{eq:converse_MISO_BC_3}
  & \leq \sum_{m = 1}^{M-1} \bigg[ \mathbb{H} \Big( \bm{\alpha}_{i_m}^{[1]} | \bm{U}_{i_m} \Big) -   \mathbb{H} \Big( \bm{\alpha}_{i_{m+1}}^{[l_{i_{m+1}}]}   |  \bm{U}_{i_m}  \Big) \bigg] +
\alpha_{i_M i_M}^{[1]}  T\log(P)  \\
%%%
\label{eq:converse_MISO_BC_4}
&  \; \dotleq \;  \sum_{m = 1}^{M-1} \max_{k \in \langle K \rangle} \left( \alpha_{i_m k}^{[1]} - \alpha_{i_{m+1} k}^{[l_{i_{m+1}}]}  \right)^{+} T\log(P)  +
\alpha_{i_M i_M}^{[1]}  T\log(P)  \\
%%%
\label{eq:converse_MISO_BC_5}
& \leq   \sum_{m = 1}^{M-1} \Big( \alpha_{i_m i_m}^{[1]} - \alpha_{i_{m+1} i_m}^{[l_{i_{m+1}}]}  \Big)T\log(P)   + \alpha_{i_M i_M}^{[1]} T\log(P).
%%% 
\end{align}
%%%
In \eqref{eq:converse_MISO_BC_2}, we used the fact $(\bm{W}_{i_{m+1}} , \bm{U}_{i_{m+1}} ) = \bm{U}_{i_{m}}$,
while the  upper bound in \eqref{eq:converse_MISO_BC_4} is obtained using Lemma \ref{lemma:AI_diff_enropies}.
%%%
On the other hand, \eqref{eq:converse_MISO_BC_5} holds due to the 
mc-SLS condition in \eqref{eq:SLS_cond_2}, which implies 
%%%
\begin{align}
\nonumber
\alpha^{[1]}_{i_m i_m}   - \alpha^{[l_{i_{m+1}}]}_{i_{m+1} i_m }  &  \geq  \alpha^{[1]}_{i_m k } 
- \alpha^{[l_{i_{m+1}}]}_{i_{m+1} k}  \\
\nonumber
\alpha^{[1]}_{i_m i_m}   - \alpha^{[l_{i_{m+1}}]}_{i_{m+1} i_m }   & \geq 0.
\end{align}
%%%
By combining  the bounds in  \eqref{eq:converse_MISO_BC_5} and \eqref{eq:converse_MISO_BC_1},
we obtain
%%%%
\begin{align}
\nonumber
T \! \! \! \! \sum_{(l_i,i) \in \{ \pi \}} \! \! \! \!  \bar{R}_{i}^{[l_{i}]}   &  \; \dotleq \;   
%%%
\sum_{m = 1}^{M -1 }    \big(    \alpha^{[l_{i_m } ]}_{i_m i_m}  -  \alpha^{[l_{i_{m+1} }]}_{i_{m+1} i_m } \big) T \log (P) 
 + \alpha_{i_M i_M}^{[l_{i_M}]} T\log(P)  \\ 
%%%%%
& =
\sum_{m = 1}^{M }    \big(    \alpha^{[l_{i_m } ]}_{i_m i_m}  -  \alpha^{[l_{i_{m+1} }]}_{i_{m+1} i_m } \big) T \log (P) 
+ \alpha_{i_1 i_M}^{[l_{i_1}]} T\log(P)   \\ 
%%%%%
& = \Delta_{\pi} T \log (P) + \alpha_{i_1 i_M}^{[l_{i_1}]} T\log(P).
 %%% 
\end{align}
%%%
This proves the $M$-th outer bound associated with cycle $\pi$.
%%%
To obtain the remaining $M-1$ bounds associated with the same cycle $\pi$, we follow the same steps while replacing $\pi$ with a shifted cycle $\pi'$, where $\pi'(m) = \pi(m+j)$ for some $j \in \langle M-1 \rangle$.
%%%%
We obtain the bound
%%%%
\begin{align}
\nonumber
T \! \! \! \! \sum_{(l_i,i) \in \{ \pi' \}} \! \! \! \!  \bar{R}_{i}^{[l_{i}]}   & \; \dotleq \; \Delta_{\pi'} T \log (P) + \alpha_{\sigma'(1) \sigma'(M)}^{[l_{\sigma'(1)}]}   T\log(P) \\
%%%
\label{eq:MISO_outer_shift_1}
&  =  \Delta_{\pi} T \log (P) +  \alpha_{\sigma(j+1) \sigma(j+M)}^{[l_{\sigma(j+1)}]}  T\log(P) \\
%%%
\label{eq:MISO_outer_shift_2}
&  =  \Delta_{\pi} T \log (P) +  \alpha_{i_{j+1} i_j}^{[l_{i_{j+1}}]}  T\log(P)
\end{align}
%%%%%
where \eqref{eq:MISO_outer_shift_1} holds since 
$\Delta_{\pi'}  = \Delta_{\pi} $ and $\sigma'(m) = \sigma(j+m) $.
%%%%
Repeating the same steps for all  $j$, we obtain all $M$ bounds associated with cycle $\pi$.
%%%%
The same can be done for all cycles $\pi \in \Pi$ with $|\pi| \geq 2$, from which we obtain the outer bound in Theorem \ref{theorem:MISO_BC_outerbound}. This concludes the proof.
%%%
\section{Optimality of Multi-Cell SLS in 2-Cell Networks}
\label{sec:2_cell_SLS}
%%%
Here we show that the MISO-BC outer bound in Theorem \ref{theorem:MISO_BC_outerbound} is tight for $2$-cell networks in the mc-SLS regime, hence proving the result in Theorem \ref{theorem:2_cell_SLS}.
%%%
In the $2$-cell case, the outer bound $\mathcal{D}_{\mathrm{out}}^{\mathrm{SLS}}$ is described by
all GDoF tuples $\mathbf{d} \in \mathbb{R}_{+}^{2K}$ with components satisfying
%%%
\begin{align}
\label{eq:MISO_outer_bound_2_cell_1}
 \bar{d}_{1}^{[l_1]}  & \leq \alpha_{11}^{[l_1]}  \\ 
%%%
\bar{d}_{2}^{[l_2]} & \leq \alpha_{22}^{[l_2]}  \\ 
%%%
\bar{d}_{1}^{[l_1]}  + \bar{d}_{2}^{[l_2]}  & \leq \alpha_{11}^{[l_1]} - \alpha_{21}^{[l_2]}  + 
\alpha_{22}^{[l_2]}   \\
%%%%
\label{eq:MISO_outer_bound_2_cell_3}
\bar{d}_{1}^{[l_1]} + \bar{d}_{2}^{[l_2]}  & \leq \alpha_{11}^{[l_1]} + 
\alpha_{22}^{[l_2]}  - \alpha_{12}^{[l_1]}
\end{align}
for all $l_1,l_2 \in \langle L \rangle$.
%%%
Recall that $\bar{d}_{i}^{[l_i]}$ denotes $\sum_{s_i = 1}^{l_i}d_{i}^{[s_i]}$.
%%%
To show that the outer bound $\mathcal{D}_{\mathrm{out}}^{\mathrm{SLS}}$ is achievable in the $2$-cell case, we consider a simplified SLS scheme and show that the corresponding achievable GDoF region coincides with the one described in \eqref{eq:MISO_outer_bound_2_cell_1}--\eqref{eq:MISO_outer_bound_2_cell_3}.
%%%
Before we proceed, we highlight that the set of all UEs, given by $\mathcal{U} = \big\{ (l_i, i) : l_i \in \langle L \rangle , i \in \langle 2 \rangle \big\}$ in this part, is partitioned as $\mathcal{U}_1 \cup \mathcal{U}_2$, where 
$\mathcal{U}_i$ denotes the set of UEs in cell $i$, for $i \in \langle 2 \rangle$.
%%%
\subsection{Simplified 2-Cell SLS}
\label{subsec:2_cell_SLS}
%%%
In the considered scheme, the message intended to UE-$(l_i,i)$ is split as $W_{i}^{[l_i]}  = \big( W_{\mathrm{s},i}^{[l_i]}  , W_{\mathrm{c},i}^{[l_i]}   \big)$, where $W_{\mathrm{s},i}^{[l_i]} $ is a single-cell sub-message, transmitted from BS-$i$ only; while 
$W_{\mathrm{c},i}$ is a multi-cell sub-message, transmitted in a cooperative fashion from both BS-$1$ and BS-$2$.
%%%
Each single-cell sub-message $W_{\mathrm{s},i}^{[l_i]}$ is encoded into a codeword $\bm{X}_{\mathrm{s},i}^{[l_i]}$. On the other hand,
%%%
the $2L$ multi-cell sub-messages $\big( W_{\mathrm{c},i}^{[l_i]} : (l_i,i) \in \mathcal{U}\big)$ are jointly encoded into the common codeword $\bm{X}_{\mathrm{c}}$. 
%%%
Codewords are independent, and each is drawn from a Gaussian codebook with unit average power.
%%%
Moreover, sub-messages $W_{\mathrm{s},i}^{[l_i]}$  and $W_{\mathrm{c},i}^{[l_i]}$ carry GDoFs of $d_{\mathrm{s},i}^{[l_i]}$  and $d_{\mathrm{c},i}^{[l_i]}$ respectively; and hence, the total GDoF achieved by UE-$(l_i,i)$ is given by a sum of two contributions as
$d_{i}^{[l_i]} = d_{\mathrm{s},i}^{[l_i]}+d_{\mathrm{c},i}^{[l_i]}$. 
%%%
GDoF tuples of single-cell and multi-cell sub-messages are given by 
$\mathbf{d}_{\mathrm{s}} \triangleq \big( d_{\mathrm{s},i}^{[l_i]} : (l_i,i) \in \mathcal{U} \big)$ and
$\mathbf{d}_{\mathrm{c}} \triangleq \big( d_{\mathrm{c},i}^{[l_i]} : (l_i,i) \in \mathcal{U} \big)$ respectively,
from which a GDoF tuple of user messages is given by $\mathbf{d} = \mathbf{d}_{\mathrm{s}}  + \mathbf{d}_{\mathrm{c}} $.

%%%
The transmit signal of BS-$i$ is composed as a  superposition of the $L$ designated single-cell codewords and the common codeword.
For a single use of the channel, this is given by:
%%%
\begin{equation}
X_{i}  = \sqrt{q_{\mathrm{c}} } X_{\mathrm{c}} +  \sum_{l_i \in \langle L \rangle } \sqrt{q_{i}^{[l_i]}  } X_{\mathrm{s},i}^{[l_i]}
\end{equation}
%%%
where $q_{\mathrm{c}} ,q_{i}^{[1]},\ldots, q_{i}^{[L]} $ are power control variables that satisfy $q_{\mathrm{c}} + \sum_{l_i \in \langle L \rangle} q_{i}^{[l_i]} \leq 1$.
%%%
For the purpose of GDoF analysis, we set the power control variables as follows:
%%%
\begin{equation}
\label{eq:power_allocation_2_cell_SLS}
q_{\mathrm{c}}   = \frac{1}{L+1} \ \ \text{and} \ \ q_{i}^{[l_i]}  =  \frac{1}{L+1} \cdot 
P^{r_i^{[l_i]}} 
\end{equation}
%%%
where the exponents $r_i^{[l_i]} \leq 0$, for all $(l_i,i) \in \mathcal{U}$, are power control variables on the GDoF scale.
%%%
The tuple of (GDoF scale) power control variables is given by $\mathbf{r} \triangleq \big( r_i^{[l_i]} : (l_i,i) \in \mathcal{U} \big)$.
%%%%

At the other end of the channel, each UE-$(l_i,i)$ successively decodes
$\bm{X}_{\mathrm{c}}, \bm{X}_{\mathrm{s},i}^{[1]} ,\ldots, \bm{X}_{\mathrm{s},i}^{[l_i]}$, in this specific order, while treating (intra- and inter-cell)  interference from all remaining codewords as additional Gaussian noise.
%%%
Note that in successive decoding, each decoded codeword is reconstructed, and its contribution is removed from the received signal before decoding the following codewords.  
%%%
We further note that $\bm{X}_{\mathrm{c}}$ is decoded by all UEs  in the network and it may be useful (in part) for each of them, depending on the GDoF allocation $\mathbf{d}_{\mathrm{c}}$. On the other hand, each codeword $\bm{X}_{\mathrm{s},i}^{[l_i]}$ is  decoded by UEs $(l_i,i), \ldots , (L,i)$ in cell $i$, and it is intended to UE-$(l_i,i)$ only.
%%%

%%%
For a given power control policy $\mathbf{r}$, the multi-cell sub-messages, carried through the common codeword $\bm{X}_{\mathrm{c}} $, achieve any GDoF tuple $\mathbf{d}_{\mathrm{c}} \in \mathbb{R}_{+}^{2K}$ with components satisfying
%%%
\begin{equation}
\label{eq:GDoF_common_2_cell_SLS}
\sum_{(l_i,i) \in \mathcal{U}} d_{\mathrm{c},i}^{[l_i]} \leq 
\min_{(l_i,i) \in \mathcal{U}} \left( \alpha_{ii}^{[l_i]} - \max\left\{ \alpha_{ii}^{[l_i]}  + 
\max_{s_i \in \langle L \rangle} \big\{r_{i}^{[s_i]} \big\} , \alpha_{ij}^{[l_i]}  + \max_{s_j \in \langle L \rangle } \big\{ r_{j}^{[s_j]} \big\} , 0 \right \} \right)^{+}.
\end{equation}
%%%
On the other hand, single-cell sub-messages achieve $\mathbf{d}_{\mathrm{s}} \in \mathbb{R}_{+}^{2K}$  with components satisfying
%%%%
\begin{equation}
\label{eq:GDoF_TIN_2_cell_SLS}
d_{\mathrm{s},i}^{[l_i]}   \leq \min_{s_i \geq l_i } \left( \alpha_{ii}^{[s_i]} + r_{i}^{[l_i]} - 
\max\left\{ \alpha_{ii}^{[s_i]}  + \max_{l_i' > l_i} \big\{r_{i}^{[l_i']} \big\}  , \alpha_{ij}^{[s_i]}  + \max_{l_j \in \langle L \rangle} \big\{r_{j}^{[l_j]} \big\} , 0 \right \} \right)^{+}.
\end{equation}
%%%%
In both \eqref{eq:GDoF_common_2_cell_SLS} and \eqref{eq:GDoF_TIN_2_cell_SLS}, it is implicitly assumed that $i \neq j$.
%%%%
Moreover, note that the successive decoding order specified earlier is reflected in \eqref{eq:GDoF_common_2_cell_SLS}
and \eqref{eq:GDoF_TIN_2_cell_SLS}.
%%%
\begin{remark}
%%%
The simplified SLS scheme presented above is a superposition of a mc-TIN scheme and an overlaying multicast codeword (also known as a common, or public, message). 
%%%
As seen from \eqref{eq:power_allocation_2_cell_SLS}, single-cell codewords, constituting the mc-TIN part, occupy lower power levels, and they are treated as noise by other-cell users.
%%%
On the other hand, the multi-cell multicast codeword occupies higher power levels, and it is decoded by all users in the network. 
%%%
This scheme is a special case of the more general mc-SLS scheme, described in Section \ref{subsubsec:scheme_SLS}.
%%%
On the other hand, the scheme may also be seen as a multi-cell extension (or generalization) of the rate-splitting schemes in \cite{Davoodi2018,Joudeh2016}.
%%%
In what follows, we refer to $\mathbf{d}_{\mathrm{s}}$ as the TIN contribution, and to $\mathbf{d}_{\mathrm{c}}$ as the multicast contribution.
%%%
\hfill $\lozenge$
%%%
\end{remark}
%%%
For a fixed power control policy $\mathbf{r}$, the set of all tuples $\mathbf{d}_{\mathrm{c}}$ that satisfy \eqref{eq:GDoF_common_2_cell_SLS} is denoted by $\mathcal{D}^{\mathrm{Mul}}(\mathbf{r})$; while $\mathcal{D}^{\mathrm{TIN}}(\mathbf{r})$ denotes the set of all tuples $\mathbf{d}_{\mathrm{s}}$ that satisfy \eqref{eq:GDoF_TIN_2_cell_SLS}.
%%%
The GDoF region achievable  through the proposed scheme is hence given by all tuples $\mathbf{d}$,
where $\mathbf{d} = \mathbf{d}_{\mathrm{c}} + \mathbf{d}_{\mathrm{s}}$
for some $\mathbf{d}_{\mathrm{c}} \in \mathcal{D}^{\mathrm{Mul}}(\mathbf{r})$, $\mathbf{d}_{\mathrm{s}} \in \mathcal{D}^{\mathrm{TIN}}(\mathbf{r})$, and $\mathbf{r} \leq \mathbf{0}$.
%%%
Denoting this achievable region by $\mathcal{D}^{\mathrm{SLS}} $, it follows that 
%%%
\begin{equation}
\label{eq:GDoF_region_SLS_2_cell}
\mathcal{D}^{\mathrm{SLS}} = \bigcup_{\mathbf{r} \leq \mathbf{0}} \mathcal{D}^{\mathrm{Mul}}(\mathbf{r}) \oplus \mathcal{D}^{\mathrm{TIN}}(\mathbf{r})
\end{equation}
%%%
where $\oplus$ is the Minkowski sum (or vector sum) operation (see Section \ref{subsec:notation}).\footnote{With a slight abuse of notation, we denote the achievable region here by $\mathcal{D}^{\mathrm{SLS}}$. Strictly speaking, however, the achievable region on the right-hand-side of \eqref{eq:GDoF_region_SLS_2_cell} is included in the SLS region $\mathcal{D}^{\mathrm{SLS}}$, described in Section \ref{subsubsec:scheme_SLS}.}
%%%

%%%
Achievable GDoF tuples in \eqref{eq:GDoF_region_SLS_2_cell} are highly coupled with auxiliary design variables. For instance, each $\mathbf{d} \in \mathcal{D}^{\mathrm{SLS}} $ is determined by a multicast contribution $\mathbf{d}_{\mathrm{c}}$ and a TIN contribution $\mathbf{d}_{\mathrm{s}}$, which in turn depend on the power control tuple $\mathbf{r}$. Together, these auxiliary variables are jointly optimized to achieve different GDoF trade-offs (or tuples $\mathbf{d}$). 
%%%
On the other hand, the outer bound $\mathcal{D}^{\mathrm{MBC}}_{\mathrm{out}}$ is described in terms of fixed channel parameters only (i.e. $\bm{\alpha}$) with no auxiliary variables, as seen in \eqref{eq:MISO_outer_bound_2_cell_1}--\eqref{eq:MISO_outer_bound_2_cell_3}.
%%%
A direct comparison between $\mathcal{D}^{\mathrm{MBC}}_{\mathrm{out}}$  and $\mathcal{D}^{\mathrm{SLS}} $ can be facilitated by eliminating the auxiliary variables in the latter, which is typically accomplished by means of Fourier-Motzkin (FM) elimination.
%%% 
Nevertheless, the intricate structure of $\mathcal{D}^{\mathrm{SLS}}$ and the high number of auxiliary variables prohibit a direct application of the FM procedure in this case.
%%% 
To circumvent this technical challenge, we devise a sequence of reductions, which are summarized as follows:
%%%
\begin{itemize}
\item We first restrict the space of admissible power control variables $\mathbf{r}$, which in turn allows us to obtain simpler inner bounds for $\mathcal{D}^{\mathrm{Mul}}(\mathbf{r})$ and $ \mathcal{D}^{\mathrm{TIN}}(\mathbf{r})$, denoted by $\mathcal{D}^{\mathrm{Mul}}(a)$  and $\mathcal{D}^{\mathrm{PTIN}}(a)$, respectively. These inner bounds  depend on a single power control variable $a$ only. As we will see,  $\mathcal{D}^{\mathrm{PTIN}}(a)$ is obtained through a new application of the potential graph approach \cite{Geng2015}, a structured FM  procedure tailored for GDoF regions achievable through TIN.  
\item We then evaluate the Minkowski sum
$\mathcal{D}^{\mathrm{Mul}}(a) \oplus \mathcal{D}^{\mathrm{PTIN}}(a)$, for any feasible $a$, in terms of inequalities that bound $\mathbf{d}$ only by eliminating the auxiliary GDoF variables $\mathbf{d}_{\mathrm{c}}$  and $\mathbf{d}_{\mathrm{s}}$.
%%%
This is accomplished by exploiting polymatroid properties of these inner bounds, which allows us to utilize a result on the Minkowski sum of polymatroids \cite[Th. 44.6]{Schrijver2002}. 
%%%
At this point, we obtain an inner bound for $\mathcal{D}^{\mathrm{SLS}}$ with $a$ as the only remaining auxiliary variable.
%%%
\item We then eliminate the remaining variable $a$ via a second round of FM elimination.
%%%
Remarkably, the region obtained after this step matches the outer bound in \eqref{eq:MISO_outer_bound_2_cell_1}--\eqref{eq:MISO_outer_bound_2_cell_3}.
%%%
\end{itemize}
%%%
The remainder of this section is dedicated to explaining the above steps in more detail. 
%%%%
\subsection{Optimality of 2-Cell SLS}
\label{subsec:optimality_of_2_cell_SLS}
%%%%
As a first simplifying step, we restrict the space of feasible power control variables $\mathbf{r}$ such that 
for both $i \in \langle 2 \rangle$, we have the following order
%%%
\begin{align}
\label{eq:power_control_order}
r_{i}^{[L]} \leq r_i^{[L-1]} \leq \cdots \leq r_i^{[1]} \leq -a \\
\label{eq:a_LB_UB}
 \text{where}  \ \ 0 \leq a \leq \max_{i,j,l_i: i\neq j} \big\{   \alpha_{ij}^{[l_i]} \big\}.
\end{align}
%%%
As it turns out, imposing the power allocation order in \eqref{eq:power_control_order} incurs no loss of generality in the regimes of interest. This is due to the SIR order in \eqref{eq:SIR_order}, which holds in the mc-SLS regime, requiring 
higher power levels for weaker users. 
%%%
On the other hand, it is sufficient to consider values of $a$ which are no greater than the strongest interfering link as in \eqref{eq:a_LB_UB}, since this level of attenuation guarantees that all inter-cell interference from the mc-TIN scheme is received below the  noise level.  

%%%
Given the power allocation order in \eqref{eq:power_control_order}, we may bound the right-hand-side of \eqref{eq:GDoF_common_2_cell_SLS}  below as
%%%
%%%
\begin{align}
\nonumber
\min_{(l_i,i) \in \mathcal{U}} & \left( \alpha_{ii}^{[l_i]} - \max\left\{ \alpha_{ii}^{[l_i]}  + 
\max_{s_i \in \langle L \rangle} \big\{r_{i}^{[s_i]} \big\} , \alpha_{ij}^{[l_i]}  + \max_{s_j \in \langle L \rangle } \big\{ r_{j}^{[s_j]} \big\} , 0 \right \} \right)^{+} \\
%%%
\label{eq:GDoF_common_2_cell_SLS_1}
& \geq \min_{(l_i,i) \in \mathcal{U}} \left( \alpha_{ii}^{[l_i]} - 
\max\left\{ \alpha_{ii}^{[l_i]}  - a , \alpha_{ij}^{[l_i]}  - a  , 0 \right \} \right)^{+} \\
%%%
& = \min_{(l_i,i) \in \mathcal{U}} \left(  
\min\left\{ a , \alpha_{ii}^{[l_i]} - \alpha_{ij}^{[l_i]}  + a  , \alpha_{ii}^{[l_i]} \right \} \right)^{+} \\
%%%
\label{eq:GDoF_common_2_cell_SLS_2}
& = a
\end{align}
%%%
where the equality in \eqref{eq:GDoF_common_2_cell_SLS_2} holds due to $a \geq 0$, $\alpha_{ii}^{[l_i]} \geq a$ and $\alpha_{ii}^{[l_i]} \geq \alpha_{ij}^{[l_i]}$.
%%%
It follows that $\mathcal{D}^{\mathrm{Mul}} (\mathbf{r})$ includes the set of GDoF tuples  $\mathbf{d}_{\mathrm{c}} \in \mathbb{R}_{+}^{2K}$ with components satisfying
%%%
\begin{equation}
\label{eq:GDoF_common_2_cell_SLS_simplified}
\sum_{(l_i,i) \in \mathcal{U}} d_{\mathrm{c},i}^{[l_i]} \leq a.
\end{equation}
%%%
For any $a$ satisfying \eqref{eq:a_LB_UB}, and with a slight abuse of notation, we denote the set of GDoF tuples satisfying \eqref{eq:GDoF_common_2_cell_SLS_simplified} as $\mathcal{D}^{\mathrm{Mul}} (a)$.
%%%
This leads to an inner bound for $\mathcal{D}^{\mathrm{SLS}} $ in \eqref{eq:GDoF_region_SLS_2_cell}, given by 
%%%
\begin{align}
\nonumber
\mathcal{D}^{\mathrm{SLS}}   &  \supseteq  \bigcup_{a} \bigcup_{\mathbf{r} \in \mathcal{R}(a)} \mathcal{D}^{\mathrm{Mul}}(a) \oplus \mathcal{D}^{\mathrm{TIN}}(\mathbf{r}) \\
%%%
\label{eq:GDoF_inner_SLS_2_cell} 
& =  \bigcup_{a}  \left(  \mathcal{D}^{\mathrm{Mul}}(a) \oplus  \bigcup_{\mathbf{r} \in \mathcal{R}(a)  }   \mathcal{D}^{\mathrm{TIN}}(\mathbf{r}) \right)
\end{align}
%%%
where $\mathcal{R}(a) $ is the set of power control tuples $\mathbf{r}$ satisfying \eqref{eq:power_control_order} for fixed $a$, while the union with respect to $a$ is taken over the interval specified in \eqref{eq:a_LB_UB}.
%%%
The equality in \eqref{eq:GDoF_inner_SLS_2_cell} holds due to the fact that $(\mathcal{A} \oplus \mathcal{B}_{1}) \cup (\mathcal{A} \oplus \mathcal{B}_{2}) = \mathcal{A} \oplus (\mathcal{B}_{1} \cup \mathcal{B}_{2})$.
%%%%
Next, we turn to obtaining a simplified inner bound for the TIN region $\cup_{\mathbf{r} \in \mathcal{R}(a)  }   \mathcal{D}^{\mathrm{TIN}}(\mathbf{r})$, for any fixed $a$ in the designated interval \eqref{eq:a_LB_UB}.
%%%

%%%
The TIN region $\mathcal{D}^{\mathrm{TIN}} (\mathbf{r})$
includes a smaller region known as the polyhedral-TIN region, denoted by $\mathcal{D}^{\mathrm{PTIN}} (\mathbf{r})$, obtained by relaxing the $(\cdot)^{+}$ operation on the right-hand-side of \eqref{eq:GDoF_TIN_2_cell_SLS}, hence posing further restrictions on admissible power control tuples---see, e.g., \cite{Geng2015,Joudeh2019b}.
%%%
Under the order in \eqref{eq:power_control_order},  $\mathcal{D}^{\mathrm{PTIN}} (\mathbf{r})$
is  described by the set of GDoF tuples  $\mathbf{d}_{\mathrm{s}} \in \mathbb{R}_{+}^{2K}$ with components satisfying
%%%%
\begin{align}
d_{\mathrm{s},i}^{[l_i]}   & \leq  \min_{s_i \geq l_i }
\left\{ \alpha_{ii}^{[s_i]} + r_{i}^{[l_i]} - 
\max\left\{ \alpha_{ii}^{[s_i]}  + \max_{l_i' > l_i} \big\{r_{i}^{[l_i']} \big\}  , \alpha_{ij}^{[s_i]}  + \max_{l_j \in \langle L \rangle} \big\{r_{j}^{[l_j]} \big\} , 0 \right \} \right\} \\
%%%
\label{eq:GDoF_TIN_2_cell_SLS_1}
& = \min_{s_i \geq l_i }
\left\{ r_{i}^{[l_i]}   - \max_{l_i' > l_i} \big\{r_{i}^{[l_i']} \big\}  , \alpha_{ii}^{[s_i]} - \alpha_{ij}^{[s_i]}  + r_{i}^{[l_i]}  - \max_{l_j \in \langle L \rangle} \big\{r_{j}^{[l_j]} \big\} , \alpha_{ii}^{[s_i]} + r_{i}^{[l_i]}  \right \} \\
%%%%
\label{eq:GDoF_TIN_2_cell_SLS_2}
& = \min_{s_i \geq l_i } \left\{ r_{i}^{[l_i]}   - r_{i}^{[l_i + 1]}  , \alpha_{ii}^{[s_i]} - \alpha_{ij}^{[s_i]}  + r_{i}^{[l_i]}  - r_{j}^{[1]} , \alpha_{ii}^{[s_i]} + r_{i}^{[l_i]}  \right \}  \\
%%%%
\label{eq:GDoF_TIN_2_cell_SLS_3}
& = \min \left\{ r_{i}^{[l_i]}   - r_{i}^{[l_i + 1]}  , \alpha_{ii}^{[l_i]} - \alpha_{ij}^{[l_i]}  + r_{i}^{[l_i]}  - r_{j}^{[1]} , \alpha_{ii}^{[l_i]} + r_{i}^{[l_i]}  \right \}
\end{align}
%%%%
where in the above, we set $r_{i}^{[L+1]} = - \infty$.
%%%
Note that \eqref{eq:GDoF_TIN_2_cell_SLS_2} is obtained from \eqref{eq:GDoF_TIN_2_cell_SLS_1} by invoking the power allocation order in \eqref{eq:power_control_order}, while the equality in \eqref{eq:GDoF_TIN_2_cell_SLS_3} holds due to the SIR order in \eqref{eq:SIR_order} and the SNR order in \eqref{eq:strength_order}, which imply  
$\alpha_{ii}^{[l_i]} - \alpha_{ij}^{[l_i]}  \leq \alpha_{ii}^{[s_i]} - \alpha_{ij}^{[s_i]} $ and $\alpha_{ii}^{[l_i]} \leq \alpha_{ii}^{[s_i]}$, for all 
$s_i \in \langle l_i: L \rangle$.
%%%
For a given $a$ satisfying \eqref{eq:a_LB_UB}, and with another slight abuse of notation, we use 
$\mathcal{D}^{\mathrm{PTIN}} (a)$ to denote $\cup_{\mathbf{r} \in \mathcal{R}(a)  }   \mathcal{D}^{\mathrm{PTIN}}(\mathbf{r})$.
%%%
As it turns out, the region $\mathcal{D}^{\mathrm{PTIN}} (a)$
lends itself to an efficient FM elimination procedure using the potential graph approach \cite{Geng2015}, yielding the following result.
%%%
\begin{lemma}
\label{lemma:PTIN_region_2_cell}
For any $0 \leq a \leq \max_{i,j,l_i: i\neq j} \big\{   \alpha_{ij}^{[l_i]} \big\}$, 
the polyehdral TIN region $\mathcal{D}^{\mathrm{PTIN}} (a)$ is equal to the region described by all tuples
$\mathbf{d}_{\mathrm{s}} \in \mathbb{R}_{+}^{2K}$ with components satisfying
%%%
\begin{align}
\label{eq:PTIN_2_cell_1}
\bar{d}_{\mathrm{s},i}^{[l_i]} & \leq \alpha_{ii}^{[l_i]}  -  a \\ 
\label{eq:PTIN_2_cell_2}
\bar{d}_{\mathrm{s},i}^{[l_i]} + \bar{d}_{\mathrm{s},j}^{[l_j]}  & \leq \alpha_{ii}^{[l_i]} + 
\alpha_{jj}^{[l_j]}  - \max\left\{\alpha_{ij}^{[l_i]}+a, \alpha_{ji}^{[l_j]} +a, 2a , \alpha_{ij}^{[l_i]} + \alpha_{ji}^{[l_j]}   \right\}
\end{align}
for all $i,j \in \langle 2 \rangle$, $i \neq j$, and $l_i, l_j \in \langle L \rangle$.
%%%
\end{lemma}
%%%
The potential graph approach, used to prove  Lemma \ref{lemma:PTIN_region_2_cell}, constructs a directed graph and maps power control variables to its vertices and GDoF inequalities to its edges.  
%%%
A potential theorem \cite[Th. 8.2]{Schrijver2002} from combinatorial optimization is then invoked to derive equivalent GDoF inequalities from the directed graph, which do not include power control variables. This approach essentially exploits the structure of polyhedral-TIN  GDoF regions to carry out FM elimination of power control variables in an efficient manner. In proving Lemma \ref{lemma:PTIN_region_2_cell}, we extend the original approach in \cite{Geng2015} to deal with the region $\mathcal{D}^{\mathrm{PTIN}} (a)$. A detailed proof of Lemma \ref{lemma:PTIN_region_2_cell} is relegated to Appendix \ref{appendix:PTIN_region_2_cell}.
%%%

%%%
After the elimination of $\mathbf{r}$ in Lemma \ref{lemma:PTIN_region_2_cell}, we now have an inner 
bound for $\mathcal{D}^{\mathrm{SLS}} $ given by
%%%
\begin{align}
\label{eq:GDoF_inner_SLS_2_cell_2} 
\mathcal{D}^{\mathrm{SLS}}    \supseteq  \bigcup_{a}  \mathcal{D}^{\mathrm{Mul}}(a) \oplus   \mathcal{D}^{\mathrm{PTIN}}(a).
\end{align}
%%%
This inner bound does not depend on $\mathbf{r}$, yet it is still characterized in terms of $a$, $\mathbf{d}_{\mathrm{c}}$ and $\mathbf{d}_{\mathrm{s}}$.
%%%

%%%
Next, we eliminate the auxiliary GDoF tuples $\mathbf{d}_{\mathrm{c}}$ and $\mathbf{d}_{\mathrm{s}}$ by  characterizing the Minkowski sum $\mathcal{D}^{\mathrm{Mul}}(a) \oplus  \mathcal{D}^{\mathrm{PTIN}}(a)$ in terms of inequalities that bound $\mathbf{d} = \mathbf{d}_{\mathrm{c}} + \mathbf{d}_{\mathrm{s}} $, for any fixed $a$.
%%%
To this end, it is easier to work with an inner bound for $\mathcal{D}^{\mathrm{PTIN}}(a)$, denoted by $\mathcal{D}^{\mathrm{PTIN}'}(a)$, given by
%%%
\begin{align}
\label{eq:PTIN_2_cell_2_1}
\bar{d}_{\mathrm{s},i}^{[l_i]} & \leq \alpha_{ii}^{[l_i]}  -  \max\left\{\alpha_{ij}^{[l_i]}, a \right\} \\ 
\label{eq:PTIN_2_cell_2_2}
\bar{d}_{\mathrm{s},i}^{[l_i]} + \bar{d}_{\mathrm{s},j}^{[l_j]}  & \leq \alpha_{ii}^{[l_i]} + 
\alpha_{jj}^{[l_j]}  - \max\left\{\alpha_{ij}^{[l_i]}+a, \alpha_{ji}^{[l_j]} +a, 2a , \alpha_{ij}^{[l_i]} + 
\alpha_{ji}^{[l_j]}   \right\}.
\end{align}
%%%
Note that $\mathcal{D}^{\mathrm{PTIN}}(a)$ and  $\mathcal{D}^{\mathrm{PTIN}'}(a)$ are almost identical, with the exception that \eqref{eq:PTIN_2_cell_2_1} in the latter is tighter than the corresponding inequality  \eqref{eq:PTIN_2_cell_1} in the former.
%%%
The rationale behind this tightening will become clear further on.
%%%
We now observe that $\mathcal{D}^{\mathrm{Mul}}(a)$ may be described  as
%%%
\begin{align}
\label{eq:common_2_cell_2_1}
\bar{d}_{\mathrm{c},i}^{[l_i]} & \leq a\\
\label{eq:common_2_cell_2_2}
\bar{d}_{\mathrm{c},i}^{[l_i]} + \bar{d}_{\mathrm{c},j}^{[l_j]}  & \leq a
\end{align}
%%%
for all $i,j \in \langle 2 \rangle$, $i \neq j$, and $l_i, l_j \in \langle L \rangle$.
%%%
This is obtained directly from \eqref{eq:GDoF_common_2_cell_SLS_simplified} by including redundant bounds 
so that the linear inequalities describing $\mathcal{D}^{\mathrm{Mul}}(a)$ in \eqref{eq:common_2_cell_2_1} and \eqref{eq:common_2_cell_2_2} are of the same type as those used to describe $\mathcal{D}^{\mathrm{PTIN}'}(a)$ in \eqref{eq:PTIN_2_cell_2_1} and \eqref{eq:PTIN_2_cell_2_2}.
%%%
This leads us to the following result.
%%% 
\begin{lemma}
\label{lemma:Minkowski_sum}
For any $0 \leq a \leq \max_{i,j,l_i: i\neq j} \big\{   \alpha_{ij}^{[l_i]} \big\}$, the Minkowski sum $\mathcal{D}^{\mathrm{Mul}}(a) \oplus \mathcal{D}^{\mathrm{PTIN}'}(a)$ is characterized by all tuples
$\mathbf{d} \in \mathbb{R}_{+}^{2K}$ with components satisfying
%%%
\begin{align}
\label{eq:Minkowski_sum_1}
\bar{d}_{i}^{[l_i]} & \leq \alpha_{ii}^{[l_i]}  -  \max \left\{\alpha_{ij}^{[l_i]} - a, 0 \right\} \\ 
%%%
\label{eq:Minkowski_sum_2}
\bar{d}_{i}^{[l_i]} + \bar{d}_{j}^{[l_j]}  & \leq \alpha_{ii}^{[l_i]} + 
\alpha_{jj}^{[l_j]}  - \max\left\{\alpha_{ij}^{[l_i]}, \alpha_{ji}^{[l_j]}, a , \alpha_{ij}^{[l_i]} + \alpha_{ji}^{[l_j]} - a   \right\}
\end{align}
%%%
for all $i,j \in \langle 2 \rangle$, $i \neq j$, and $l_i, l_j \in \langle L \rangle$.
%%%
\end{lemma}
%%%
It is perhaps clear from Lemma \ref{lemma:Minkowski_sum} that the linear inequalities that describe  $\mathcal{D}^{\mathrm{Mul}}(a) \oplus \mathcal{D}^{\mathrm{PTIN}'}(a)$ are simply the direct sums of the corresponding inequalities describing the constituent polyhedra $\mathcal{D}^{\mathrm{Mul}}(a)$ and  
$\mathcal{D}^{\mathrm{PTIN}'}(a)$. That is, \eqref{eq:Minkowski_sum_1} is obtained by adding \eqref{eq:PTIN_2_cell_2_1} and \eqref{eq:common_2_cell_2_1}; and \eqref{eq:Minkowski_sum_2} is obtained by adding  \eqref{eq:PTIN_2_cell_2_2} and \eqref{eq:common_2_cell_2_2}.
%%%
Since summing inequalities loosens them in general, one can directly conclude that 
$\mathcal{D}^{\mathrm{Mul}}(a) \oplus \mathcal{D}^{\mathrm{PTIN}'}(a)$ is included in, yet not necessarily equal to, the region described by \eqref{eq:Minkowski_sum_1} and \eqref{eq:Minkowski_sum_2} in Lemma \ref{lemma:Minkowski_sum}.
%%% 
Remarkably, it turns out that in this special case $\mathcal{D}^{\mathrm{Mul}}(a) \oplus \mathcal{D}^{\mathrm{PTIN}'}(a)$ is equal to the polyhedron described by \eqref{eq:Minkowski_sum_1} and \eqref{eq:Minkowski_sum_2}. 
%%%
This holds since both $\mathcal{D}^{\mathrm{Mul}}(a)$ and $\mathcal{D}^{\mathrm{PTIN}'}(a)$ are polymatroids (see Appendix \ref{appendix:Minkowski_sum}); and polymatroids have the nice property that their Minkowski sums are given by the direct sums of their corresponding linear inequalities \cite[Th. 44.6]{Schrijver2002}  (see also \cite[Th. 3]{McDiarmid1975}).
%%%
This point, and the proof of Lemma \ref{lemma:Minkowski_sum}, are discussed in detail in Appendix \ref{appendix:Minkowski_sum}.
%%%%

%%%%
In light of the above, going from $\mathcal{D}^{\mathrm{PTIN}}(a)$ to the smaller region $\mathcal{D}^{\mathrm{PTIN}'}(a)$ can now be explained. 
%%%
In particular, this step guarantees that the set function associated with the polyhedron in \eqref{eq:PTIN_2_cell_2_1} and \eqref{eq:PTIN_2_cell_2_2}, defined over the ground set $\mathcal{U}$, is non-decreasing. 
%%%
This monotonicity, alongside submodularity, imply that the region $\mathcal{D}^{\mathrm{PTIN}'}(a)$ is a polymatroid, as shown in Appendix \ref{appendix:Minkowski_sum}.
%%%

%%%
Building upon the result in Lemma \ref{lemma:Minkowski_sum}, it follows that the inner bound given by
%%%
\begin{align}
\label{eq:inner_region_a_0}
 \bigcup_{a}  \mathcal{D}^{\mathrm{Mul}}(a) \oplus   \mathcal{D}^{\mathrm{PTIN}'}(a) 
\end{align}
%%%
is described by all tuples $\mathbf{d} \in \mathbb{R}_{+}^{2K}$  that satisfy
%%%
\begin{align}
\label{eq:inner_region_a_1}
\bar{d}_{i}^{[l_i]} & \leq \alpha_{ii}^{[l_i]}  -  \max \left\{\alpha_{ij}^{[l_i]} - a, 0 \right\} \\ 
\label{eq:inner_region_a_2}
\bar{d}_{i}^{[l_i]} + \bar{d}_{j}^{[l_j]}  & \leq \alpha_{ii}^{[l_i]} + 
\alpha_{jj}^{[l_j]}  - \max\left\{\alpha_{ij}^{[l_i]}, \alpha_{ji}^{[l_j]}, a , \alpha_{ij}^{[l_i]} + \alpha_{ji}^{[l_j]} - a   \right\} 
\\
%%%
\label{eq:inner_region_a_3}
0 \leq a & \leq \max_{i,j,l_i: i\neq j} \alpha_{ij}^{[l_i]}.
\end{align}
%%%
for all $i,j \in \langle 2 \rangle$, $i \neq j$, and $l_i, l_j \in \langle L \rangle$.
%%%
Now it remains to eliminate the last auxiliary variable $a$.
%%%
This is accomplished by a standard application of the FM procedure \cite[Appendix D]{ElGamal2011}.
%%%

%%%
To this end, we classify the inequalities in \eqref{eq:inner_region_a_1}--\eqref{eq:inner_region_a_3} with respect to the presence and sign of the variable $a$ on the right-hand-side. We have the three following classes:
%%%
\begin{itemize}
\item Inequalities with no $a$:
%%%
\begin{align}
\label{eq:FM_elim_1_1}
\bar{d}_{i}^{[l_i]} & \leq \alpha_{ii}^{[l_i]}  \\ 
%%%
\label{eq:FM_elim_1_2}
\bar{d}_{i}^{[l_i]} + \bar{d}_{j}^{[l_j]}  & \leq \alpha_{ii}^{[l_i]} + 
\alpha_{jj}^{[l_j]}  - \max\left\{\alpha_{ij}^{[l_i]}, \alpha_{ji}^{[l_j]}  \right\}.
\end{align}
%%%
\item Inequalities with $+a$:
%%%
\begin{align}
\label{eq:FM_elim_2_1}
\bar{d}_{i}^{[l_i]} & \leq \alpha_{ii}^{[l_i]}  - \alpha_{ij}^{[l_i]} + a \\ 
%%%
\label{eq:FM_elim_2_2}
\bar{d}_{i}^{[l_i]} + \bar{d}_{j}^{[l_j]}  & \leq \alpha_{ii}^{[l_i]} + 
\alpha_{jj}^{[l_j]}  - \left( \alpha_{ij}^{[l_i]} + \alpha_{ji}^{[l_j]}  \right) + a   \\
%%%
\label{eq:FM_elim_2_3}
0 & \leq a.
%%%
\end{align}
%%%
\item Inequalities with $-a$:
%%%
\begin{align}
%%%
\label{eq:FM_elim_3_1}
\bar{d}_{i}^{[l_i]} + \bar{d}_{j}^{[l_j]}  & \leq \alpha_{ii}^{[l_i]} + 
\alpha_{jj}^{[l_j]}  -  a   \\
%%%
\label{eq:FM_elim_3_2}
0 & \leq \max_{i,j,l_i: i\neq j} \alpha_{ij}^{[l_i]} - a.
\end{align}
%%%
\end{itemize}
%%%
To eliminate $a$, we add each inequality with $+a$ to every inequality with $-a$.
%%%
We start by adding an arbitrary inequality from \eqref{eq:FM_elim_3_1}, given by 
$\bar{d}_{i}^{[l_i']} + \bar{d}_{j}^{[l_j']}   \leq \alpha_{ii}^{[l_i']} + 
\alpha_{jj}^{[l_j']}  -  a $, to all inequalities in \eqref{eq:FM_elim_2_1}--\eqref{eq:FM_elim_2_3}.
%%%
We obtain the following set of inequalities 
%%%
\begin{align}
\label{eq:FM_elim_21_1}
\bar{d}_{i}^{[l_i']} + \bar{d}_{j}^{[l_j']}  + \bar{d}_{i}^{[l_i]} & \leq \alpha_{ii}^{[l_i']} + 
\alpha_{jj}^{[l_j']}  + \alpha_{ii}^{[l_i]}  - \alpha_{ij}^{[l_i]} \\ 
%%%
\label{eq:FM_elim_21_2}
\bar{d}_{i}^{[l_i']} + \bar{d}_{j}^{[l_j']}  + \bar{d}_{i}^{[l_i]} + \bar{d}_{j}^{[l_j]}  & \leq \alpha_{ii}^{[l_i']} + 
\alpha_{jj}^{[l_j']}  + \alpha_{ii}^{[l_i]} + 
\alpha_{jj}^{[l_j]}  - \left( \alpha_{ij}^{[l_i]} + \alpha_{ji}^{[l_j]}  \right)   \\
%%%
\label{eq:FM_elim_21_3}
\bar{d}_{i}^{[l_i']} + \bar{d}_{j}^{[l_j']}  & \leq \alpha_{ii}^{[l_i']} + 
\alpha_{jj}^{[l_j']}.
%%%
\end{align}
%%%
Next, we show that these resulting inequalities in \eqref{eq:FM_elim_21_1}--\eqref{eq:FM_elim_21_3} are all redundant, as they are implied by the set of inequalities in \eqref{eq:FM_elim_1_1} and \eqref{eq:FM_elim_1_2}.
%%%
In particular, from  \eqref{eq:FM_elim_1_1} and \eqref{eq:FM_elim_1_2}, we obtain
%%%
%%%
\begin{align}
\label{eq:FM_elim_11_1}
\bar{d}_{i}^{[l_i']} & \leq \alpha_{ii}^{[l_i']}  \\ 
%%%
\label{eq:FM_elim_11_2}
\bar{d}_{j}^{[l_j']} & \leq \alpha_{ii}^{[l_j']}  \\ 
%%%
\label{eq:FM_elim_11_3}
\bar{d}_{i}^{[l_i]} + \bar{d}_{j}^{[l_j']}   & \leq \alpha_{ii}^{[l_i]} + 
\alpha_{jj}^{[l_j']}  - \alpha_{ij}^{[l_i]} \\
%%%
\label{eq:FM_elim_11_4}
\bar{d}_{i}^{[l_i']} + \bar{d}_{j}^{[l_j]}   & \leq \alpha_{ii}^{[l_i']} + 
\alpha_{jj}^{[l_j]}  - \alpha_{ji}^{[l_j]}
\end{align}
%%%
from which it is clear that \eqref{eq:FM_elim_21_1} is implied by \eqref{eq:FM_elim_11_1} and \eqref{eq:FM_elim_11_3}, i.e. by summing them; \eqref{eq:FM_elim_21_2} is implied by \eqref{eq:FM_elim_11_3} and \eqref{eq:FM_elim_11_4};
and \eqref{eq:FM_elim_21_3} is implied by \eqref{eq:FM_elim_11_1} and \eqref{eq:FM_elim_11_2}.
%%%

%%%
We move on to the next step of the FM elimination, where we add the inequality in \eqref{eq:FM_elim_3_2} to all inequalities in \eqref{eq:FM_elim_2_1}--\eqref{eq:FM_elim_2_3}, from which we obtain
%%%
\begin{align}
\label{eq:FM_elim_22_1}
\bar{d}_{i}^{[l_i]} & \leq \alpha_{ii}^{[l_i]}  - \alpha_{ij}^{[l_i]} + \max_{i,j,l_i: i\neq j} \alpha_{ij}^{[l_i]}  \\ 
%%%
\label{eq:FM_elim_22_2}
\bar{d}_{i}^{[l_i]} + \bar{d}_{j}^{[l_j]}  & \leq \alpha_{ii}^{[l_i]} + 
\alpha_{jj}^{[l_j]}  - \left( \alpha_{ij}^{[l_i]} + \alpha_{ji}^{[l_j]}  \right) + \max_{i,j,l_i: i\neq j} \alpha_{ij}^{[l_i]}    \\
%%%
\label{eq:FM_elim_22_3}
0 & \leq \max_{i,j,l_i: i\neq j} \alpha_{ij}^{[l_i]}.
%%%
\end{align}
%%%
It is evident that \eqref{eq:FM_elim_22_1} and  \eqref{eq:FM_elim_22_2} are redundant, as they are implied by \eqref{eq:FM_elim_1_1} and  \eqref{eq:FM_elim_1_2} respectively.  Moreover, \eqref{eq:FM_elim_22_3} holds by definition of channel strength parameters.
%%%

From the above, it follows that all inequalities obtained by eliminating $a$ in \eqref{eq:FM_elim_2_1}--\eqref{eq:FM_elim_3_2} are redundant with respect to \eqref{eq:FM_elim_1_1} and \eqref{eq:FM_elim_1_2}.
%%%%
Therefore, the inner bound  in \eqref{eq:inner_region_a_0} is equal to the region specified by all GDoF tuples $\mathbf{d} \in \mathbb{R}_{+}^{2K}$  that satisfy
%%%
\begin{align}
\bar{d}_{i}^{[l_i]} & \leq \alpha_{ii}^{[l_i]}  \\ 
%%%
\bar{d}_{i}^{[l_i]} + \bar{d}_{j}^{[l_j]}  & \leq \alpha_{ii}^{[l_i]} + 
\alpha_{jj}^{[l_j]}  - \max\left\{\alpha_{ij}^{[l_i]}, \alpha_{ji}^{[l_j]}  \right\}
\end{align}
%%%
for all $i,j \in \langle 2 \rangle$, $i \neq j$, and $l_i, l_j \in \langle L \rangle$.
%%%
This exactly matches the $2$-cell outer bound in  \eqref{eq:MISO_outer_bound_2_cell_1}--\eqref{eq:MISO_outer_bound_2_cell_3}, which in turn concludes the proof of Theorem \ref{theorem:2_cell_SLS}.
%%%
\section{Extremal Gains of Multi-Cell Cooperation over Multi-Cell TIN}
\label{sec:extremal_gains}
%%%
In this section, we present a proof for Theorem \ref{theorem: extremal gain}.
%%%%
The main idea of the proof is to exploit the redundancy of weaker UEs in each cell of the network, exhibited in the regimes of interest due to the SIR order in \eqref{eq:SIR_order}, which allows us to form a direct relationship between $K \times KL$ networks and $K \times K$  networks. This in turn enables us to utilize previous extremal results in \cite{Chan2020}.
%%%%

%%%%
Similar to the sum-GDoF definitions in \eqref{eq:sum_GDoF_TIN} and \eqref{eq:sum_GDoF_MBC}, we define 
a sum-GDoF outer bound for the MISO-BC in the mc-SLS regime from the outer bound region in Theorem \ref{theorem:MISO_BC_outerbound} as 
%%% 
\begin{equation}
\label{eq:sum_GDoF_outer_MISO_BC}
d^{\mathrm{MBC}}_{\mathrm{out},\Sigma}(\bm{\alpha} )  \triangleq \max_{\mathbf{d} \in \mathcal{D}^{\mathrm{SLS}}_{\mathrm{out}} (\bm{\alpha} )}
\sum_{(l_k,k) \in \mathcal{U}} d_{k}^{[l_k]}.
\end{equation}
%%%
Moreover, for an arbitrary network $\bm{\alpha} \in \mathbb{R}^{K \times K \times L}_{+}$, we use $\bm{\alpha}^{[L]} \in \mathbb{R}^{K \times K}_{+}$ to denote the sub-network  obtained by keeping the strongest UE (in the SNR or SIR sense) in each cell, and  eliminating all remaining UEs.  Therefore, the $(i,j)$-th element of $\bm{\alpha}^{[L]}$ is given by  $\alpha^{[L]}_{ij}$.
%%%%
To avoid confusion, we use $\bm{\alpha}$ exclusively to denote $K \times KL$ networks, while $K \times K$ networks are denoted by $\bm{\alpha}^{[L]}$. 
%%%%
With a slight abuse of notation, a regime $\mathcal{A} $ is automatically adjusted to the dimensions of the network under consideration, i.e.  $\bm{\alpha}^{[L]} \in \mathcal{A} $  implies $\mathcal{A} \subset \mathbb{R}_{+}^{K \times K}$, while 
$\bm{\alpha}\in \mathcal{A} $ implies $\mathcal{A} \subset \mathbb{R}_{+}^{K \times KL}$. 
%%%%

%%%%
We now recall the main results in \cite{Chan2020}, which are analog to Theorem \ref{theorem: extremal gain} but for $K\times K$ networks.
%%%%
%%%
\begin{theorem}
\label{theorem:extremal_gain_IC}
 {\normalfont{\textbf{( \cite[Th. 5.1, 6.1 and 7.1]{Chan2020} )} }}
The extremal GDoF gains of transmitter cooperation over TIN in the three regimes of interest for $K \times K$ networks 
are as follows:
%%%
\begin{equation}
\label{eq:extremal_gains_IC}
\max_{\bm{\alpha}^{[L]} \in \mathcal{A}} 
\frac{d_{\Sigma}^{\mathrm{MBC}} (\bm{\alpha}^{[L]} )}
{d_{\Sigma}^{\mathrm{TIN}}(\bm{\alpha}^{[L]} )} = \eta_{K,1}(\mathcal{A})  = 
\begin{cases}
\frac{3}{2}, & \mathcal{A} = \mathcal{A}^{\mathrm{TIN}} \\ 
2 -  \frac{1}{K}, & \mathcal{A} = \mathcal{A}^{\mathrm{CTIN}} \\
\Theta\big(\log(K)\big), & \mathcal{A} = \mathcal{A}^{\mathrm{SLS}}.
\end{cases}
\end{equation}
%%%
\end{theorem}
%%%%
In each regime of interest $\mathcal{A}$, drawn from $ \{\mathcal{A}^{\mathrm{TIN}},\mathcal{A}^{\mathrm{CTIN}},\mathcal{A}^{\mathrm{SLS}} \}$, the authors in \cite{Chan2020} identify a network 
$\bm{\alpha}^{\star[L]} \in \mathcal{A}$ (or a class of networks) for which the following lower bound holds 
%%%%
\begin{equation}
d_{\Sigma}^{\mathrm{MBC}} (\bm{\alpha}^{\star[L]} ) \geq \eta_{K,1}(\mathcal{A}) 
d_{\Sigma}^{\mathrm{TIN}} (\bm{\alpha}^{\star[L]} )
\end{equation}
%%%%%
where $\eta_{K,1}(\mathcal{A})$ takes on values as given in \eqref{eq:extremal_gains_IC}.
%%%
This lower bound is combined with a matching upper bound, derived using an
analog of the sum-GDoF outer bound in \eqref{eq:sum_GDoF_outer_MISO_BC}, but specialized to $K \times K $ networks.
%%%%
This matching upper bound is given by
%%%%
\begin{equation}
\label{eq:extremal_gain_UB_IC}
d_{\mathrm{out},\Sigma}^{\mathrm{MBC}} (\bm{\alpha}^{[L]} ) \leq  \eta_{K,1}(\mathcal{A}) 
d_{\Sigma}^{\mathrm{TIN}} (\bm{\alpha}^{[L]} )
\end{equation}
%%%%%
which is shown to hold for all networks $\bm{\alpha}^{[L]} \in \mathcal{A}$, in each of the regimes of interest.
%%%%%
Next, these bounds for $K \times K$ networks are utilized to derive similar bounds for $K \times KL$ networks.
%%%%
In what follows, we fix an arbitrary regime $\mathcal{A}$ drawn from $ \{\mathcal{A}^{\mathrm{TIN}},\mathcal{A}^{\mathrm{CTIN}},\mathcal{A}^{\mathrm{SLS}} \}$.  
%%%
\begin{remark}
It is worthwhile highlighting that the upper bound in \eqref{eq:extremal_gain_UB_IC} does not appear in this explicit form in \cite{Chan2020}, yet it can be inferred from the proofs of Theorems 5.1, 6.1 and 7.1, presented respectively in Sections V.A, VI.A and VII.A of the same paper. It is shown that the inequality $d_{\Sigma}^{\mathrm{MBC}} (\bm{\alpha}^{[L]} ) \leq  \eta_{K,1}(\mathcal{A}) d_{\Sigma}^{\mathrm{TIN}} (\bm{\alpha}^{[L]} )$ holds for the three regimes of interest by bounding 
$d_{\Sigma}^{\mathrm{MBC}} (\bm{\alpha}^{[L]} ) $ above using analogs of the MISO-BC cycle bounds in Definition \ref{def:cycle_bounds_MBC}, specialized to $K \times K$ networks. This is equivalent to bounding $d_{\Sigma}^{\mathrm{MBC}} (\bm{\alpha}^{[L]} ) $ above by $d_{\mathrm{out},\Sigma}^{\mathrm{MBC}} (\bm{\alpha}^{[L]} )$, obtained from the outer bound region $\mathcal{D}^{\mathrm{SLS}}_{\mathrm{out}} (\bm{\alpha}^{[L]} )$ 
in Theorem \ref{theorem:MISO_BC_outerbound}. This in turn allows us to write the upper bound in \eqref{eq:extremal_gain_UB_IC}.
%%%
\hfill $\lozenge$
%%%
\end{remark}
%%%
\subsection{Lower Bound}
%%%
We first show that a lower bound on the extremal gain in Theorem \ref{theorem: extremal gain}  given  by 
%%%
\begin{equation}
\max_{\bm{\alpha} \in \mathcal{A}} 
\frac{d_{\Sigma}^{\mathrm{MBC}} (\bm{\alpha} )}
{d_{\Sigma}^{\mathrm{TIN}}(\bm{\alpha} )}  \geq  \eta_{K,1}(\mathcal{A})
\end{equation}
%%%
is easily obtained from Theorem \ref{theorem:extremal_gain_IC}.
%%%
Let  $\bm{\alpha}^{\star[L]}$ be a $K \times K$ network that attains $\eta_{K,1}(\mathcal{A})$,
that is 
%%%
\begin{equation}
\bm{\alpha}^{\star[L]} = \arg \max_{\bm{\alpha}^{[L]} \in \mathcal{A}} 
\frac{d_{\Sigma}^{\mathrm{MBC}} (\bm{\alpha}^{[L]} )}
{d_{\Sigma}^{\mathrm{TIN}}(\bm{\alpha}^{[L]} )}.
\end{equation}
%%%
Now let  $\bm{\alpha}^{\star}$ be a $K \times KL$ network that includes $\bm{\alpha}^{\star[L]}$ as a sub-network, and which is obtained by adding $L - 1$ trivial UEs to each cell (i.e. with all channel strengths parameters set to zero).
%%%
Adding trivial users does not alter the GDoF, and hence we have 
$d_{\Sigma}^{\mathrm{MBC}}(\bm{\alpha}^{\star} ) = d_{\Sigma}^{\mathrm{MBC}}(\bm{\alpha}^{\star [L]} )$ and $d_{\Sigma}^{\mathrm{TIN}}(\bm{\alpha}^{\star} ) = 
d_{\Sigma}^{\mathrm{TIN}}(\bm{\alpha}^{\star[L]} )$. 
%%%
This directly leads to
%%%
\begin{equation}
\max_{\bm{\alpha} \in \mathcal{A}} 
\frac{d_{\Sigma}^{\mathrm{MBC}} (\bm{\alpha} )}
{d_{\Sigma}^{\mathrm{TIN}}(\bm{\alpha} )}  \geq  
\frac{d_{\Sigma}^{\mathrm{MBC}} (\bm{\alpha}^{\star} )}
{d_{\Sigma}^{\mathrm{TIN}}(\bm{\alpha}^{\star} )} 
%%%
= \frac{d_{\Sigma}^{\mathrm{MBC}} (\bm{\alpha}^{\star [L]} )}
{d_{\Sigma}^{\mathrm{TIN}}(\bm{\alpha}^{\star [L]} )} 
= \eta_{K,1}(\mathcal{A}).
\end{equation}
%%%
Having obtained a lower bound, it now remains to prove a matching upper bound.
%%%
\subsection{Upper Bound}
%%%
We first observe that for any network $\bm{\alpha}$, we have 
%%%
\begin{equation}
\label{eq:sum_GDoF_TIN_equality}
d^{\mathrm{TIN}}_{\Sigma}(\bm{\alpha}) = 
d^{\mathrm{TIN}}_{\Sigma}(\bm{\alpha}^{[L]}).
\end{equation}
%%%
In other words, starting from a network $\bm{\alpha}^{[L]}$ with single-user cells, including additional (weaker) users in each cell does not increase the sum-GDoF achieved using TIN. 
%%%
To see this, we first note that  $d^{\mathrm{TIN}}_{\Sigma}(\bm{\alpha}) \geq 
d^{\mathrm{TIN}}_{\Sigma}(\bm{\alpha}^{[L]})$ evidently holds, since $\bm{\alpha}^{[L]}$ is a sub-network of $\bm{\alpha}$, and the TIN scheme is a special case of mc-TIN scheme.
%%%
The other direction,  i.e. $d^{\mathrm{TIN}}_{\Sigma}(\bm{\alpha}) \leq  d^{\mathrm{TIN}}_{\Sigma}(\bm{\alpha}^{[L]})$, holds by construction of the mc-TIN scheme. All $L$ messages of cell $k$ are decoded by UE-$(L,k)$, which in turn bounds the achievable GDoF in each cell $k$ by the achievable GDoF of UE-$(L,k)$. 
%%%

%%%
Next, we make a somewhat similar observation to the one above, but in the context of the MISO-BC GDoF outer bound.
%%%
In particular, we observe that for any network $\bm{\alpha}$, we have 
%%%
\begin{equation}
\label{eq:MBC_sum_GDoF_UB}
d_{\mathrm{out},\Sigma}^{\mathrm{MBC}} (\bm{\alpha} ) \leq d_{\mathrm{out},\Sigma}^{\mathrm{MBC}} (\bm{\alpha}^{[L]} )
\end{equation}
%%%
To see this, we first observe that the outer bound  $\mathcal{D}^{\mathrm{SLS}}_{\mathrm{out}}(\bm{\alpha})$ in Theorem \ref{theorem:MISO_BC_outerbound} is included in the region $\bar{\mathcal{D}}^{\mathrm{MBC}}_{\mathrm{out}}(\bm{\alpha})$, 
described  by all tuples $ \mathbf{d} \in \mathbb{R}_{+}^{KL}$ that satisfy
%%%
\begin{align}
\nonumber
 &\sum_{k \in \{ \sigma \}}  \bar{d}_k^{[L]}     \leq \Delta_{\pi,m}^{+}, \;
 \forall  m \in \langle |\pi| \rangle, \pi \in \Pi, \\
 %%%
 \label{eq:MBC_GDoF_outerbound_2}
 & \{\pi\} = \{(L,\sigma(1)),\ldots,(L,\sigma(|\pi|))\}.
\end{align}
%%%
This holds as \eqref{eq:MBC_GDoF_outerbound_2} is obtained from \eqref{eq:MBC_GDoF_outerbound} by removing some of the cycle bounds in the latter, and only keeping the bounds that involve all $L$ users in each participating cell of any cycle $\pi$.
%%%
Denoting the sum-GDoF outer bound obtained from $\bar{\mathcal{D}}^{\mathrm{MBC}}_{\mathrm{out}}(\bm{\alpha})$ in  \eqref{eq:MBC_GDoF_outerbound_2} as
$\bar{d}_{\mathrm{out},\Sigma}^{\mathrm{MBC}} (\bm{\alpha} ) $,
we clearly have
%%%
\begin{equation}
\label{eq:d_sum_LB_bar}
 d_{\mathrm{out},\Sigma}^{\mathrm{MBC}} (\bm{\alpha} )  \leq 
 \bar{d}_{\mathrm{out},\Sigma}^{\mathrm{MBC}} (\bm{\alpha} ).
\end{equation}
%%%
Note that $\bar{\mathcal{D}}^{\mathrm{MBC}}_{\mathrm{out}}(\bm{\alpha})$ in \eqref{eq:MBC_GDoF_outerbound_2} may be viewed as a GDoF region for a $K \times K$ network, with GDoF tuples given by
$\bar{\mathbf{d}} \triangleq \big( \bar{d}_k^{[L]} : k \in \langle K \rangle \big) \in \mathbb{R}_{+}^{K}$.
%%%
From this observation, it can be easily verified that $\bar{\mathcal{D}}^{\mathrm{MBC}}_{\mathrm{out}}(\bm{\alpha})$ in \eqref{eq:MBC_GDoF_outerbound_2}  coincides with $\mathcal{D}^{\mathrm{MBC}}_{\mathrm{out}}(\bm{\alpha}^{[L]})$, and therefore we have 
%%%
\begin{equation}
\label{eq:d_sum_bar_d_sum}
\bar{d}_{\mathrm{out},\Sigma}^{\mathrm{MBC}} (\bm{\alpha} ) = d_{\mathrm{out},\Sigma}^{\mathrm{MBC}} (\bm{\alpha}^{[L]} ).
\end{equation}
%%%
By combining \eqref{eq:d_sum_LB_bar} and \eqref{eq:d_sum_bar_d_sum}, the inequality in \eqref{eq:MBC_sum_GDoF_UB} is
obtained. 
%%%

%%%
Next, we employ \eqref{eq:sum_GDoF_TIN_equality} and \eqref{eq:MBC_sum_GDoF_UB} to obtain the following
%%%%
\begin{align}
\max_{\bm{\alpha} \in \mathcal{A}} 
\frac{d_{\Sigma}^{\mathrm{MBC}} (\bm{\alpha} )}
{d_{\Sigma}^{\mathrm{TIN}}(\bm{\alpha} )}   &  \leq \max_{\bm{\alpha} \in \mathcal{A}} 
\frac{d_{\Sigma,\mathrm{out}}^{\mathrm{MBC}} (\bm{\alpha} )}
{d_{\Sigma}^{\mathrm{TIN}}(\bm{\alpha} )}\\
%%%
\label{eq:eta_K_L_UB}
& \leq \max_{\bm{\alpha}^{[L]} \in \mathcal{A}} 
\frac{d_{\Sigma,\mathrm{out}}^{\mathrm{MBC}} (\bm{\alpha}^{[L]} )}
{d_{\Sigma}^{\mathrm{TIN}}(\bm{\alpha}^{[L]} )}.
\end{align}
%%%%
Finally,  combining \eqref{eq:eta_K_L_UB} with the upper bound in \eqref{eq:extremal_gain_UB_IC},
we obtain the desired upper bound
%%%
\begin{equation}
\max_{\bm{\alpha} \in \mathcal{A}} 
\frac{d_{\Sigma}^{\mathrm{MBC}} (\bm{\alpha} )}
{d_{\Sigma}^{\mathrm{TIN}}(\bm{\alpha} )}  \leq \eta_{K,1}(\mathcal{A}).
\end{equation}
%%%
This completes the proof of Theorem \ref{theorem: extremal gain}.
%%%
\section{Conclusion}
%%%
We studied the GDoF of downlink cellular networks with finite precision CSIT, modeled by the IBC under no BS cooperation and the overloaded MISO-BC under full BS cooperation; while focusing 
on three weak inter-cell interference regimes of interest, given (in a monotonically increasing order) by: the mc-TIN regime, mc-CTIN regime, and mc-SLS regime. 
%%%
Through a new application of AI bounds to $K \times KL$ cellular-type networks,
we derived outer bounds for the $K \times KL$ IBC  and the $K \times KL$ MISO-BC under finite precision CSIT.
%%%
The former outer bound is utilized to show that the mc-TIN scheme is GDoF optimal for the  $K \times KL$ IBC under finite precision CSIT in the entire mc-CTIN regime; while the latter is combined with a non-trivial achievability argument to show that the mc-SLS scheme is GDoF optimal for the  $2 \times 2L$ MISO-BC in the mc-SLS regime.
%%%
We then leveraged the recently introduced extremal network analysis framework to study the GDoF gain from mc-Co over mc-TIN in large asymmetric cellular networks. 
%%%
Our analysis reveals that the extremal GDoF gain from mc-Co over mc-TIN is limited to (small) constant factors in the mc-TIN and mc-CTIN regimes, and scales logarithmically with the number of cells $K$ in the mc-SLS regime.
%%%
These results mirror and generalize recent results for $K \times K$ networks by Chan \emph{et al.} \cite{Chan2020}.
%%%

%%%
The results presented in this paper give rise to a number of interesting questions, which remain unanswered. 
%%%
For instance, it is not clear whether the logarithmic extremal GDoF gain in the mc-SLS regime is a fundamental gain of the MISO-BC over the IBC, or rather an artefact of limiting the IBC to mc-TIN schemes.
%%%
Will this gain diminish if we replace the mc-TIN scheme with a GDoF-optimal robust scheme for the IBC in the mc-SLS?
%%%
Such robust scheme for the IBC in the mc-SLS regime will most likely rely on layered superposition and inter-cell rate-splitting with overwhelmingly many parameters and optimization variables, rendering extremal network analysis all the more essential in this case. 
%%%
Another intriguing direction is to study extremal GDoF gains of the MISO-BC over the IBC beyond the mc-SLS regime, e.g. in the general weak inter-cell interference regime specified in \eqref{eq:SIR_order}.
%%%
Moreover, while the extremal network analysis carried out in this paper focuses on the sum-GDoF viewpoint,
it is of interest to investigate extremal gains from the perspective of \emph{fair} utility functions, such as the symmetric-GDoF.
%%%
In this case, switching off weaker users and reducing $K \times KL$ networks to $K \times K$ networks may not be possible any more; and the cell-edge benefits of robust BS cooperation may be better reflected.
%%%
\appendix 
%%%
\section{Proof of Lemma \ref{lemma:PTIN_region_2_cell}}
\label{appendix:PTIN_region_2_cell}
%%%%
Here we show that the polyhedral TIN region $\mathcal{D}^{\mathrm{PTIN}}(a)$, given by all tuples $\mathbf{d}_{\mathrm{s}} \in \mathbb{R}_{+}^{2K}$
that satisfy \eqref{eq:GDoF_TIN_2_cell_SLS_3} for some $\mathbf{r} \in \mathcal{R}(a)$,
is identical to the region specified by \eqref{eq:PTIN_2_cell_1} and \eqref{eq:PTIN_2_cell_2}.
%%%
To this end, we first note from \eqref{eq:power_control_order} and \eqref{eq:GDoF_TIN_2_cell_SLS_3}  that $\mathcal{D}^{\mathrm{PTIN}}(a)$ may be described by all tuples 
$\mathbf{d} \in \mathbb{R}_{+}^{2K}$ that satisfy:
%%%%
\begin{align}
\label{eq:GDoF_potential_1}
d_{i}^{[l_i]}  &  \leq r_{i}^{[l_i]}   - r_{i}^{[l_i + 1]}  \\
%%%
\label{eq:GDoF_potential_2}
d_{i}^{[l_i]} - \big( \alpha_{ii}^{[l_i]} - \alpha_{ij}^{[l_i]}  \big) &  \leq r_{i}^{[l_i]}  - r_{j}^{[1]}\\
%%%
\label{eq:GDoF_potential_3}
d_{i}^{[l_i]}  - \alpha_{ii}^{[l_i]}  &  \leq r_{i}^{[l_i]}
\end{align}
%%%%
for all $i,j \in \langle 2 \rangle$, $i \neq j$, and $l_i, l_j \in \langle L \rangle$, where
%%%%
\begin{align}
\label{eq:GDoF_potential_4}
0  &  \leq r_{i}^{[l_i]}   - r_{i}^{[l_i + 1]}  \\
%%%
\label{eq:GDoF_potential_5}
a  & \leq  -  r_{i}^{[1]}.
\end{align}
%%%%
In the above, we set $r_i^{[L +1 ]} = - \infty$ as in \eqref{eq:GDoF_TIN_2_cell_SLS_3}.
%%%
Note that in \eqref{eq:GDoF_potential_1}--\eqref{eq:GDoF_potential_3} and the remainder of this appendix, we drop the subscript in $\mathbf{d}_{\mathrm{s}} $, and we use $\mathbf{d}$ instead (as no confusion may arise due to this). 
%%%%
Moreover, it is worthwhile noting that inequalities as the one in \eqref{eq:GDoF_potential_4} are in fact redundant, as they are implied by \eqref{eq:GDoF_potential_1} and the non-negativity of GDoF components.

%%%%
Next, we construct a directed graph (digraph) known as the potential graph \cite{Geng2015}, through which we can carry out an efficient FM elimination of the power control variables $\mathbf{r}$ in \eqref{eq:GDoF_potential_1}--\eqref{eq:GDoF_potential_5}.
%%%%
\subsection{Potential Graph}
%%%%
We define the potential graph associated with the $2$-cell network of interest as $\mathcal{G} = (\mathcal{V} , \mathcal{E})$, where $\mathcal{V}$ is a set of $2L + 1$ vertices, while $\mathcal{E}$ is a set of directed edges.
%%%
The set of vertices is given by 
%%%
\begin{equation}
\nonumber
\mathcal{V} \triangleq \big\{ u  \big\} \cup \big\{ v_i^{[l_i]} : (l_i,i) \in \mathcal{U}  \big\}
\end{equation}  
%%%
comprising a vertex $v_i^{[l_i]}$ for each UE-$(l_i,i) $, and a \emph{ground} vertex $u$.
%%%
The set of directed edges $\mathcal{E}$ is given by the union $\mathcal{E}_1 \cup \mathcal{E}_2  \cup \mathcal{E}_3 \cup \mathcal{E}_4$,
%%%
where constituent subsets are defined as follows:
%%%
\begin{align}
\nonumber
\mathcal{E}_1 & = \left\{ (v_i^{[l_i]}, v_i^{[l_i+1]}) : i \in \langle 2 \rangle , l_i \in \langle L - 1 \rangle \right\}  \\
\nonumber
\mathcal{E}_2 & = \left\{ (v_i^{[l_i]}, v_j^{[1]}) : (l_i,i) \in \mathcal{U} ,j \in \langle 2 \rangle , i\neq j \right\} \\
\nonumber
\mathcal{E}_3 & = \left\{ (v_i^{[l_i]}, u) : (l_i,i) \in \mathcal{U}  \right\}  \\
\nonumber
\mathcal{E}_4 & = \left\{ (u,v_i^{[1]} ) : i \in \langle 2 \rangle  \right\}.
\end{align}
%%%
It is worthwhile noting that $\mathcal{G}$ is not a complete digraph, e.g. vertex  $v_i^{[l_i + 1]}$ with $l_i \in \langle L - 1 \rangle$ may only be reached through its preceding vertex $v_i^{[l_i]}$ of the same cell---see Fig. \ref{fig:potential_graph}.
%%%
This reflects the successive decoding order of the mc-TIN scheme as described in Section 
\ref{subsec:2_cell_SLS}, as well as the power allocation order in \eqref{eq:power_control_order}.
%%%
Note that this incompleteness of the potential graph in the multi-cell setting is a key difference to the potential graph in \cite{Geng2015}, constructed  for the interference channel.
%%% 
\begin{figure}[h]
\vspace{-1mm}
\centering
\includegraphics[width = 0.7\textwidth]{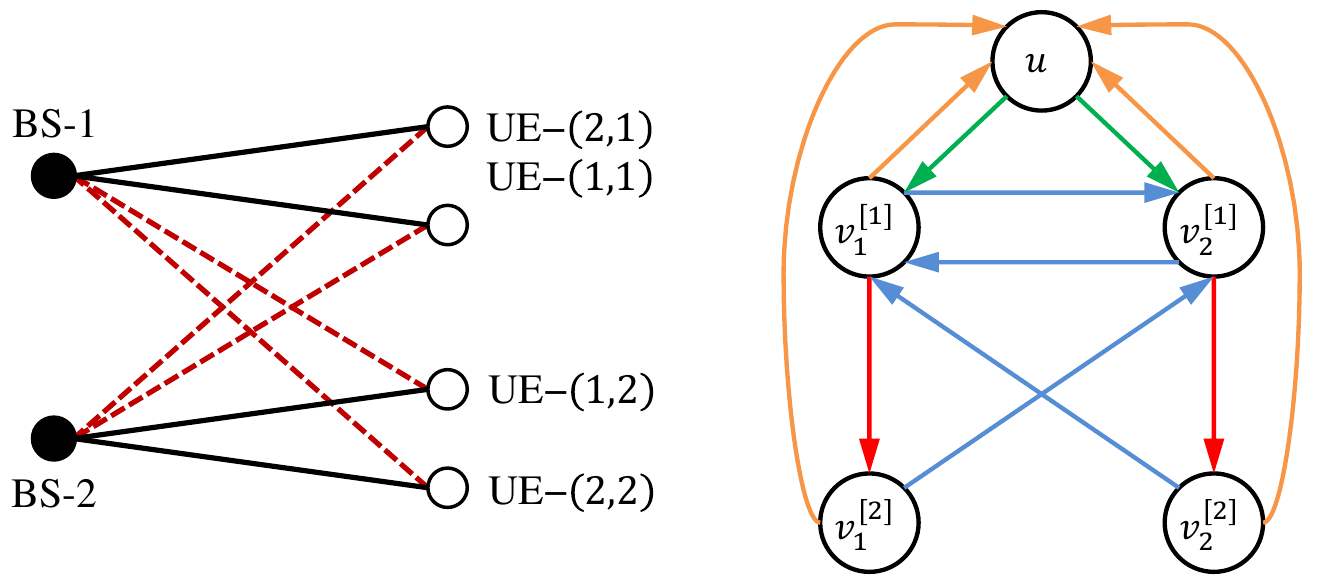}
\caption{Left: $2$-cell network with $2$ users per-cell. Right: The corresponding potential graph $\mathcal{G}$, with subsets of edges given by $\mathcal{E}_1$ in red, $\mathcal{E}_2$ in blue, $\mathcal{E}_3$ in orange, and $\mathcal{E}_4$ in green.}
\label{fig:potential_graph}
\end{figure}
%%%

%%%
Next, we assign the following lengths to edges from the four subsets defined above:
%%%
\begin{align}
\label{eq:edge_lengths_1}
l(v_i^{[l_i]}, v_i^{[l_i + 1]}) & = - d_i^{[l_i]}   \\
l(v_i^{[l_i]}, v_j^{[1]})  & =  \alpha_{ii}^{[l_i]} - \alpha_{ij}^{[l_i]} - d_i^{[l_i]} \\
l(v_i^{[l_i]}, u)  & = \alpha_{ii}^{[l_i]} - d_i^{[l_i]}   \\
\label{eq:edge_lengths_4}
l(u,v_i^{[1]} )  & = - a.
\end{align}
%%%
As these lengths clearly depend on $\mathbf{d}$, we parametrize $\mathcal{G}$ by $\mathbf{d}$ in what follows, and we write $\mathcal{G}(\mathbf{d} )$.
%%%
A \emph{potential} associated with $\mathcal{G}(\mathbf{d} )$ is a function $p: \mathcal{V} \to \mathbb{R}$ such that for any $v', v'' \in \mathcal{V}$, we have 
\begin{equation}
\nonumber
l(v',v'') \geq p(v'') - p(v').
\end{equation}
%%%
These inequalities only depend on differences of potentials, and hence we may set the potential of the ground node as  $p(u) = 0$, without any loss in generality.
%%%
It follows that $p$ must satisfy
%%%
\begin{align}
\label{eq:potential_1}
- d_i^{[l_i]} & \geq p(v_i^{[l_i + 1]}) - p(v_i^{[l_i]})     \\
\alpha_{ii}^{[l_i]} - \alpha_{ij}^{[l_i]} - d_i^{[l_i]} & \geq p(v_j^{[1]}) - 
p(v_i^{[l_i]}) \\
\alpha_{ii}^{[l_i]} - d_i^{[l_i]}  & \geq  - p(v_i^{[l_i]})    \\
\label{eq:potential_4}
- a & \geq  p(v_i^{[1]})
\end{align}
%%%
for all $i,j \in \langle 2 \rangle$, $i \neq j$, and $l_i, l_j \in \langle L \rangle$.
%%%
Setting $p(v_i^{[l_i]}) = r_i^{[l_i]}$, for all $(l_i, i) \in \mathcal{U}$, we notice that the set of inequalities for potentials in  \eqref{eq:potential_1}--\eqref{eq:potential_4}  is identical to the set of inequalities involving power control variables in\footnote{Recall that \eqref{eq:GDoF_potential_4} in the set of inequalities \eqref{eq:GDoF_potential_1}--\eqref{eq:GDoF_potential_5} is redundant.}  \eqref{eq:GDoF_potential_1}--\eqref{eq:GDoF_potential_5}.
%%% 
Therefore, it holds that: \emph{for any GDoF tuple $\mathbf{d} \in \mathcal{D}^{\mathrm{PTIN}}(a)$, 
a power control tuple $\mathbf{r}$ is feasible and achieves $\mathbf{d}$, i.e. by satisfying  the set of inequalities in \eqref{eq:GDoF_potential_1}--\eqref{eq:GDoF_potential_5}, if an only if it is a valid potential for the digraph  $\mathcal{G}(\mathbf{d} )$.} 
%%%%

%%%%
We are now ready to invoke the potential theorem \cite[Th. 8.2]{Schrijver2002},  which states that: \emph{there exists a potential function for $\mathcal{G}(\mathbf{d} )$ if and only if each directed circuit in $\mathcal{G}(\mathbf{d} )$  has a non-negative length}.
%%%
Therefore, we may conclude that for any $\mathbf{d} \in \mathcal{D}^{\mathrm{PTIN}}(a)$,
there exists a feasible power control tuple $\mathbf{r}$ that  achieves $\mathbf{d}$ 
if and only if each directed circuit in $\mathcal{G}(\mathbf{d} )$  has a non-negative length.
%%%
It remains to interpret non-negative lengths of directed circuits in terms of GDoF inequalities.
%%%
To this end, a directed circuit is represented as $(v_1 \rightarrow v_2 \rightarrow \cdots \rightarrow v_M)$, where $v_m \in \mathcal{V}$ such that $(v_m,v_{m+1}) \in \mathcal{E}$, for all $m \in \langle M \rangle$ and $v_{M+1} = v_1$.
%%%
The length of such circuit is given by  $\sum_{m = 1}^{M} l(v_m,v_{m+1})$.
%%%%
\subsection{From Directed Circuits to GDoF Inequalities}
%%%%
We start with single-cell circuits, involving users from one cell only.
%%%
These circuits take the form 
%%%%
\begin{equation}
\label{eq:single_cell_circuit}
\big( v_{i}^{[1]} \to  v_{i}^{[2]}  \to \cdots \to v_{i}^{[l_i]} \to u  \big)
\end{equation}
%%%%
for some $(l_i,i) \in \mathcal{U}$.
%%%
From the non-negative circuit length condition, it can be easily checked using the lengths in \eqref{eq:edge_lengths_1}--\eqref{eq:edge_lengths_4} that a circuit of the form in \eqref{eq:single_cell_circuit}  gives rise to a GDoF bound given by
%%%%
\begin{equation}
\label{eq:potential_graph_GDoF_1_cell}
\bar{d}_i^{[l_i]}     \leq   \alpha_{ii}^{[l_i]}  - a.
\end{equation}
%%%
Recall that $\bar{d}_i^{[l_i]} $ is a shorthand notation for the sum $\sum_{s_i = 1}^{l_i}d_i^{[s_i]} $.
%%%

Next, we move on to multi-cell circuits that involve users from both cells. 
%%%
For any pair of users $(l_i,i),(l_j,j) \in \mathcal{U}$, where $i \neq j$, 
%%%
we have the four following types of circuits:
%%%%
\begin{align}
\nonumber
& \big( v_{i}^{[1]} \to v_{i}^{[2]} \cdots \to v_{i}^{[l_i]} \to u \to v_{j}^{[1]}  \to v_{j}^{[2]} \cdots \to v_{j}^{[l_j]}  \big) \\
\nonumber
& \big( v_{i}^{[1]} \to v_{i}^{[2]} \cdots \to v_{i}^{[l_i]} \to v_{j}^{[1]}  \to v_{j}^{[2]} \cdots \to v_{j}^{[l_j]} \to u  \big) \\
\nonumber
& \big( v_{i}^{[1]} \to v_{i}^{[2]} \cdots \to v_{i}^{[l_i]}  \to u \to v_{j}^{[1]}  \to v_{j}^{[2]} \cdots \to v_{j}^{[l_j]} \to u  \big) \\
\nonumber
& \big( v_{i}^{[1]} \to v_{i}^{[2]} \cdots \to v_{i}^{[l_i]} \to v_{j}^{[1]}  \to v_{j}^{[2]} \cdots \to v_{j}^{[l_j]}  \big).
\end{align}
%%%%
It can be verified that the above circuits give rise o the following GDoF inequalities:
%%%%
\begin{align}
\label{eq:potential_graph_GDoF_2_cell_1}
\bar{d}_i^{[l_i]}   + 
\bar{d}_j^{[l_j]}  &  \leq   \alpha_{ii}^{[l_i]} - a
+ \alpha_{jj}^{[l_j]} - \alpha_{ji}^{[l_j]}  \\
%%%
\bar{d}_i^{[l_i]}   + 
\bar{d}_j^{[l_j]}  &  \leq   \alpha_{ii}^{[l_i]} - \alpha_{ij}^{[l_i]} 
+ \alpha_{jj}^{[l_j]} - a \\
%%%
\bar{d}_i^{[l_i]}   + 
\bar{d}_j^{[l_j]}  &  \leq   \alpha_{ii}^{[l_i]} - a
+ \alpha_{jj}^{[l_j]} - a  \\
%%%
\label{eq:potential_graph_GDoF_2_cell_4}
\bar{d}_i^{[l_i]}   + 
\bar{d}_j^{[l_j]}  &  \leq   \alpha_{ii}^{[l_i]} - \alpha_{ij}^{[l_i]} 
+ \alpha_{jj}^{[l_j]} - \alpha_{ji}^{[l_j]}.
\end{align}
%%%%
It is evident that together with the non-negativity of the GDoF, the inequalities in \eqref{eq:potential_graph_GDoF_1_cell} and \eqref{eq:potential_graph_GDoF_2_cell_1}--\eqref{eq:potential_graph_GDoF_2_cell_4} define a GDoF region
identical to the one defined in \eqref{eq:PTIN_2_cell_1} and \eqref{eq:PTIN_2_cell_2} in Lemma 
\ref{lemma:PTIN_region_2_cell}. 
%%%%
Therefore, after eliminating the power control variables $\mathbf{r}$, the polyhedral TIN region $\mathcal{D}^{\mathrm{PTIN}}(a)$ does indeed reduce to the one in Lemma 
\ref{lemma:PTIN_region_2_cell}, which concludes the proof.
%%%%
\section{Proof of Lemma \ref{lemma:Minkowski_sum}}
\label{appendix:Minkowski_sum}
%%%%
The result in Lemma \ref{lemma:Minkowski_sum} is a consequence of
$\mathcal{D}^{\mathrm{Mul}}(a)$ and $\mathcal{D}^{\mathrm{PTIN}'}(a)$ both being polymatroids, and 
the direct summability property of polymatroids given in \cite[Th. 44.6]{Schrijver2002} 
(see also \cite[Th. 3]{McDiarmid1975}), as shown in detail next.
%%%
Before we proceed, it is worthwhile recalling that the set of UEs in the $2$-cell setting of interest is given by 
$\mathcal{U} = \mathcal{U}_1 \cup \mathcal{U}_2$, where $\mathcal{U}_i \triangleq \{ (l_i,i): l_i \in \langle L \rangle \}$ for all $i \in \langle 2 \rangle$.
%%%

%%%
First, we observe that $\mathcal{D}^{\mathrm{Mul}}(a)$, described in \eqref{eq:common_2_cell_2_1} 
and \eqref{eq:common_2_cell_2_2}, can be equivalently expressed in a polymatroid-like fashion as all
 tuples $\mathbf{d}_{\mathrm{c}} \in \mathbb{R}_{+}^{2K}$ that satisfy
%%%
\begin{equation}
\label{eq:common_polymatroid}
\mathbf{d}_{\mathrm{c}}(\mathcal{S} )  \leq a \cdot \mathbbm{1}(\mathcal{S} \neq  \emptyset), \ \forall \mathcal{S} \subseteq \mathcal{U}
\end{equation}
%%%%
where $\mathbf{d}_{\mathrm{c}}(\mathcal{S} ) $ is just a shorthand notation for the sum $\sum_{(l_i,i) \in \mathcal{S}} d_{\mathrm{c},i}^{[l_i]}$.
%%%%
This follows as the additional inequalities in  \eqref{eq:common_polymatroid}, not included in \eqref{eq:common_2_cell_2_1}  and \eqref{eq:common_2_cell_2_2}, are redundant.
%%%%
It can be easily checked that \eqref{eq:common_polymatroid} is a polymatroid---note that the associated set function is 
normalized (since $\mathbf{d}_{\mathrm{c}}(\emptyset)  = 0$), non-decreasing and submodular (since it is constant for all $\mathcal{S} \neq \emptyset$).
%%%%

%%%%
Next, we move on to the region $\mathcal{D}^{\mathrm{PTIN}'}(a)$, described in \eqref{eq:PTIN_2_cell_2_1} and \eqref{eq:PTIN_2_cell_2_2}, and we express it in a similar polymatroid-like fashion as all tuples $\mathbf{d}_{\mathrm{s}} \in \mathbb{R}_{+}^{2K}$ that satisfy
%%%
\begin{equation}
\label{eq:PTIN_polymatroid}
\mathbf{d}_{\mathrm{s}}(\mathcal{S} )  \leq f(\mathcal{S} ) , \ \forall \mathcal{S} \subseteq \mathcal{U}.
\end{equation}
%%%
In the above, $f: 2^{\mathcal{U}} \rightarrow \mathbb{R}_{+}$ is a set function associated with the polyhedron $\mathcal{D}^{\mathrm{PTIN}'}(a)$, where  $2^{\mathcal{U}}$ denotes the power set over $\mathcal{U}$.
%%%
This set function $f$ is given by
%%%
\begin{align}
\label{eq:set_function_2_cell}
 f(\mathcal{S}) & \triangleq \begin{cases}
 0, &  \mathcal{S} =\emptyset \\
 \alpha_{11}^{[s_1]}  -  \max \big\{\alpha_{12}^{[s_1]}, a \big\}, & 
 \mathcal{S} \neq \emptyset, \mathcal{S} \cap \mathcal{U}_2 = \emptyset    \\
 %%%
 \alpha_{22}^{[s_2]}  -  \max \big\{\alpha_{21}^{[s_2]}, a \big\}, & 
 \mathcal{S} \neq \emptyset, \mathcal{S} \cap \mathcal{U}_1 = \emptyset    \\
 %%%
\alpha_{11}^{[s_1]} + 
\alpha_{22}^{[s_2]}  - \max\big\{\alpha_{12}^{[s_1]} + a, a + \alpha_{21}^{[s_2]} , 2a , \alpha_{12}^{[s_1]} + \alpha_{21}^{[s_2]} \big\}, &  
\mathcal{S} \cap \mathcal{U}_1 \neq \emptyset, \mathcal{S} \cap \mathcal{U}_2 \neq \emptyset,
 \end{cases} \\
 %%%
 \nonumber
 & \text{where we define}  \ \ \  s_i \triangleq \max_{(l_i,i) \in \mathcal{U}_i \cap \mathcal{S}} l_i, \ \forall 
 i \in \langle 2 \rangle.
\end{align}
%%%
It is readily seen from \eqref{eq:set_function_2_cell} that for any $\mathcal{S} \subseteq \mathcal{U}$,
the value of $ f(\mathcal{S})$ depends on the pair of UEs $(s_1,1)$ and $(s_2,2)$ only, or more specifically, 
on the user indices $s_1$ and $s_2$ as defined in \eqref{eq:set_function_2_cell}.
%%%
Therefore, and with a slight abuse of notation, $f(\mathcal{S})$ is written as $f(s_1,s_2)$ whenever convenient.
%%%

%%%
As mentioned above, the region $\mathcal{D}^{\mathrm{PTIN}'}(a)$, described in \eqref{eq:PTIN_2_cell_2_1} and \eqref{eq:PTIN_2_cell_2_2}, is equivalently represented by \eqref{eq:PTIN_polymatroid}. This holds since the inequalities in \eqref{eq:PTIN_polymatroid} include those in \eqref{eq:PTIN_2_cell_2_1} and \eqref{eq:PTIN_2_cell_2_2}, while the remaining inequalities in \eqref{eq:PTIN_polymatroid}, which are not included in \eqref{eq:PTIN_2_cell_2_1} and \eqref{eq:PTIN_2_cell_2_2}, are redundant.
%%% 
This is seen by observing that any inequality $\mathbf{d}_{\mathrm{s}}(\mathcal{S} )  \leq f(\mathcal{S} ) = f(s_1,s_2 ) $, where $\mathcal{S} = \mathcal{S}_1 \cup \mathcal{S}_2$ and  $\mathcal{S}_{i} \subseteq \mathcal{U}_i$ for all $i \in \langle 2 \rangle$, is implied by the inequality 
$\mathbf{d}_{\mathrm{s}}(\mathcal{S}' )  \leq f(\mathcal{S}' ) = f(s_1,s_2 ) $, where $\mathcal{S}' = \mathcal{S}'_1 \cup \mathcal{S}'_2$
and $\mathcal{S}'_{i} = \{ (l_i,i) : l_i \in \langle s_i \rangle \} $ for all $i \in \langle 2 \rangle$. The latter inequality is in turn included in 
\eqref{eq:PTIN_2_cell_2_1} and \eqref{eq:PTIN_2_cell_2_2}. 
%%%

In addition to the above, it turns out that the region described \eqref{eq:PTIN_polymatroid} is a polymatroid,
as the set function $f$ is normalized (by definition), non-decreasing and submodular.
%%%
This is shown in detail further on in part \ref{appendix:subsec_polymatroid}. 
%%%
With $\mathcal{D}^{\mathrm{Mul}}(a)$ and $\mathcal{D}^{\mathrm{PTIN}'}(a)$ both being polymatroids, we are now ready to invoke \cite[Th. 44.6]{Schrijver2002}.
%%%
This theorem states that: \emph{the Minkowski sum of a pair of polymatroids, associated with the set functions $f(\mathcal{S})$ and $g(\mathcal{S})$, is a polymatroid associated with the set function $f(\mathcal{S})+ g(\mathcal{S})$.}
%%%
Through a direct application of this theorem, it follows that
$\mathcal{D}^{\mathrm{Mul}}(a) \oplus \mathcal{D}^{\mathrm{PTIN}'}(a)$ is given by all tuples $\mathbf{d} \in \mathbb{R}_{+}^{2K}$ that satisfy
%%%
\begin{equation}
\label{eq:sum_polymatroid}
\mathbf{d}(\mathcal{S} )  \leq  f(\mathcal{S} )  + a \cdot \mathbbm{1}(\mathcal{S} \neq  \emptyset), \ \forall \mathcal{S} \subseteq \mathcal{U}.
\end{equation}
%%%
As for the region in  \eqref{eq:PTIN_polymatroid}, it can be seen that an inequality in \eqref{eq:sum_polymatroid} associated with a set $\mathcal{S} = \mathcal{S}_1 \cup \mathcal{S}_2$, where $\mathcal{S}_{i} \subseteq \mathcal{U}_i$ for all $i \in \langle 2 \rangle$, is redundant unless  each $\mathcal{S}_i $ takes the form
$\{ (l_i,i) : l_i \in \langle s_i \rangle \} $, for some $s_i \in \langle L \rangle$.
%%%
After removing redundant inequalities, the region in \eqref{eq:sum_polymatroid} reduces to the one in 
Lemma \ref{lemma:Minkowski_sum}, hence completing the proof.
%%%
It remains to show that $\mathcal{D}^{\mathrm{PTIN}'}(a)$ is indeed a polymatroid.
%%%
\subsection{Polymatroidality of \eqref{eq:PTIN_polymatroid} }
\label{appendix:subsec_polymatroid}
%%%
$\mathcal{D}^{\mathrm{PTIN}'}(a)$ is a polymatroid if the associated set function $f$ satisfies  the following conditions:
%%%
\begin{itemize}
\item Normalized: $f(\emptyset ) = 0$
\item Non-decreasing: $ f(\mathcal{S}) \leq f(\mathcal{T}) $, for all $\mathcal{S}  \subseteq \mathcal{T} \subseteq \mathcal{U}$.
\item Submodular : $  f\big( \mathcal{S}  \cup \mathcal{T} \big) + f\big( \mathcal{S}  \cap \mathcal{T} \big)  \leq f\big( \mathcal{S} \big) + f\big( \mathcal{T} \big) $, for all $\mathcal{S},\mathcal{T}   \subseteq  \mathcal{U}$.
\end{itemize}
%%%
As the first condition holds by definition, we have to show that the second and third conditions hold as well.
%%%
Let us start by establishing some notation.
%%%
Whenever we consider a subset $ \mathcal{S} \subseteq  \mathcal{U}  $, we assume that $ \mathcal{S} =  \mathcal{S}_1  \cup \mathcal{S}_2 $, where $\mathcal{S}_1   \subseteq \mathcal{U}_1$ and $\mathcal{S}_2 \subseteq \mathcal{U}_2$. Moreover, $s_1$ and $s_2$ denote $\max_{(l_1,1) \in \mathcal{S}_1} l_1$ and 
$\max_{(l_2,2) \in \mathcal{S}_2} l_2$, respectively.
%%%
Note that whenever $\mathcal{S}_i = \emptyset$ for some $i \in \langle 2 \rangle$, we set the corresponding $s_i$ to $0$.
%%%
Similar notation is used for a subset  $ \mathcal{T} \subseteq  \mathcal{U}  $, given by $   \mathcal{T}_1  \cup \mathcal{T}_2 $, and with  $t_1$ and $t_2$ denoting the corresponding maximum user  indices in 
$ \mathcal{T}_1 $ and $  \mathcal{T}_2 $ respectively.
%%%

%%%
Now let us revisit the shorthand notation for $f(\mathcal{S})$, given by $f(s_1,s_2)$, which was introduced in the previous part. 
%%%
It can be verified that  $f(s_1,s_2)$, defined in \eqref{eq:set_function_2_cell},  is equivalently given as 
%%%
\begin{align}
\nonumber
f(s_1,s_2) & =  \min \left\{
\begin{aligned}
&(\alpha_{11}^{[s_1]} - \alpha_{12}^{[s_1]}) \mathbbm{1}(s_1 \neq 0) +  (\alpha_{22}^{[s_2]} - a) \mathbbm{1}(s_2 \neq 0)  , \\
&(\alpha_{11}^{[s_1]} - a)\mathbbm{1}(s_1 \neq 0)  +  (\alpha_{22}^{[s_2]} - \alpha_{21}^{[s_2]})\mathbbm{1}(s_2 \neq 0)  ,  \\
&(\alpha_{11}^{[s_1]} - a)\mathbbm{1}(s_1 \neq 0)  +  ( \alpha_{22}^{[s_2]} - a) \mathbbm{1}(s_2 \neq 0)  ,  \\
&(\alpha_{11}^{[s_1]} - \alpha_{12}^{[s_1]}) \mathbbm{1}(s_1 \neq 0)  + (\alpha_{22}^{[s_2]} - \alpha_{21}^{[s_2]})  
\mathbbm{1}(s_2 \neq 0)  
\end{aligned}
\right\}.
\end{align}
%%%
Moreover, it is readily seen from the above that $f(s_1,s_2)$ can be written more compactly as
%%%
\begin{equation}
\nonumber
f(s_1,s_2) \triangleq \min_{m_1,m_2\in \langle 2 \rangle } \big\{ \delta_{1}^{m_1}(s_1) + \delta_{2}^{m_2}(s_2) \big\}
\end{equation}
%%%%
where we define 
\begin{align}
\nonumber
\delta_{i}^{1}(s_i) & \triangleq \begin{cases} 
\alpha_{ii}^{[s_i]} - \alpha_{ij}^{[s_i]}, & s_i \in \langle L \rangle \\
0, & s_i = 0
\end{cases} \\
%%%%
\nonumber
\delta_{i}^{2}(s_i) & \triangleq \begin{cases} 
\alpha_{ii}^{[s_i]} - a, & s_i \in \langle L \rangle \\
0, & s_i = 0.
\end{cases}
\end{align}
%%%%
Note that both $\delta_{i}^{1}(s_i)$ and $\delta_{i}^{2}(s_i)$ are non-decreasing in their arguments, i.e.
%%%
\begin{equation}
\label{eq:delta_order}
\delta_{i}^{m_i}(s_i - 1)  \leq \delta_{i}^{m_i}(s_i), \ \forall s_i \in \langle L \rangle.
\end{equation}
%%%
This holds for $m_i = 1$ due to the SIR order $\alpha_{ii}^{[s_i]}  - \alpha_{ij}^{[s_i]} \geq \alpha_{ii}^{[s_i - 1]}  - \alpha_{ij}^{[s_i - 1]}$; and for $m_i = 2$ due to the SNR order $\alpha_{ii}^{[s_i]} \geq \alpha_{ii}^{[s_i - 1]}$.
%%%
Recall that both SIR and SNR orders hold in the mc-SLS regime.
%%%

%%%
Next, we prove that $f$ is monotonic. Consider the subsets $\mathcal{S}$ and $\mathcal{T}$, where $\mathcal{S} \subseteq \mathcal{T}$. We have 
%%%
\begin{align}
f(\mathcal{T}) & = f(t_1,t_2) \\ 
\label{eq:f_monoton_1}
& = \min_{m_1,m_2\in \langle 2 \rangle } \big\{ \delta_{1}^{m_1}(t_1) + \delta_{2}^{m_2}(t_2) \big\} \\
\label{eq:f_monoton_1_1}
& =  \delta_{1}^{m_1^{\star} } (t_1) + \delta_{2}^{m_2^{\star}}(t_2) 
\end{align}
%%%%
where $m_1^{\star}$ and $m_2^{\star}$ in \eqref{eq:f_monoton_1_1} are the coefficients that attain the minimum in 
\eqref{eq:f_monoton_1}.
%%%%
For $\mathcal{S}$, we have
%%%
\begin{align}
f(\mathcal{S}) & = f(s_1,s_2) \\ 
\label{eq:f_monoton_2}
& = \min_{m_1,m_2\in \langle 2 \rangle } \big\{ \delta_{1}^{m_1}(s_1) + \delta_{2}^{m_2}(s_2) \big\} \\
\label{eq:f_monoton_3}
& \leq  \delta_{1}^{m_1^{\star} } (s_1) + \delta_{2}^{m_2^{\star}}(s_2)  \\
\label{eq:f_monoton_4}
& \leq  \delta_{1}^{m_1^{\star} } (t_1) + \delta_{2}^{m_2^{\star}}(t_2)  \\
& = f(\mathcal{T}).
\end{align}
%%%%
The inequality in \eqref{eq:f_monoton_3} holds since $m_1^{\star}$ and $m_2^{\star}$, which attain the minimum in  \eqref{eq:f_monoton_1}, do not necessarily  minimize \eqref{eq:f_monoton_2}.
%%%
On the other hand, \eqref{eq:f_monoton_4} holds since $\delta_{i}^{m_i^{\star} } (s_i)$
is  non-decreasing in   $s_i$, as highlighted in \eqref{eq:delta_order}.
%%%
Next, we move on to showing that $f$ is submodular.
%%%

%%%
Consider an arbitrary pair of subsets $\mathcal{S},\mathcal{T} \subseteq \mathcal{U}$.
%%%
We consider the two following cases:
%%%
\begin{enumerate}
\item $s_1 \geq t_1$ and $s_2 \geq t_2$: In this case, we have
%%%% 
\begin{align}
\label{eq:submod_proof_1}
f (\mathcal{S} \cup \mathcal{T} ) + f (\mathcal{S} \cap \mathcal{T} )  & = f\big(s_1, s_2 \big) + f (\mathcal{S} \cap \mathcal{T} )  \\
\label{eq:submod_proof_2}
& \leq f\big(s_1, s_2 \big) + f\big(t_1, t_2 \big)  \\
& =  f (\mathcal{S} ) + f (\mathcal{T} ).
\end{align}
%%%%
In the above, \eqref{eq:submod_proof_1} holds since here users with maximal indices in $\mathcal{S} \cup \mathcal{T} $  are also in $\mathcal{S}$.
%%%%
On the other hand, the inequality  in \eqref{eq:submod_proof_2} holds as users with maximal indices in  $\mathcal{S} \cap \mathcal{T} $ must also be in $\mathcal{T}$, and the fact that $f\big(l_1, l_2 \big)  \leq f\big(t_1, t_2 \big)$, for all $(l_1,1),(l_2,2) \in \mathcal{T}$.
%%%%
\item $s_1 \geq t_1$ and $s_2 \leq t_2$: 
%%%%
We may express $f (\mathcal{S}  ) + f ( \mathcal{T} )$  as 
%%%% 
\begin{align}
f (\mathcal{S}  ) + f ( \mathcal{T} )  & = f\big(s_1, s_2 \big) + f (t_1, t_2  )  \\
\label{eq:f_submod_1}
& =  \min_{m_1,n_2\in \langle 2 \rangle } \big\{ \delta_{1}^{m_1}(s_1) + \delta_{2}^{n_2}(s_2) \big\} + 
\min_{n_1,m_2\in \langle 2 \rangle } \big\{ \delta_{1}^{n_1}(t_1) + \delta_{2}^{m_2}(t_2) \big\} \\
& =  \delta_{1}^{m_1^{\star}}(s_1)  + \delta_{2}^{n_2^{\star}}(s_2) + \delta_{1}^{n_1^{\star}}(t_1) + \delta_{2}^{m_2^{\star}}(t_2)
\end{align}
%%%%
where $m_1^{\star}$, $m_2^{\star}$, $n_1^{\star}$ and $n_2^{\star}$ are the coefficients that attain the minimum in \eqref{eq:f_submod_1}.
%%%
On the other hand, in this case the sum $f (\mathcal{S} \cup \mathcal{T} ) + f (\mathcal{S} \cap \mathcal{T} )$ is bounded above as 
%%%% 
\begin{align}
f (\mathcal{S} \cup \mathcal{T} ) + f (\mathcal{S} \cap \mathcal{T} )  & = f\big(s_1, t_2 \big) + f (\mathcal{S} \cap \mathcal{T} )  \\
\label{eq:f_submod_2}
& \leq   f\big(s_1, t_2 \big) +   f\big(\min\{s_1,t_1\}, \min \{ s_2, t_2 \} \big)  \\
& =   f\big(s_1, t_2 \big) +   f\big(t_1, s_2 \big)  \\
\label{eq:f_submod_3}
& =  \min_{m_1,m_2\in \langle 2 \rangle } \big\{ \delta_{1}^{m_1}(s_1) + \delta_{2}^{m_2}(t_2) \big\} + 
\min_{n_1,n_2\in \langle 2 \rangle } \big\{ \delta_{1}^{n_1}(t_1) + \delta_{2}^{n_2}(s_2) \big\} \\
\label{eq:f_submod_4}
& \leq  \delta_{1}^{m_1^{\star}}(s_1) + \delta_{2}^{m_2^{\star}}(t_2)  + \delta_{1}^{n_1^{\star}}(t_1) + \delta_{2}^{n_2^{\star}}(s_2) \\
& = f (\mathcal{S}  ) + f ( \mathcal{T} ).
\end{align}
%%%%
The inequality in \eqref{eq:f_submod_2} holds as the maximal indices of users in $\mathcal{S} \cap \mathcal{T} $ are at most $\min\{s_1,t_1\}$ and  $\min \{ s_2, t_2 \}$.
%%%%
On the other hand, the inequality in \eqref{eq:f_submod_4} holds since $m_1^{\star}$, $m_2^{\star}$, $n_1^{\star}$ and $n_2^{\star}$, the minimizers in \eqref{eq:f_submod_1},
are not necessarily minimizers for \eqref{eq:f_submod_3}.
%%%
\end{enumerate}
%%%
The remaining two cases, i.e. $t_1 \geq s_1$ and $t_2 \leq s_2$, and 
$t_1 \geq s_1$ and $t_2 \leq s_2$, 
can be addressed in a similar manner by swapping indices.
%%%
This proves that $f$ is submodular, and hence completes the proof of polymatroidality for the region described in \eqref{eq:PTIN_polymatroid}.
%%%
%%%%%
%%%%%
\section{Optimality of mc-SLS Under Homogeneous Interference}
\label{appendix:H_ICI}
%%%%%
Here we consider $K \times KL$ networks in which all inter-cell interference links have the same strength level, i.e. networks with homogeneous inter-cell interference (H-ICI).
%%%%
In particular, we have
%%%%%
\begin{equation}
\alpha_{ij}^{[l_i]} =  \beta,  \; \forall  l_i \in \langle L \rangle, i,j \in \langle K \rangle,  i\neq j.
\end{equation}
%%%%%
The set of all $K \times KL$ networks $\bm{\alpha}$ that satisfy this property is denoted by $\mathcal{A}^{\mathrm{H\text{-}ICI}}$.
%%%%
From the order in \eqref{eq:strength_order} and Definition \ref{def:SLS_regime}, it follows that 
the  mc-SLS regime in this setting is simply specified by
%%%%%
\begin{equation}
\label{eq:SLS_cond_sym}
\alpha_{ii}^{[1]} \geq \beta, \forall i \in \langle K \rangle.
\end{equation}
%%%%%
That is, all direct link strength levels must be no less than the inter-cell interference strength level. 
%%%%%
We have the following mc-SLS optimality result for this class of networks.
%%%%%
\begin{theorem}
\label{theorem:h_ICI_SLS}
In the mc-SLS regime, mc-SLS is GDoF optimal for the $K \times KL$ MISO-BC with homogeneous inter-cell interference under finite precision CSIT. Moreover, the GDoF region in this case is equal to the outer bound $\mathcal{D}^{\mathrm{SLS}}_{\mathrm{out}}$. That is,
$\bm{\alpha} \in \big\{  \mathcal{A}^{\mathrm{H\text{-}ICI}} \cap \mathcal{A}^{\mathrm{SLS}}  \big\} \implies 
\mathcal{D}^{\mathrm{MBC}} =  \mathcal{D}^{\mathrm{SLS}}  = 
\mathcal{D}^{\mathrm{SLS}}_{\mathrm{out}}$.
%%%
\end{theorem}
%%%%%
\subsection{Proof of Theorem \ref{theorem:h_ICI_SLS}}
%%%%%
We simplify the notation and write direct link strength parameters $\alpha_{ii}^{[l_i]}$ as 
$\alpha_{i}^{[l_i]}$. 
%%%%%
Due to the SNR order and the mc-SLS regime, we have
$\beta \leq \alpha_i^{[1]}   \leq \alpha_i^{[2]}  \leq  \cdots \leq \alpha_i^{[L]} $ in each cell $i$.
%%%%

%%%%
Next, we show that under full multi-cell cooperation,  the outer bound $\mathcal{D}_{\mathrm{out}}^{\mathrm{SLS}}$ in Theorem \ref{theorem:MISO_BC_outerbound} is achievable using the mc-SLS scheme.
%%%%%
To this end, we first observe that $\mathcal{D}_{\mathrm{out}}^{\mathrm{SLS}}$ in this case simplifies to the region described by all GDoF tuples $\mathbf{d} \in \mathbb{R}_{+}^{KL}$ that satisfy
%%%%%
\begin{align}
\label{eq:coop_out_SLS_sym}
\sum_{i \in \mathcal{K}}  \bar{d}_{i}^{[l_i]}    &  \leq  (1 - |\mathcal{K}| )\beta + \sum_{i \in \mathcal{K}} \alpha^{[l_i]}_i,  \  \forall l_i \in \langle L \rangle, \mathcal{K} \subseteq \langle K \rangle.
\end{align}
%%%%
We show that \eqref{eq:coop_out_SLS_sym} is achievable by employing the simplified mc-SLS scheme described in Section \ref{sec:2_cell_SLS}. 
%%%%
In particular, a superposition of $L$ single-cell codewords is transmitted in each cell by the designated BS in that cell; 
%%%%
and a common codeword is cooperatively transmitted by all BSs on top.   
%%%%
We further set the power control variable $a$ to $\beta$, i.e. each single-cell transmission has a power that scales at most as $P^{-\beta}$ (see Section \ref{subsec:optimality_of_2_cell_SLS}).
%%%%
It follows that the GDoF contributed by the common codeword is given by the region 
$\mathcal{D}^{\mathrm{Mul}}(\beta)$, described by all tuples 
$\mathbf{d}_{\mathrm{c}} \in \mathbb{R}_{+}^{KL}$ that satisfy
%%%
\begin{equation}
\label{eq:mult_region_SLS_sym}
\sum_{(l_i,i) \in \mathcal{U}} d_{\mathrm{c},i}^{[l_i]} \leq \beta.
\end{equation}
%%%
Recall that the common codeword is decoded (and then removed) by all UEs, and hence it does not interfere with single-cell codewords.
%%%
Since cross link strengths are all equal to $\beta$, by setting $a = \beta$, 
inter-cell interference experienced by single-cell transmissions is all received at noise level, rendering it inconsequential for the GDoF achieved by these transmissions.
%%%%
The GDoF region $\mathcal{D}^{\mathrm{PTIN}}(\beta)$ achieved by single-cell codewords is hence given by all tuples 
$\mathbf{d}_{\mathrm{s}} \in \mathbb{R}_{+}^{KL}$ that satisfy
%%%%
\begin{align}
\label{eq:mc_TIN_region_SLS_sym}
\bar{d}_{\mathrm{s},i}^{[l_i]}   & \leq \alpha^{[l_i]}_i - \beta, \ \forall (l_i,i) \in \mathcal{U}.
\end{align}
%%%%
\eqref{eq:mc_TIN_region_SLS_sym} is obtained by observing that single-cell transmissions see a network where direct links have strengths $\alpha^{[l_i]}_i - \beta$, for all $(l_i,i) \in \mathcal{U}$, while cross links have strength zero. 
%%%%
This network is mc-TIN optimal, and by specializing the mc-TIN achievable region in \eqref{eq:TIN_region_CTIN_regime}, we obtain \eqref{eq:mc_TIN_region_SLS_sym}.
%%%%

%%%%
From the above, it follows that the region given by the Minkowski sum
$\mathcal{D}^{\mathrm{Mul}}(\beta) \oplus \mathcal{D}^{\mathrm{PTIN}}(\beta)$
is achievable.
%%%%
It remains to show that this region coincides with the outer bound in \eqref{eq:coop_out_SLS_sym}.
%%%%
We show this by invoking the result on the Minkowski sums of polymatroids, used in Appendix \ref{appendix:Minkowski_sum}.
%%%%
First, however, we express \eqref{eq:mult_region_SLS_sym} and \eqref{eq:mc_TIN_region_SLS_sym} in the standard polymatroid form, with a sum GDoF inequality for each subset of $\mathcal{U}$.
%%%%
To this end, we note that by including redundant inequalities, $\mathcal{D}^{\mathrm{Mul}}(\beta)$ in \eqref{eq:mult_region_SLS_sym}  is equivalently characterized by the set of inequalities
%%%
\begin{equation}
\label{eq:mult_region_SLS_sym_polymat}
\sum_{i \in \mathcal{K}} \sum_{ (l_i,i)  \in \mathcal{S}_i }  d_{\mathrm{c},i}^{[l_i]} \leq \beta, \ \forall \mathcal{S}_i \subseteq \mathcal{U}_i, \mathcal{K} \subseteq \langle K \rangle
\end{equation}
%%%
where $\mathcal{U}_i \triangleq \{ (l_i,i): l_i \in \langle L \rangle \}$, for all $i \in \langle K \rangle$.
%%%%
Similarly, it can be verified that $\mathcal{D}^{\mathrm{PTIN}}(\beta)$ in \eqref{eq:mc_TIN_region_SLS_sym}  is equivalently characterized by the set of inequalities
%%%
\begin{align}
\label{eq:mc_TIN_region_SLS_sym_polymat}
\sum_{i \in \mathcal{K}} \sum_{ (l_i,i)  \in \mathcal{S}_i }  d_{\mathrm{s},i}^{[l_i]}    &  \leq \sum_{i \in \mathcal{K}}  (\alpha^{[s_i]}_i - \beta), \ \forall \mathcal{S}_i \subseteq \mathcal{U}_i, \mathcal{K} \subseteq \langle K \rangle
\end{align}
%%%%
where $s_i \triangleq \max_{(l_i,i) \in \mathcal{S}_i} l_i$. This holds since all additional inequalities in \eqref{eq:mc_TIN_region_SLS_sym_polymat} are redundant with respect to the inequalities in \eqref{eq:mc_TIN_region_SLS_sym}.
%%%%
The region in \eqref{eq:mult_region_SLS_sym_polymat} is clearly a polymatroid. 
%%%
Moreover, in part \ref{appendix:subsec_polymatroid_homog} at the end of this appendix, we show that the region in \eqref{eq:mc_TIN_region_SLS_sym_polymat} is also a polymatroid.
%%%%

%%%%
We now invoke \cite[Th. 44.6]{Schrijver2002} and characterize the Minkowski sum $\mathcal{D}^{\mathrm{Mul}} \oplus \mathcal{D}^{\mathrm{PTIN}}$  by directly summing the corresponding inequalities in \eqref{eq:mult_region_SLS_sym_polymat} and 
\eqref{eq:mc_TIN_region_SLS_sym_polymat},  yielding
%%%%
\begin{align}
\sum_{i \in \mathcal{K}} \sum_{ (l_i,i)  \in \mathcal{S}_i }  d_{i}^{[l_i]}    &  \leq \beta + \sum_{i \in \mathcal{K}}  (\alpha^{[s_i]}_i - \beta), \ \forall \mathcal{S}_i \subseteq \mathcal{U}_i, \mathcal{K} \subseteq \langle K \rangle.
\end{align}
%%%%
After removing redundant inequalities, we are left with the following inequalities
%%%%
\begin{align}
\label{eq:coop_in_SLS_sym}
\sum_{i \in \mathcal{K}}  \bar{d}_{i}^{[l_i]}    &  \leq \beta  + \sum_{i \in \mathcal{K}} ( \alpha^{[l_i]}_i - \beta), \ \forall l_i \in \langle L \rangle, \mathcal{K} \subseteq \langle K \rangle.
\end{align}
%%%%
From \eqref{eq:coop_in_SLS_sym} and \eqref{eq:coop_out_SLS_sym}, it follows that   $\mathcal{D}^{\mathrm{Mul}}(\beta) \oplus \mathcal{D}^{\mathrm{PTIN}}(\beta) =  \mathcal{D}_{\mathrm{out}}^{\mathrm{SLS}}$, which completes the proof.
%%%%
\subsection{Polymatroidality of \eqref{eq:mc_TIN_region_SLS_sym_polymat}}
\label{appendix:subsec_polymatroid_homog}
%%%
Here we show that the region in \eqref{eq:mc_TIN_region_SLS_sym_polymat} is a polymatroid.
%%%%
For this, we define the set function
%%%%
\begin{equation}
\label{eq:set_function_f_SLS_sym}
f(\mathcal{S}) \triangleq 
\begin{cases}
0, &  \mathcal{S} =\emptyset \\
\sum_{i \in \mathcal{K}}  (\alpha^{[s_i]}_i - \beta), & \mathcal{S} = \cup_{i\in \mathcal{K}} \mathcal{S}_i,  \mathcal{S}_i \subseteq \mathcal{U}_i,
\mathcal{K} \subseteq \langle K \rangle
\end{cases}
\end{equation}
%%%%%
associated with the polyhedron in \eqref{eq:mc_TIN_region_SLS_sym_polymat}.
%%%%%
The function $f$ is normalized and non-decreasing, 
and hence it remains to show that it is also submodular for \eqref{eq:mc_TIN_region_SLS_sym_polymat} to be a polymatroid 
(see Appendix \ref{appendix:subsec_polymatroid}).
%%%%%
By a well known equivalent definition of submodularity, $f$  is submodular if  
\begin{equation}
\label{eq:submod_cond_equiv}
f \big(\mathcal{S} \cup \{(l_k,k)\} \big)  -  f\big( \mathcal{S} \big)   \geq f \big(\mathcal{S} \cup \{(l_k,k), (l_j,j) \} \big)  - 
f \big( \mathcal{S} \cup \{(l_j,j)\} \big) 
\end{equation}
%%%%
holds for all $\mathcal{S} \subseteq \mathcal{U}$, 
$(l_k,k) \neq (l_j,j) $ and $(l_k,k), (l_j,j)  \in  \mathcal{U} \setminus \mathcal{S}$.
%%%%
Now we examine both sides of the inequality in \eqref{eq:submod_cond_equiv}. 
%%%%
From \eqref{eq:set_function_f_SLS_sym}, and recalling that $s_i \triangleq \max_{(l_i,i) \in \mathcal{S}_i} l_i$, it follows that
%%%%
\begin{equation}
\label{eq:set_function_f_sobmod_SLS_sym_0}
f \big( \mathcal{S} \cup \{(l_k,k)\} \big) - f \big( \mathcal{S} \big) =  \mathbbm{1}( k \notin \mathcal{K}) \big(\alpha^{[l_k]}_k  - \beta \big) + \mathbbm{1}( k \in \mathcal{K}) \mathbbm{1}( l_k > s_k )   \big(\alpha^{[l_k]}_k  - \alpha^{[s_k]}_k \big)
\end{equation}
%%%%%
In a similar fashion, we also have
\begin{multline}
\label{eq:set_function_f_sobmod_SLS_sym}
f \big( \mathcal{S} \cup \{(l_k,k),(l_j,j)\} \big) - f \big( \mathcal{S} \cup \{(l_j,j)\}  \big) =  \mathbbm{1}( k \notin \mathcal{K})  \mathbbm{1}( k \neq j)   \big(\alpha^{[l_k]}_k  - \beta \big)   \\ + \mathbbm{1}( k \notin \mathcal{K}) \mathbbm{1}( k = j)  \mathbbm{1}( l_k > l_j )   \big(\alpha^{[l_k]}_k  - \alpha^{[l_j]}_k \big) + 
\mathbbm{1}( k \in \mathcal{K}) \mathbbm{1}( l_k > s_k )   \big(\alpha^{[l_k]}_k  - \alpha^{[s_k]}_k \big).
\end{multline}
%%%%%
From the mc-SLS condition in \eqref{eq:SLS_cond_sym}, it immediately follows that 
$\mathbbm{1}( l_k > l_j )   \big(\alpha^{[l_k]}_k  - \alpha^{[l_j]}_k \big)  \leq  \big(\alpha^{[l_k]}_k  - \beta \big)  $.
%%%%
Combining this observation with \eqref{eq:set_function_f_sobmod_SLS_sym} and \eqref{eq:set_function_f_sobmod_SLS_sym_0}, we obtain the following upper bound
\begin{align}
\nonumber
& f \big( \mathcal{S} \cup \{(l_k,k),(l_j,j)\} \big) - f \big( \mathcal{S} \cup \{(l_j,j)\}  \big)  \\
%%%%%
\nonumber
& \leq \mathbbm{1}( k \notin \mathcal{K})  \big(  \mathbbm{1}( k \neq j) +\mathbbm{1}( k =  j)   \big) \big(\alpha^{[l_k]}_k  - \beta  \big) 
+ \mathbbm{1}( k \in \mathcal{K}) \mathbbm{1}( l_k > s_k )   \big(\alpha^{[l_k]}_k  - \alpha^{[s_k]}_k \big) \\
%%%%%
\nonumber
&  = f \big( \mathcal{S} \cup \{(l_k,k)\} \big) - f \big( \mathcal{S} \big)
\end{align}
%%%%%
which in turn proves the submodularity of the function $f$.
%%%%
Therefore, \eqref{eq:mc_TIN_region_SLS_sym_polymat} is a polymatroid.
%%%
\section*{Acknowledgements}
The authors would like to thank the Associate Editor and the  anonymous reviewers for their suggestions that helped improve both 
the technical content and presentation of this paper.
%%%
%\newpage
%\footnotesize
\bibliographystyle{IEEEtran}
\bibliography{References}
%%%
\end{document}